%% file: ijcga-universal.tex
\documentclass{ws-ijcga}
\usepackage{graphicx,amssymb,amsmath}
\usepackage{lmodern}

\usepackage{graphicx}
\graphicspath{{./images/}}
\usepackage{epsf}
\usepackage{epsfig}
\usepackage{color}
\usepackage{cite}
\usepackage[titlenumbered, ruled]{algorithm2e}
\usepackage[font=small]{caption}
\usepackage{subfig}

\newcommand{\qedhere}{%
  \begingroup \let\mathqed\math@qedhere
    \let\@elt\setQED@elt \QED@stack\relax\relax \endgroup
}

\newcommand{\uniguard}[1]{\ensuremath{\boldsymbol{u} \left(#1\right)}}
\newcommand{\kuniguard}[2]{\ensuremath{\boldsymbol{u}_{#2} \left(#1\right)}}
\newcommand{\uniguardholes}[2]{\ensuremath{\boldsymbol{h}_{#2} \left(#1\right)}}
\newcommand{\kuniguardholes}[2]{\ensuremath{\boldsymbol{h}_{#2} \left(#1\right)}}
\newcommand{\uniguardinterior}[2]{\ensuremath{\boldsymbol{i}_{#2} \left(#1\right)}}
\newcommand{\uniwatch}[1]{\ensuremath{\textit{w} \left(#1\right)}}
\newcommand{\uniwatchinterior}[1]{\ensuremath{\boldsymbol{i} \left(#1\right)}}
\newcommand{\uniguardgrid}[1]{\ensuremath{\boldsymbol{g} \left( #1 \right)}}
\newcommand{\uniguardshells}[2]{\ensuremath{\boldsymbol{s} \left(#1,#2\right)}}
\newcommand{\points}[1]{\ensuremath{\mathcal{S}\left(#1\right)}}
\newcommand{\gridpoints}[1]{\ensuremath{\mathcal{S}_g\left(#1\right)}}
\newcommand{\pointsshells}[2]{\ensuremath{\mathcal{S}\left(#1,#2\right)}}
\newcommand{\simplepolygons}[1]{\ensuremath{\mathcal{P}\left( #1 \right)}}
\newcommand{\polygons}[1]{\ensuremath{\mathcal{H}\left( #1 \right)}}
\newcommand{\ptriangle}[1]{\ensuremath{\triangle}\left( #1 \right)}

\newcommand{\statement}[2]{\vskip2ex\noindent{\sffamily\bfseries#1~\ref{#2}.}}
\newcommand{\sees}[1]{\ensuremath{\leftrightarrow}_{#1}}
\newcommand{\old}[1]{{}}

\title{Universal Guard Problems}

\author{S\'{a}ndor~P.~Fekete}
\address{Department of Computer Science, TU Braunschweig, 38106 Braunschweig, 
Germany. \texttt{s.fekete@tu-bs.de}}
\author{Qian Li, Joseph~S.~B.~Mitchell}
\address{Department of Applied Mathematics and Statistics, Stony Brook University, Stony Brook, NY 11794, USA
\{qian.li.1, joseph.mitchell\}@stonybrook.edu}
\author{Christian Scheffer}
\address{Department of Computer Science, TU Braunschweig, 38106 Braunschweig, 
Germany. \texttt{c.scheffer@tu-bs.de}}

\index{Fekete, S\'{a}ndor P.}
\index{Li, Qian}
\index{Mitchell, Joseph~S.~B.}
\index{Scheffer, Christian}

\begin{document}
\thispagestyle{empty}
\maketitle

\input{00-abstract.tex}
\input{01-introduction.tex}
\input{02-preliminaries.tex}
\input{03-boundsuniversal.tex}

\input{04-boundskuniversal.tex}
\input{05-other.tex}

\input{06-conclusion.tex}

\subparagraph*{Acknowledgements}
Q. Li and J. Mitchell are partially supported by the National Science Foundation (CCF-1526406).

\bibliographystyle{abbrv}

\bibliography{refs}

\end{document}

%% file: 00-abstract.tex
\begin{abstract}
We provide a spectrum of results for the {\em Universal Guard Problem}, in
which one is to obtain a small set of points (``guards'') that are
``universal'' in their ability to guard any of a set of possible
polygonal domains in the plane. We give upper and lower bounds on the
number of universal guards that are always sufficient to guard all
polygons having a given set of $n$ vertices, or to guard all polygons in a
given set of $k$ polygons on an $n$-point vertex set.  Our upper bound
proofs include algorithms to construct universal guard sets of the
respective cardinalities.
\end{abstract}

%% file: 01-introduction.tex
\section{Introduction}
\label{sec:intro}

Problems of finding optimal covers are among the most fundamental algorithmic
challenges that play an important role in many contexts. One of the 
best-studied prototypes in a geometric setting is the classic Art Gallery Problem (AGP),
which asks for a small number of points (``guards'') required for
covering (``seeing'') all of the points within a geometric domain. 
An enormous body of work on algorithmic aspects of visibility coverage and
related problems (see, e.g., O'Rourke~\cite{o-agta-87},
Keil~\cite{k-pd-00}, and  \cite{s-rrag-92}) was spawned by Klee's question for worst-case bounds
more than 40 years ago: How many guards are always sufficient
to guard all of the points in a simple polygon having $n$ vertices?
The answer, as shown originally by Chv\'atal~\cite{c-ctpg-75}, and with
a very simple and elegant proof by Fisk~\cite{f-spcwt-78}, is that
$\lfloor n/3\rfloor$ guards are always sufficient, and sometimes
necessary, to guard a simple $n$-gon. 

While Klee's question was posed about guarding an $n$-vertex {\em
  simple polygon}, a related question about {\em point sets} was posed
at the 2014 NYU Goodman-Pollack Fest:
Given a set $S$ of $n$ points in the plane, how many {\em universal} guards are sometimes necessary and always
sufficient to guard any simple polygon with vertex set $S$?  
This problem, and several related questions,
are studied in this paper.  We give the first set of results on
universal guarding, including combinatorial bounds and efficient
algorithms to compute universal guard sets that achieve the upper
bounds we prove.  We focus on the case in which guards
must be placed at a subset of the input set $S$ and thus will be
vertex guards for any polygonalization of $S$.

A strong motivation for our study is the problem of computing guard sets in the face of
uncertainty.  In our model, we require that the guards are {\em
  robust} with respect to different possible polygonalizations
consistent with a given set of points (e.g., obtained by scanning an
environment).  Our Universal Guard Problem is, in a sense, an extreme
version of the problem of guarding a set of possible polygonalizations
that are consistent with a given set of sample points that are the
polygon vertices: In the universal setting, we require that the guards
are a rich enough set to achieve visibility coverage for {\em all}
possible polygonalizations.  Another variant studied here is the {\em
  $k$-universal} guarding problem in which the guards must perform
visibility coverage for a set of $k$ different polygonalizations of
the input points.  Further, in the full version of the paper, we study
the case in which guards are required to be placed at non-convex hull
points of $S$, or at points of a regular rectangular grid.

\subsection*{Related Work}
In addition to the worst-case results for the AGP,
related work includes algorithmic results for computing a
minimum-cardinality guard set.  The problem of computing an optimal
guard set is known to be NP-hard~\cite{o-agta-87}, even in very basic
settings such as guarding a 1.5D
terrain~\cite{DBLP:journals/siamcomp/KingK11}.
Ghosh~\cite{ghosh-87,ghosh2010approximation} observed that greedy set
cover yields an $O(\log n)$-approximation for guarding with the fewest
vertices.  Using techniques of Clarkson~\cite{clarkson1993algorithms}
and Br\"onnimann-Goodrich~\cite{bg-aoscf-95}, $O(\log
OPT)$-approximation algorithms were given, if guards are restricted to
vertices or points of a discrete
grid~\cite{DBLP:journals/ipl/EfratH06,ehm-aatol-05,Hector-art-gallery}.
For the special case of {\em rectangle visibility} in rectilinear
polygons, an exact optimization algorithm is known~\cite{wk-pdoag-06}.
Recently, for vertex guards (or discrete guards on the boundary) in a
simple polygon $P$, King and Kirkpatrick~\cite{king2011improved}
obtained an $O(\log\log OPT)$-approximation, by building
$\epsilon$-nets of size $O((1/\epsilon)loglog (1/\epsilon))$ for the
associated hitting set instances, and applying~\cite{bg-aoscf-95}.
For the special case of guarding 1.5D terrains, local search yields a
PTAS~\cite{krohn2014guarding}. 
Experiments based on heuristics for computing upper and lower bounds on guard
numbers have been shown to perform very well in practice~\cite{amp-lgvcp-10}.
Methods of combinatorial optimization with insights and algorithms from
computational geometry have been successfully combined for the Art Gallery
Problem, leading to provably optimal guard sets for instances of 
significant size~\cite{bdd-pgpcs-13,crs-exmvg-11,DBLP:journals/jea/KrollerBFS12,DaviPedroCid-OO2013,ffk+-fagp-15}.

The notion of ``universality'' has been studied in other contexts in
combinatorial
optimization~\cite{jia2005universal,hajiaghayi2006improved}, including
the traveling salesman problem (TSP), Steiner trees, and set cover.
For example, in the universal TSP, one desires a single ``master''
tour on all input points so that, for {\em any} subset $S$ of the
input points, the tour obtained by visiting $S$ in the order
specified by the master tour yields a tour that approximates an
optimal tour on the subset.

\subsection*{Our Results}

We introduce a family of universal coverage problems for the classic Art Gallery Problems. We provide
a spectrum of lower and upper bounds for the required numbers of guards. 
See Table~\ref{table:overviewgeneral} and \ref{table:overviewkuniversal} for a detailed overview, and
the following Section~\ref{sec:prelim} for involved notation.

%% file: 02-preliminaries.tex
\section{Preliminaries}
\label{sec:prelim}


	For $n \in \mathbb{N}$, let \points{n} be the set of all discrete point
sets in the plane that have cardinality $n$. 
{A single \emph{shell} of a point set $S$ is the subset of points of $S$ on the boundary of the convex hull of $S$. 
Recursively, for $k \geq 2$, a point set lies on $k$ shells, if removing the points 
on its convex hull, leaves a set that lies on $k-1$ shells.} We denote by $\gridpoints{n} \subset
\points{n}$ and $\pointsshells{n}{m} \subset \points{n}$ the set of all
discrete point sets that form a rectangular $a \times b$-grid of $n$ points for
$a,b,a \cdot b = n \in \mathbb{N}, $ and the set of all discrete point sets
that lie on $m$ shells for $m \in \mathbb{N}$, respectively.
	
	For $S \in \points{n}$, let $\simplepolygons{S}$ (resp., $\polygons{S}$) be the set of all simple polygons (resp., polygons with holes) whose vertex set equals~$S$.
	
	Let $P$ be a polygon. We say a point $p \in P$ \emph{sees} (w.r.t. $P$) another point $q \in P$ if $pq \subset P$; we then write $p \sees{P} q$. The \emph{visible region (w.r.t. $P$)} of a point $g \in P$ is $V_P(g) = \{a \in P:  g \sees{P} a\}$. A point set $G \subseteq S$ is a \emph{guard set} for $P$ if $\bigcup_{g \in G} V_P(g) = P$. Furthermore, we say that $G$ is an \emph{interior guard set for $P$} if $G$ is a guard set for $P$ and no $g \in G$ is a vertex of the convex hull of $P$.
	
	For a set $A$ of polygons we say that $G \subseteq S$ is a(n) (interior) guard set of $A$ if $G$ is a(n) \emph{(interior) guard set} for each $P \in A$. We denote by $\uniwatch{A}$ the minimum cardinality guard set for $A$ and by $\uniwatchinterior{A}$ the minimum cardinality interior guard set for $A$. Furthermore, for any given point set $S$ we say that $G \subseteq S$ is a \emph{guard set for $S$} if $G$ is a guard set for $\simplepolygons{S}$. 
 For $k,m,n \in \mathbb{N}$, the guard numbers are listed in Table~\ref{table:overviewumbers}.

 \begin{table*}[h!]
\centering
\begin{tabular}{|l|c|l|}
                \hline
                \emph{universal guards} & $\uniguard{n}$&$ \max_{S \in \points{n}} \uniwatch{\mathcal{P}(S)}$\\
                \hline
                \emph{$m$-shelled universal guards} & $\uniguardshells{n}{m}$&$ \max_{S \in \pointsshells{n}{m}} \uniwatch{\simplepolygons{S}}$\\
                \hline
                \emph{interior universal guards} & $\uniguardinterior{n}{}$& $\max_{S \in \points{n}} \uniwatchinterior{\simplepolygons{S}}$\\
                \hline
                \emph{$k$-universal guards, simple polygons} & $\kuniguard{n}{k}$& $\max_{S \in \mathcal{S}(n)} \max_{\substack{A \subseteq \simplepolygons{S})\\ \textit{s.t. } |A|=k}} \uniwatch{A}$\\
                \hline
                \emph{$k$-universal guards, polygons with holes} & $\uniguardholes{n}{k}$ &$\max_{S \in \mathcal{S}(n)} \max_{\substack{A \subseteq \polygons{S}\\ \textit{s.t. } |A|=k}} \uniwatch{A}$\\
                \hline
                \emph{grid universal guards} & $\uniguardgrid{n}$&$ \max_{S \in \mathcal{S}_g(n)} \uniwatch{\simplepolygons{S}}$\\
                \hline
\end{tabular}
\vspace*{1mm}
\caption[Our universal guard numbers]{The universal guard numbers considered in this paper.}
\label{table:overviewumbers}
\end{table*}

\begin{table*}[h!]
\centering
\begin{tabular}{|c||c|c|c|c|c|c|c|}
                \hline
                  $m,n \in \mathbb{N}$& \uniguard{n}{} &\uniguardshells{n}{m} & \uniguardgrid{n}{} & \uniguardinterior{n}{}  \\
                  \hline
                  \hline
                  $\begin{array}{c} \text{lower} \\ \text{bounds} \end{array}$ & $\left(1 - \Theta \left( \frac{1}{\sqrt{n}} \right) \right)n$ & $ \left( 1- \frac{1}{2(m-1)} - \frac{8m}{n(m-1)} \right)n$ & $\lfloor \frac{n}{2} \rfloor$ & $n - \mathcal{O}(1)$\\
                \hline
                $\begin{array}{c} \text{upper} \\ \text{bounds} \end{array}$ & $\left( 1 - \Theta \left( \frac{1}{n} \right) \right)n$ &$\left( 1 - \frac{1}{16 n^{\left( 1 - \frac{1}{2m} \right)}} \right)n$& $\lfloor \frac{n}{2} \rfloor$ & $n- \Omega(1)$ \\
                \hline
\end{tabular}
\vspace*{1mm}
\caption[Result overview for general universal guard numbers]{Results for simple polygons. The approaches for the upper bounds for $\uniguard{n}{}$ and $\uniguardshells{n}{m}$ also apply to polygons with holes, yielding the same upper bounds.}
\label{table:overviewgeneral}
\end{table*}

\begin{table*}[h!]
\centering
\begin{tabular}{|c||c|c|c|c|c|c|c|c|c|}
                \hline
                 $n \in \mathbb{N}$& \kuniguard{n}{2} & \kuniguard{n}{3} & \kuniguard{n}{4} & \kuniguard{n}{5}& $\begin{array}{c} \kuniguard{n}{k} \\ \text{for } k \geq 6 \end{array}$  & $\begin{array}{c} \kuniguardholes{n}{k} \\ \text{for } k \in \mathbb{N} \end{array}$\\
                  \hline
                  \hline
                  $\begin{array}{c} \text{lower} \\ \text{bounds} \end{array}$  & $\lfloor \frac{3n}{8} \rfloor$ & $\frac{4n}{9}$ & $\frac{n}{2} - \mathcal{O}(\sqrt{n})$&$\frac{n}{2} - \mathcal{O}(\sqrt{n})$&$\frac{5n}{9}$& $\frac{5n}{9}$ \\
                \hline
                $\begin{array}{c} \text{upper} \\ \text{bounds} \end{array}$ & $\frac{5n}{9}$ & $\frac{19n}{27}$ & $\frac{65n}{81}$ &$\frac{211n}{243}$ & $(1-(\frac{2}{3})^k)n$ & $( 1 - (\frac{5}{8})^k )n$\\
                \hline
\end{tabular}
\vspace*{1mm}
\caption[Result overview for $k$-universal guard numbers]{Overview of our results for $k$-universal guard numbers of simple polygons and of polygons with holes. We give a new corresponding approach for the upper bounds of $\kuniguardholes{n}{1},\kuniguardholes{n}{2}, \dots$. We also consider the lower bounds for $\kuniguard{n}{1},\kuniguard{n}{2}, \dots$ as lower bounds for $\kuniguardholes{n}{1},\kuniguardholes{n}{2}, \dots$.}
\label{table:overviewkuniversal}
\end{table*}

%% file: 03-boundsuniversal.tex
\section{Bounds for Universal Guard Numbers}\label{sec:boundsuniguard}
	
	 In the following, we provide different lower and upper bounds for the universal guard numbers. In particular, the provided bounds can be classified by the number of shells on which the points of the considered point set are located. 
	 
\subsection{Lower Bounds for Universal Guard Numbers}\label{sec:lowerboundsuniguard}

	In this section we give lower bounds for the universal guard numbers $\uniguard{n}{}$ and~$\uniguardshells{n}{m}$ for $n \in \mathbb{N}$ and $m \geq 2$. In particular, we provide lower bound constructions that can be described by the following approach: For any given $n \in \mathbb{N}$ and $m \geq 2$, we construct a point set $S_m \in \points{n}$ as follows. $S_m$ is partitioned into pairwise disjoint subsets $B_1, \dots, B_{m}$, such that $\bigcup_{i=1}^{m} B_i = S$. For $i \in \{ 1,...,m \}$, each $B_i$ lies on a circle~$C_i$ such that $C_i$ is enclosed by $C_{i+1}$ for $i \in \{1,...,m-1 \}$. Furthermore, $C_1,\dots,C_m$ are concentric and have ``sufficiently large'' radii; see Sections~\ref{sec:m2}, \ref{sec:m3}, and~\ref{sec:mk} for details. In particular, the radii depend on the approaches that are applied for the different cases $m=2$, $m=3$, and $m\geq 4$. We place four equidistant points on $C_m$. The remaining points are placed on $C_{m-1},\dots,C_1$.
	
	Note that $\uniguardshells{n}{1}=1$ holds, because for every convex point set $S \in \points{n}$, $\simplepolygons{S}$ consists of only the boundary of the convex hull of $S$. 
Thus we start with the case of $m=2$.
	
\subsubsection{Lower Bounds for $\uniguardshells{n}{2}$}\label{sec:m2}

	We give an approach that provides a lower bound for $\uniguardshells{n}{2}$. In particular, for any $n \in \mathbb{N}$, we construct a point set $S_2 \in \points{n}$ having $n-4$ equally spaced points lie on circle $C_1$ and $4$ equally spaced points on a larger concentric circle $C_2$, such that these $4$ points form a square containing $C_1$; see Figure~\ref{fig:lowerboundtwoshells}. In order to assure that the constructed subsets of $S_2$ and $S_3,S_4,\dots$ (which are described later) are nonempty, we require $n \geq 32$ for the rest of Section~\ref{sec:lowerboundsuniguard}.
	
	Let $v$ be a point from the square and let $p,q$ be two consecutive points from the circle $C_1$, such that the segments $vp$ and $vq$ do not intersect the interior of the circle $C_1$; see Figure~\ref{fig:lowerboundtwoshells}(a). We choose the side lengths of the square such that the cone $c$ that is induced by $p$ and $q$ with apex at $v$ contains at most $\frac{n}{8}$ points from $C_1$ for all choices of $v$, $p$, and $q$.

\begin{lemma}\label{lem:twoshells}
	Let $G$ be a guard set of $S_2$. Then we have $|G| > \frac{n}{2} - 4$.
\end{lemma}
\begin{proof}
	Suppose $|G| \leq \lfloor \frac{n-4}{2} \rfloor - 1$. This implies that there are two points $p,q \in S_m \setminus G$ such that $p$ and $q$ lie adjacent on $C_1$; see Figure~\ref{fig:lowerboundtwoshells}(b). Let $w_1$, $w_2$, $w_3$, and $w_4$ be the four points from the square. At most two points $v_1,v_2 \in \{ w_1,w_2,w_3,w_4 \}$ span a cone, such that $v_1p,v_1q,v_2p,v_2q$ do not intersect the interior of $C_1$. Without loss of generality, we assume {that these two different cones $c_1$ and $c_2$ exist.} $c_1$ and $c_2$ contain at most $\frac{n}{4}$ points from $C$. Thus, there is another point $w \in S_2 \setminus G$ such that $v \notin c_1 \cup c_2$. This implies that there is a polygon in which $w$ is not seen by a guard from $G${; see Figure~\ref{fig:lowerboundtwoshells}(b)}. This is a contradiction to the assumption that $G$ is a guard set.
	
	Thus we have $|G| > \lfloor \frac{n-4}{2} \rfloor - 1 \geq \frac{n-4}{2} -2 = \frac{n}{2} - 4$. This concludes the proof.
	
	\end{proof}
	
\begin{figure}[ht]
  \begin{center}
    \begin{tabular}{ccc}
      \includegraphics[height=2.5cm]{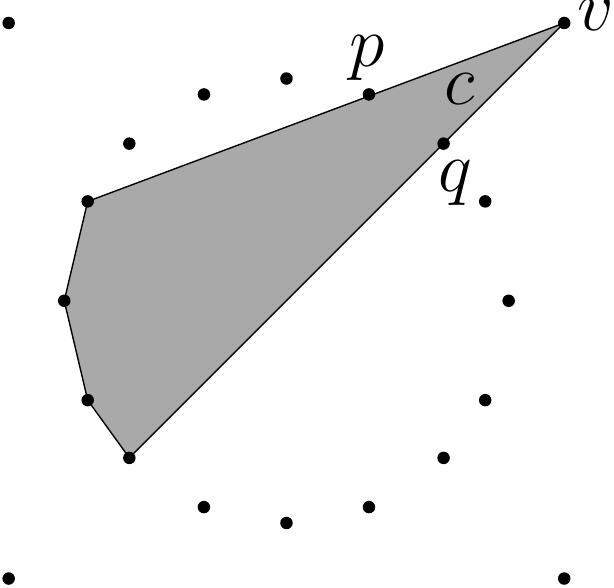} &&
       \includegraphics[height=2.5cm]{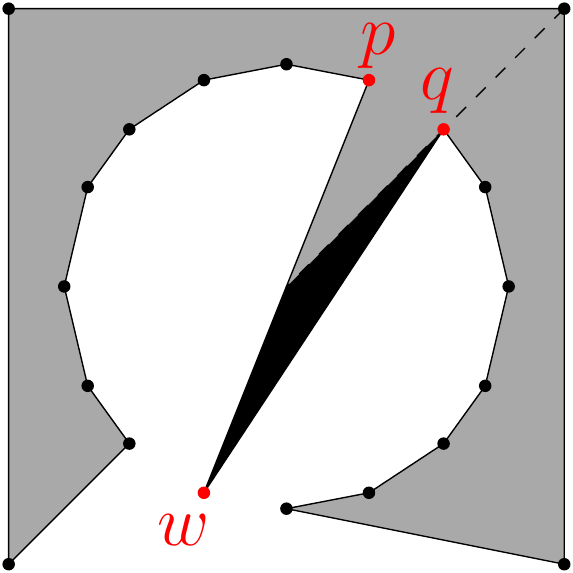}\\
      {\small (a) } &&      
      {\small (b) }
    \end{tabular}
  \end{center}
  \vspace*{-6pt}
   \caption{Lower-bound construction for $\uniguardshells{n}{2}$.}
   \label{fig:lowerboundtwoshells}
\end{figure}

\begin{corollary}
	$\uniguardshells{n}{2} \geq \lfloor \frac{n}{2} \rfloor - 4$
\end{corollary}
	
\subsubsection{A First Lower Bound for $\uniguardshells{n}{3}$}\label{sec:m3}

{The high-level idea is to guarantee in the construction of $S_3$ that at most two points on $C_1$ are unguarded; see Figure~\ref{fig:lowerboundconstructionm3} for the idea of the proof of contradiction. By constructing $S_3 = B_1 \cup B_2 \cup B_3$ such that $|B_1| = \lfloor \frac{n-4}{2} \rfloor$, $|B_2| = \lceil \frac{n-4}{2} \rceil$, and $|B_3| = 4$, we obtain $|G| \geq \frac{n}{2} - 5$ for any guard set $G$ of $S_3$.}
	
	\begin{figure}[ht]
  \begin{center}
    \begin{tabular}{ccc}
      \includegraphics[height=3.5cm]{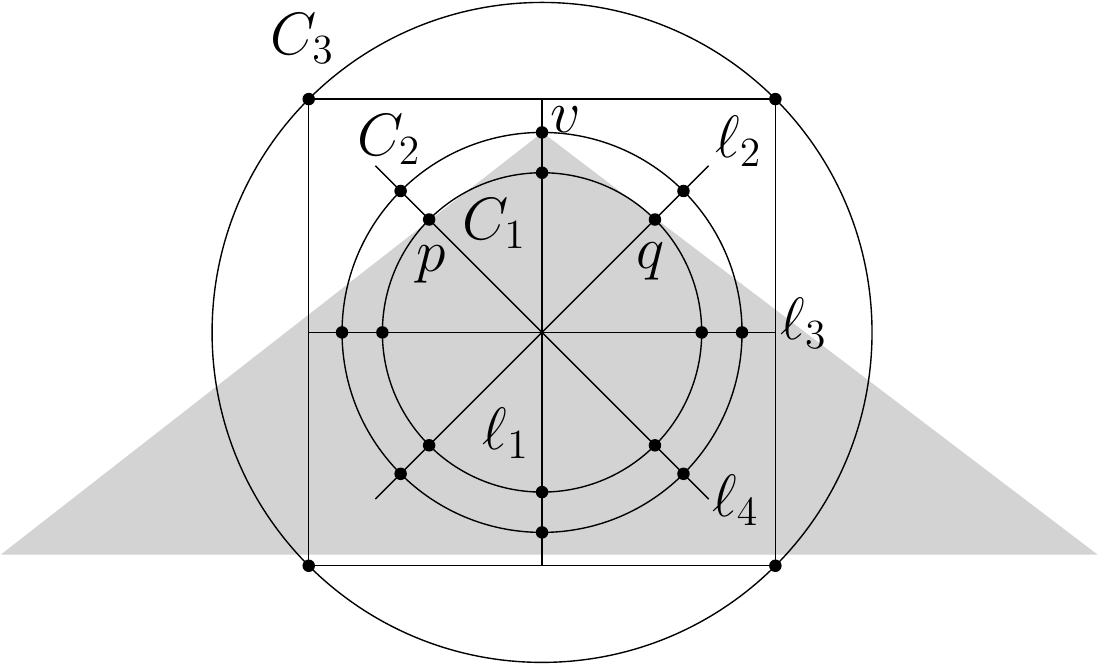} &&
       \includegraphics[height=3.5cm]{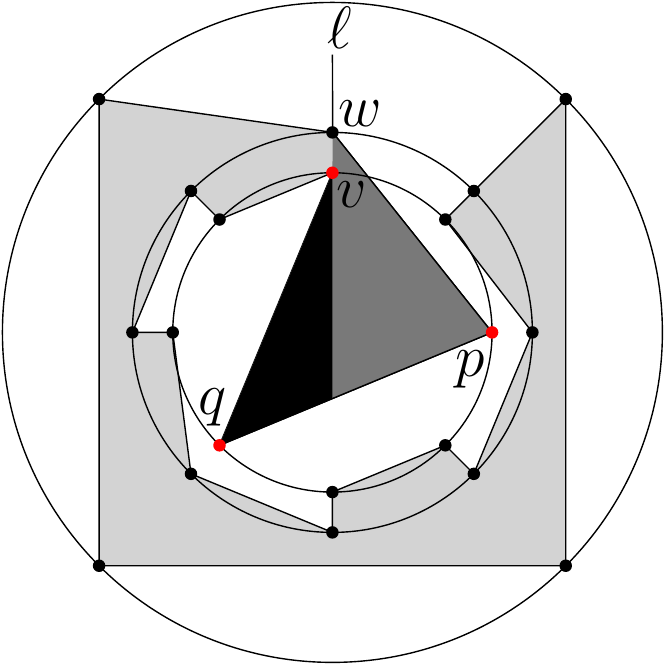}\\
      {\small (a) Lower-bound construction for $\uniguardshells{n}{3}$.} &&      
      {\small (b) An empty {chamber} $\ptriangle{w,p,q,v}$.}
    \end{tabular}
  \end{center}
  \vspace*{-12pt}
  \caption{The lower-bound construction for $\uniguardshells{n}{3}$.}
  \label{fig:lowerboundconstructionm3}
\vspace*{-2mm}
\end{figure}
	
	We consider the lower-bound construction $S_m$ for $m-1=2$ and $n = (m-1) 2^{l}+4 = 3\cdot2^l+4$ for any $l \geq 4$, i.e., for all $S_3 \in \points{2\cdot2^l+4}$ for any $l \geq 2$. The argument can easily be extended to $n \in \mathbb{N}$. 
	
	The points of $B_2$ and $B_3$ are placed on $C_2$ and $C_3$, such that they lie on $2^{l-1}$ lines; see Figure~\ref{fig:lowerboundconstructionm3}(a). Let $v \in B_2$ be chosen arbitrarily and $p,q \in B_1$ such that $p$ and~$q$ are the neighbors of the point from $B_1$ that corresponds to $v \in B_2$.  We choose the radius of $C_2$ such that the cone that is induced by~$p$ and~$q$ and with apex at~$v$ contains all points from~$B_1$; see the gray cone in Figure~\ref{fig:lowerboundconstructionm3}(a). Furthermore, we choose the radius of $C_1$ such that the square that is induced by the four points from~$B_1$ contains all points from $B_1 \cup B_2$.

	The key construction that we apply in the proofs of our lower bounds are \emph{{chambers}}.

\begin{definition}\label{def:pseudotriangle}
	Let $S$ be an arbitrary discrete point set in the plane. Four points $p_1,p_2,p_3,p_4 \in S$ form a \emph{{chamber}, denoted $\ptriangle{p_1,p_2,p_3,p_4}$}, if
	\begin{itemize}
		\item (1) $p_1$ and $p_2$ lie on different sides of the line $p_3p_4$,
		\item (2) $p_3$ and $p_4$ lie on the same side of the line $p_1p_2$, and
		\item (3) there is no point from $S$ that lies inside the polygon that is bounded by the polygonal chain $\langle p_1,p_2,p_3,p_4 \rangle$.
	\end{itemize} 
	
	Let $G \subseteq S$. We say that $\ptriangle{p_1,p_2,p_3,p_4}$ is \emph{empty} (with respect to $G$) if $p_2,p_3,p_4 \notin G$. Let $P \in \simplepolygons{S}$. We say that $\ptriangle{p_1,p_2,p_3,p_4}$ \emph{is part of $P$} if $p_1p_2, p_2p_3,p_3p_4 \subset \partial P$.
\end{definition}

	Our proofs are based on the following simple observation.

\begin{observation}\label{obs:noemptypseudtriangles}
	Let $G$ be a guard set for a polygon $P$. There is no empty {chamber} that is part of $P$.
\end{observation}

	Based on Observation~\ref{obs:noemptypseudtriangles} we prove the following lemma, which we then apply to the construction above to obtain our lower bounds for $\uniguardshells{n}{m}$.

\begin{lemma}\label{lem:lowerboundm3}
	Let $G$ be a guard set for $\simplepolygons{S_3}$. Then we have $|B_1 \setminus G|\leq 2$.
\end{lemma}
\begin{proof} Suppose there are three points $v,q,p \in B_1 \setminus G$. Without loss of generality, we assume that $q$ and $p$ lie on different sides with respect to the line $\ell$ that corresponds to the placement of $v$; see Figure~\ref{fig:lowerboundconstructionm3}(b). Furthermore, we denote the point from $B_2$ that lies above $v$ by $w$. By construction it follows that $w$, $p$, $q$, and $v$ form an empty {chamber} $\ptriangle{w,p,q,v}$. Furthermore, we construct a polygon $P \in \simplepolygons{S_3}$ such that $\ptriangle{w,p,q,v}$ is part of $P$; see Figure~\ref{fig:lowerboundconstructionm3}(b). By Observation~\ref{obs:noemptypseudtriangles} it follows that $G$ is not a guard set for $P$, a contradiction. This concludes the proof.
\end{proof}

	There is a corresponding construction for all other values $n \in \mathbb{N}$. In particular, we place four points equidistant on $C_3$, $\lceil \frac{n-4}{2} \rceil$ equidistant points on $C_2$, and $\lfloor \frac{n-4}{2} \rfloor$ points on $C_1$, such that each point from $C_1$ lies below a point from $C_2$. The same argument as above applies to the resulting construction of a point set. 
	The constructions of $S_m$ can be modified so that no three points lie on the same line, by a slight perturbation. Thus, $S_3$ can be assumed to be in general position.
	We obtain the following corollary.

\begin{corollary}\label{cor:lowerboundm3}
	$\uniguardshells{n}{3} \geq \frac{n}{2} - 5$.
\end{corollary}
\begin{proof}
	Lemma~\ref{lem:lowerboundm3} implies that at least $\lfloor \frac{n-4}{2} \rfloor-2$ points from $B_1$ are guarded. Let $G$ be an arbitrarily chosen guard set for $\simplepolygons{S_3}$. Thus we obtain $|G| \geq \lfloor \frac{n-4}{2} \rfloor-2 \geq \frac{n-4}{2} - 3 = \frac{n}{2} - 5$.
\end{proof} In the following section we generalize the above approach from the case of three shells to the case of $m$ shells and combine that argument with the approach that we applied for the case of $m=2$. This also leads to the improved lower bound $\kuniguard{n}{3} \geq (\frac{3}{4} - \mathcal{O}(\frac{1}{n}))n$.
	
\subsubsection{(Improved) Lower Bounds for $\uniguard{n}{}$ and $\uniguardshells{n}{m}$ for $m \geq 3$}\label{sec:mk}
	
	In this section we give general constructions $S_3,S_4, \dots$ of the point sets that yield our lower bounds for $\uniguardshells{n}{m}$ for $m \geq 3$. The main difference in the construction of $S_m$ for $m \geq 3$, compared to the previous section, is the choice of the radii of $C_1,...,C_m$. Similar as in the previous section, we guarantee that on each circle $C_3,C_4,\dots$ at most $\mathcal{O}(1)$ points are unguarded. The general idea is to choose five arbitrary points $q_1,q_2,q_3,q_4,q_5$ on $C_i$ for $i \in \{ 3,4,\dots \}$. There are three points $u_1,u_2,u_3 \in \{ q_1,q_2,q_3,q_4,q_5 \}$, such that the triangle induced by $u_1,u_2,u_3$ does not contain the common mid point of $C_1,C_2, \dots$. By choosing the radius of $C_{i+1}$ sufficiently large, we obtain that there is a {chamber}~$\ptriangle{u_1,u_2,u_3,p}$, where~$p$ is a point on $C_{i+1}$; see Figure~\ref{fig:construction}. This implies that $\ptriangle{u_1,u_2,u_3,p}$ is empty if $q_1,q_2,q_3,q_4,q_5$ are unguarded. Thus, at most four points on $C_i$ are allowed to be unguarded; see Corollary~\ref{cor:atmostfourunguarded}.
	
	Finally, we show how the arguments for $S_m$ yield lower bounds for $\uniguardshells{n}{m}$ and~$\uniguard{n}$.
	
	Similar to the approach of the previous section, the constructed point sets $S_3, S_4, \dots$ can be modified to be in general position.\\

	\textbf{The Construction of $S_m$ for $m \geq 3$:} We construct $S_m$ such that $|B_1|=\dots=|B_{m-1}|= 2^l$, $|B_m|=4$, and hence $n = (m-1)2^l+4$ for $l \geq 4$. In particular, similar as for the construction of $S_3$ from the previous section, we place the points of $B_1,\dots,B_{m-1}$ equidistant on the circles $C_1,\dots,C_{m-1}$, such that the points lie on $2^{l-1}$ lines $\ell_1,\dots,\ell_{2^{l-1}}$; see Figure~\ref{fig:construction}(a). 
	
	In order to apply an argument that makes use of {chambers}, we need the following notation of points on a circle $C_i$. Let $n':=2^l$. Let $v_1,..., v_{1+n'/2}$ be the points on~$C_i$ to one side or on $\ell \in \{ \ell_1,...,\ell_{n'/2}\}$. Let $w_1, ..., w_{1+n'/2}$ be their reflection across $\ell$; see Figure~\ref{fig:construction}(b)+(c). 
Let $v \in C_{i+1}$ be the point that lies above $v_{1+n'/4}$. As the construction of $S_m$ is symmetric with respect to rotations the following discussion applies to each choice of $\ell$ and $v$ such that $v$ and the midpoint of the circles $C_1,\dots,C_m$ lie orthogonal to $\ell$.

	For $i \in \{1,\dots,m-1 \}$, we choose the radius of $C_{i+1}$ compared to the radius of~$C_{i}$ sufficiently large, such that $v$, two points $v_j$ and $w_j$ that lie orthogonal to $\ell$, and a fourth point $p$ from $C_i$ build a chamber~$\ptriangle{v,w_j,p,v_j}$; see Figure~\ref{fig:construction}(b). Simultaneously, we ensure that $v$, $p$, $w_j$, and $v_{j-1}$ build another chamber~$\ptriangle{v,p,w_j,v_{j-1}}$; see Figure~\ref{fig:construction}(c). 
	
	 In particular, we have to choose the radius of $C_{i+1}$ large enough such that the polygons bounded by the polygonal chains $\langle v,w_j,p,v_j \rangle$ and $\langle v,p,w_j,v_{j-1} \rangle$ do not contain any other points from $S$. In order to do this, we ensure that (1) the segment $vw_i$ intersects $C_i$ in the arc between $v_j$ and $v_{j+1}$; see Figure~\ref{fig:construction}(a) and (2) the segment $vw_j$ intersects $C_i$ in the arc between $v_{j-1}$ and $v_{j-2}$; see Figure~\ref{fig:construction}(b).
	
	
	Finally, we place the four points $w_1,w_2,w_3,w_4 \in B_m$ such that all circles lie in the convex hull of $w_1$, $w_2$, $w_3$, and $w_4$; see Figure~\ref{fig:construction}(a).\\
	
\begin{figure}[]
  \begin{center}
    \begin{tabular}{ccccc}
      \includegraphics[height=3cm]{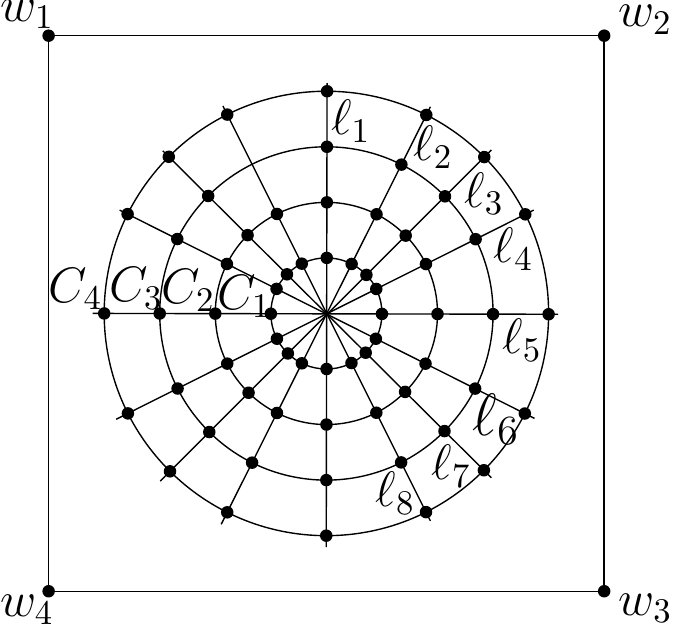} &&
       \includegraphics[height=4cm]{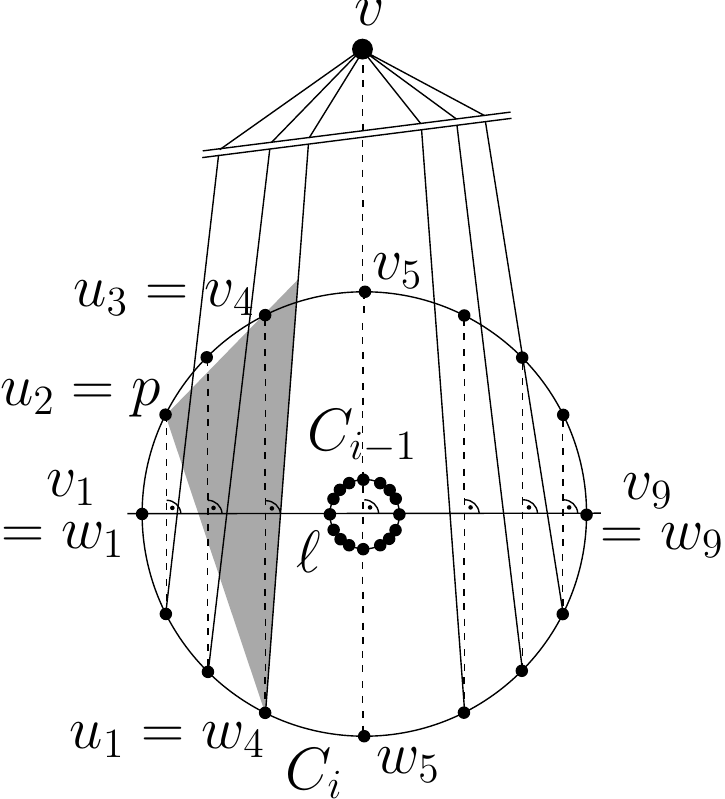}&&
       \includegraphics[height=4cm]{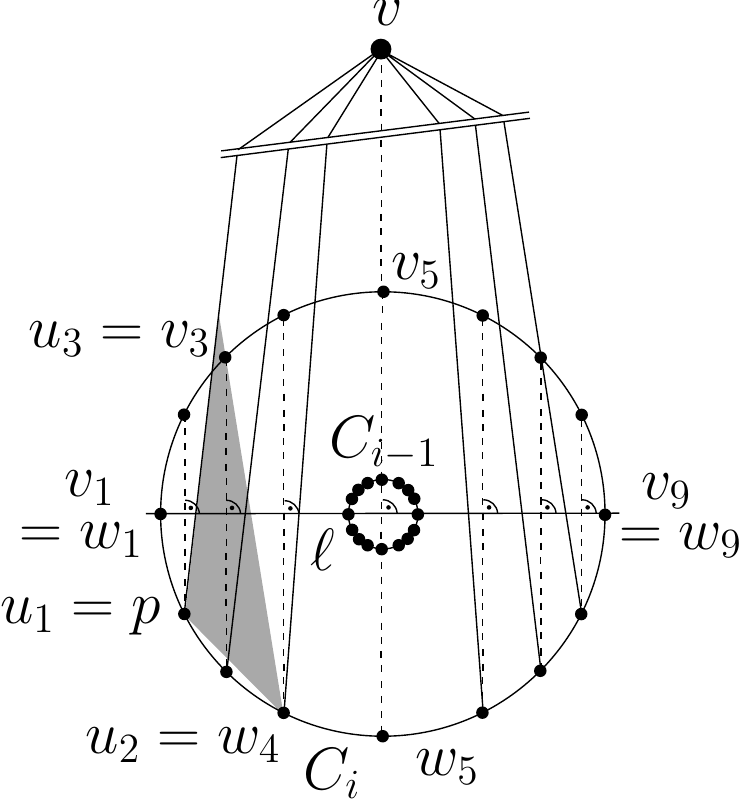}\\
      {\small (a) Construction of the } &&      
      {\small (b) Segments between $v$} &&
      {\small (c) Segments between $v$}\\
      {\small circles $C_1,...,C_m$.} &&      
      {\small and vertices from} &&
      {\small and vertices from}\\
      {\small } &&      
      {\small the opposite side of $C_i$.} &&
      {\small the opposite side of $C_i$.}
    \end{tabular}
  \end{center}
  \vspace*{-12pt}
  \caption{Construction of $S_m$ for $n=68$. For a simplified illustration we changed the ratios of the circles' radii and we shortened the lines adjacent to $v$. In the configuration of Lemma~\ref{lem:lowerboundmkpseudotriangle}, three points from $C_i$ in the same half of $C_i$ imply a {chamber} with a point $v \in C_{i+1}$ that lies above $\ell$. Chambers with a point $w \in C_{i+1}$ can be constructed symmetrically with respect to the line $\ell$.
}
  \label{fig:construction}
\vspace*{-6mm}
\end{figure}

\textbf{The Analysis of $S_m$ for $m\geq 3$:} First we show that we can choose three points $u_1,u_2,u_3$ from five arbitrarily chosen points from $C_i$, such that there is another point $u \in C_{i+1}$ with $\ptriangle{u,u_1,u_2,u_3}$ being a {chamber}; see Lemma~\ref{lem:lowerboundmkpseudotriangle}. Next, we construct a polygon $P \in \simplepolygons{S_m}$, such that $\ptriangle{u,u_1,u_2,u_3}$ is a part of $P$; see Lemma~\ref{lem:pseudtrianglepart}. Finally, by combining Lemma~\ref{lem:lowerboundmkpseudotriangle} and Lemma~\ref{lem:pseudtrianglepart} we establish that on each $C_i$, at most four points are allowed to be unguarded; see Corollary~\ref{cor:atmostfourunguarded}. This leads to several lower bounds for $\uniguardshells{n}{m}$ and $\uniguard{n}{}$.
	
	\begin{lemma}\label{lem:lowerboundmkpseudotriangle}
		Let $q_1,q_2,q_3,q_4,q_5 \in A_i$ be chosen arbitrarily. There are three points $u_1,u_2,u_3 \in \{ q_1,q_2,q_3,q_4,q_5 \}$ and a point $u \in A_{i+1}$, such that $\ptriangle{u,u_1,u_2,u_3}$ is a {chamber}.
	\end{lemma}

        \begin{proof}
                We choose $u_1,u_2,u_3$ from $\{ q_1,q_2,q_3,q_4,q_5 \}$, such that $u_1,u_2,u_3$ lie in the same half of $C_i$, i.e., such that the midpoint of $C_i$ does not lie inside the triangle $t$ that is induced by $u_1,u_2,u_3$, see Figure~\ref{fig:construction}(b)+(c). Without loss of generality, we assume that $u_2$ lies between $u_1$ and $u_3$. Otherwise, we rename the points appropriately.

                We distinguish two cases. (C1) The number of points between
$u_1$ and $u_3$ is odd and (C2) the number of points between $u_1$ and $u_3$ is
even. For both cases (C1) and (C2) we can ensure the existence of a corresponding {chamber} for achieving the required
contradiction; see Figure~\ref{fig:construction}(b) for even (C1) and Figure~\ref{fig:construction}(c) for odd (C2).\\

        \end{proof}

	Based on Lemma~\ref{lem:lowerboundmkpseudotriangle}, we can construct the required polygon $P$ such that the chamber constructed in Lemma~\ref{lem:lowerboundmkpseudotriangle} is part of $P$.
	
	\begin{lemma}\label{lem:pseudtrianglepart}
		There is a polygon $P \in \simplepolygons{S_m}$ such that $\ptriangle{u,u_1,u_2,u_3}$ is part of $P$.
	\end{lemma}
	\begin{proof} We construct $P$ for the cases (C1) and (C2) of Lemma~\ref{lem:lowerboundmkpseudotriangle} separately; see Figure~\ref{fig:constructionPolygon}. In both cases we walk upwards on the line $\ell \in \{ \ell_1,\dots,\ell_{n'/2} \}$ until we reach $C_1$. Next we orbit $C_i$ in a zig-zag approach and finally connect all points from $C_{i-1},\dots,C_1$ in a similar manner; see Figure~\ref{fig:constructionPolygon}.
	\end{proof}

\begin{figure}[ht]
  \begin{center}
    \begin{tabular}{ccc}
      \includegraphics[height=5cm]{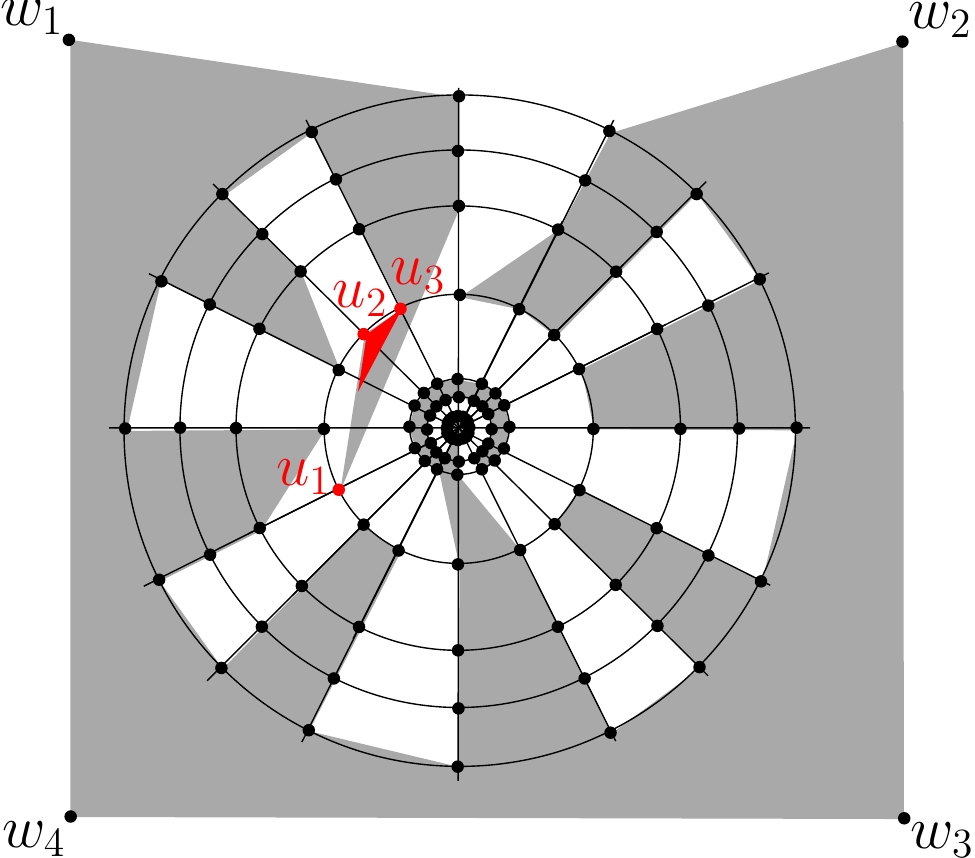} &&
       \includegraphics[height=5cm]{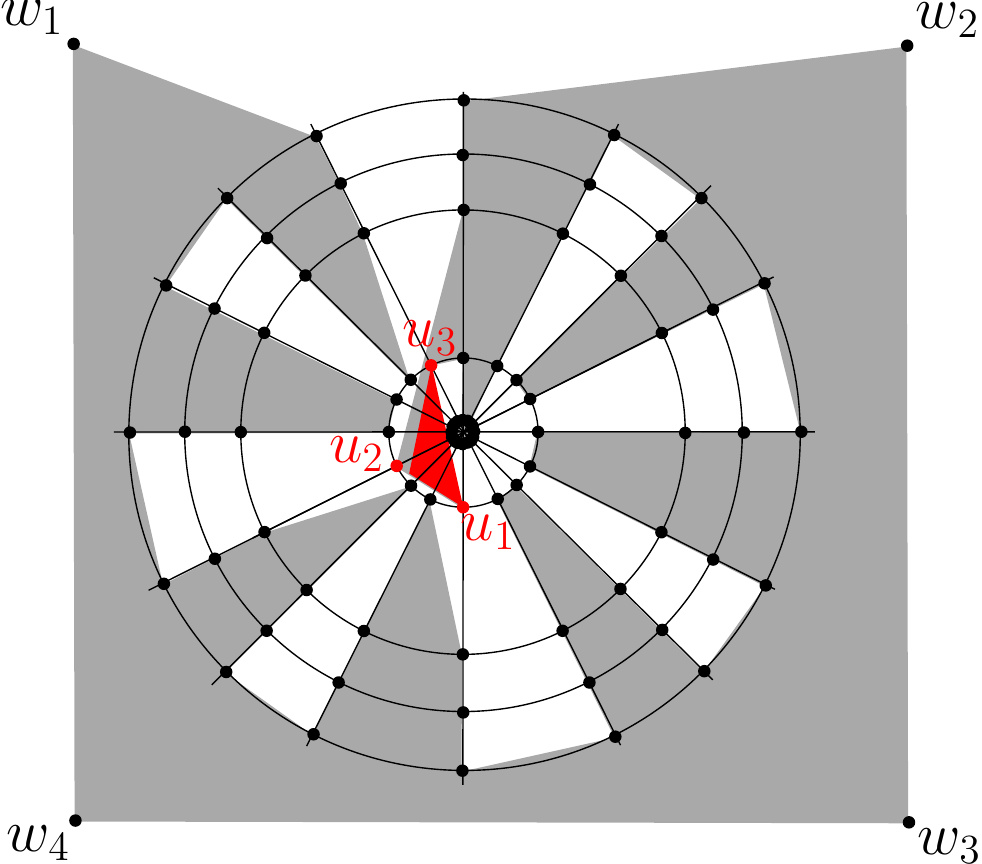}\\
      {\small (a) The case in which the number} &&      
      {\small (b) The case in which the number}\\
      {\small of points between $u_1$ and $u_2$ is odd.}&&
      {\small of points between $u_1$ and $u_2$ is even.} 
    \end{tabular}
  \end{center}
  \vspace*{-12pt}
  \caption{Construction of $\mathcal{P}$ for $k= 6$ and $n = 16$. For a simplified illustration we changed the ratios of the circles' radii (otherwise the figure would become too large).}
  \label{fig:constructionPolygon}
\vspace*{-3mm}
\end{figure}

	The combination of Lemma~\ref{lem:lowerboundmkpseudotriangle} and Lemma~\ref{lem:pseudtrianglepart} implies the following corollary.
	
	\begin{corollary}\label{cor:atmostfourunguarded}
		Let $G \subset S_m$ be a guard set of $\simplepolygons{S_m}$. Then $|B_i \setminus G| \leq 4$, for $i \in \{ 1,...,m-2 \}$.
	\end{corollary}
	
\textbf{Lower bounds for $\uniguardshells{n}{m}$ and $\uniguard{n}{}$ that are implied by Corollary~\ref{cor:atmostfourunguarded}:} We combine the approach for $\uniguardshells{n}{2}$ with Corollary~\ref{cor:atmostfourunguarded}, which yields the following lower bound for $\uniguardshells{n}{m}$ for $m \geq 3$.

\begin{corollary}\label{constantkImproved}
	Let $m \geq 3$ and $n' = 2^l$ with $l \geq 4$. Furthermore, let $G \subseteq S_m$ be a guard set of $S_m$. Then we have $|G| \geq \left( 1 - \frac{1}{2(m-1)} + \frac{8m}{n(m-1)} \right) |S_m|$.
\end{corollary}
\begin{proof}
	By Corollary~\ref{cor:atmostfourunguarded} it follows that $(m-2)(n'-4)$ points from $B_1 \cup \dots \cup B_{m-2}$ are guarded where $n' = |B_1| = \dots |B_{m-2}|$. Furthermore, by applying the approach of Lemma~\ref{lem:twoshells} to $B_{m-1}$ and $B_m$ yields that at least $\frac{n'}{2} - 4$ points from $B_{m-1} \cup B_m$ are guarded because $n' = |B_{m-1}|$. Thus we obtain $|G|  \geq (m-2)(n'-4) + \frac{n'}{2} - 4$, which is lower-bounded by $|S_m| \left(1 - \frac{1}{2(m-1)} - \frac{8m}{|S_m|(m-1)} \right)$ because $n' = \frac{|S_m| - 4}{m-1}$.
\end{proof}
	
\begin{theorem}\label{thm:improvedlowerbounduniguardshells}
	$\uniguardshells{n}{m} \geq n \left( 1 - \frac{1}{2(m-1)} + \frac{8m}{n(m-1)} \right)$ for $m \geq 3$.
\end{theorem}

	By choosing $m$ appropriately, we obtain the following lower bound:

\begin{lemma}\label{lem:lowerBound}
	For any $c < 1$ and any guard set $G$ for $S_m$ there is an $m \in \mathbb{N}$ with $|G| > c |S_m|$. 
\end{lemma}
\begin{proof}
	The approach is to choose $m := \lceil \frac{2n'}{n'-4-cn'} \rceil$, which will imply  $|G| > c |S_m|$. 
	
	Suppose $|G| \leq c |S_m|$. This leads to a contradiction as follows. We have $|S_m|=4+(m-1)n'$. Corollary~\ref{cor:atmostfourunguarded} implies implies that on $C_1,...,C_{m-2}$ there are at most four vertices that are unguarded. Thus, $(m-2)(n'-4) \leq |G|$. By assumption we know $|G| \leq c (4+(m-1)n')$. Thus, we obtain $(m-2)(n'-4) \leq c(4+(m-1)n')$, which implies that $8 \leq 4$ holds because $m = \lceil \frac{2n'}{n'-4-cn'} \rceil$.\\
\end{proof}

	By choosing $c$ appropriately, Lemma~\ref{lem:lowerBound} leads to our general upper bound for~$\uniguard{n}$.

\begin{theorem}\label{thm:ext}
	There is an $m \in \mathbb{N}$ such that $|G| > \left(1 - \frac{10}{\sqrt{|S_m|}}\right)|S_m|$ holds for any guard set $G$ for $\simplepolygons{S_m}$.
\end{theorem}
\begin{proof}
	Lemma~\ref{lem:lowerBound} implies that at least $(1 - \frac{5}{n'})|S_m|$ points are guarded for $c := (1 - \frac{5}{n'})$. Note that we chose  $m := \lceil \frac{2n'}{n'-4-cn'} \rceil$ in the proof of Lemma~\ref{lem:lowerBound}. Furthermore, we have $|S_m| = 4 + (m-1)n'$. This implies $m \leq \frac{2n'}{n'-4- \left( 1-\frac{5}{n'} \right)n'} +1 = 2n'+1$. Additionally, by combining $m := \lceil \frac{2n'}{n'-4-cn'} \rceil$ and $|S_m| = 4 + (m-1)n'$, we obtain $|S_m| \leq 4 + 2(n')^2$, which implies that $\sqrt{|S_m|/2} - \sqrt{2} \leq n'$. As  least $(1 - \frac{5}{n'})|S_m|$ points are guarded, we get $|G|\geq \left(1 - \frac{5\sqrt{2}}{\sqrt{|S_m|}-2}\right)|S_m|>\left( 1- \frac{10}{\sqrt{|S_m|}} \right) |S_m|$ as required.
	
\end{proof}
	
\begin{theorem}\label{thm:lowerbounduniguard}
	$\uniguard{n}{} \geq \left( 1- \frac{10}{\sqrt{n}} \right)n$.
\end{theorem}

\subsection{Upper Bounds for Universal Guard Numbers}

	In the following we give an approach to computing a non-trivial guard set of a given point set. The number of the computed guards depends on the number $m$ of shells of the considered point set $S$. This approach yields upper bounds for~$\uniguardshells{n}{m}$ for~$m \geq 2$.

	For the case of $m=1$, a na\"{i}ve approach is simply to select one arbitrarily chosen guard from $S$. In that case, $\simplepolygons{S}$ just consists of the polygon that corresponds to the boundary of the convex hull of $S$ and an arbitrarily chosen point from $S$ sees all points from all polygons of $\simplepolygons{S}$. 
	
	In the following, we first give an approach for the case of $m=2$. Then, we generalize that approach to the case of $m \geq 3$.
	
\subsubsection{Upper Bound for $\uniguardshells{n}{2}$}\label{sec:upperbound2m}

First we describe the approach, followed by showing that the computed point set~$G$ is a guard set. This leads to an upper bound for $|G|$, which implies the required upper bound for $\uniguardshells{n}{2}$.

\begin{figure}[ht]
  \begin{center}
    \begin{tabular}{ccccc}
      \includegraphics[height=2cm]{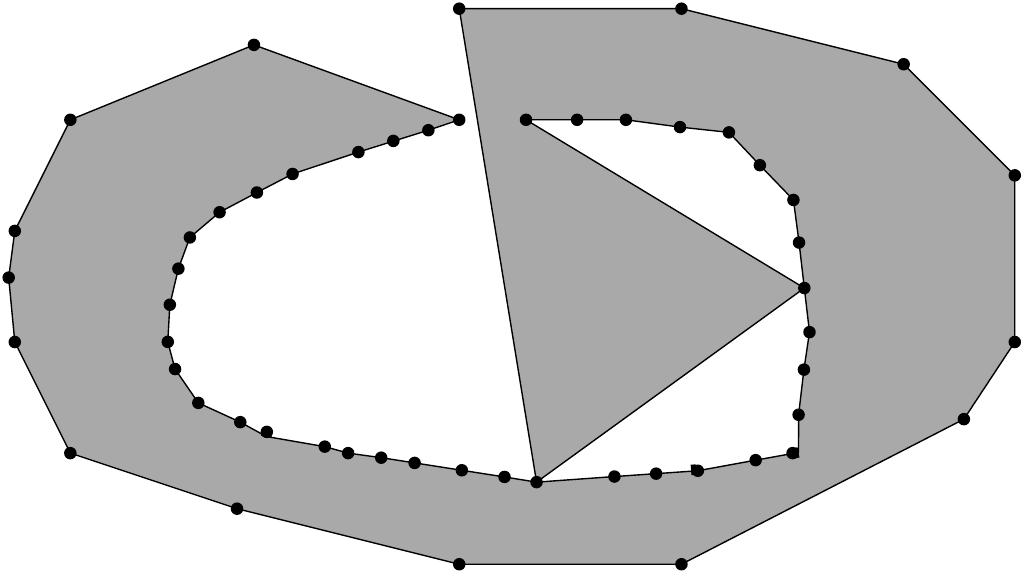} &&
       \includegraphics[height=2cm]{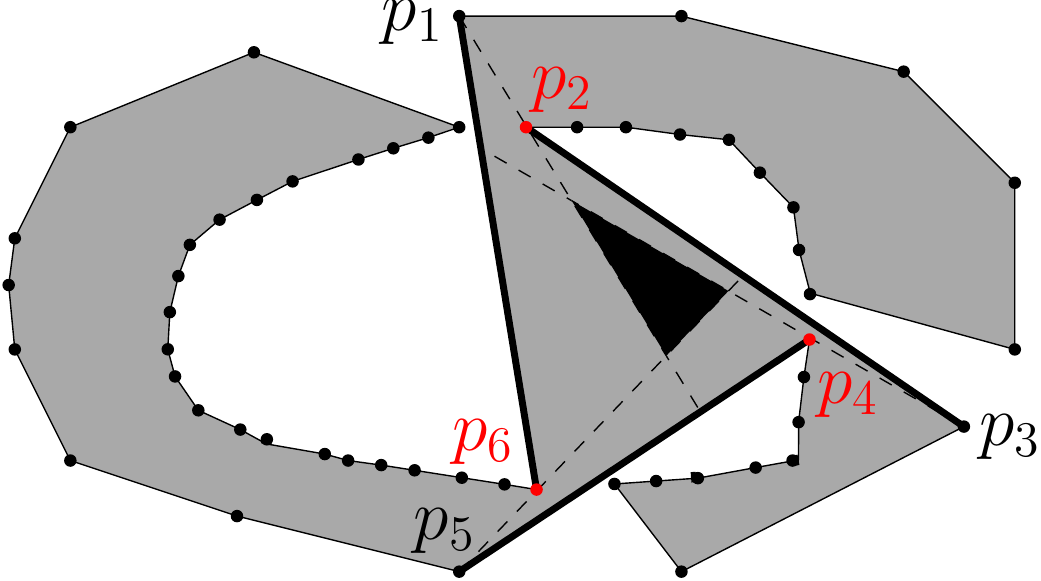}&&
       \includegraphics[height =2cm]{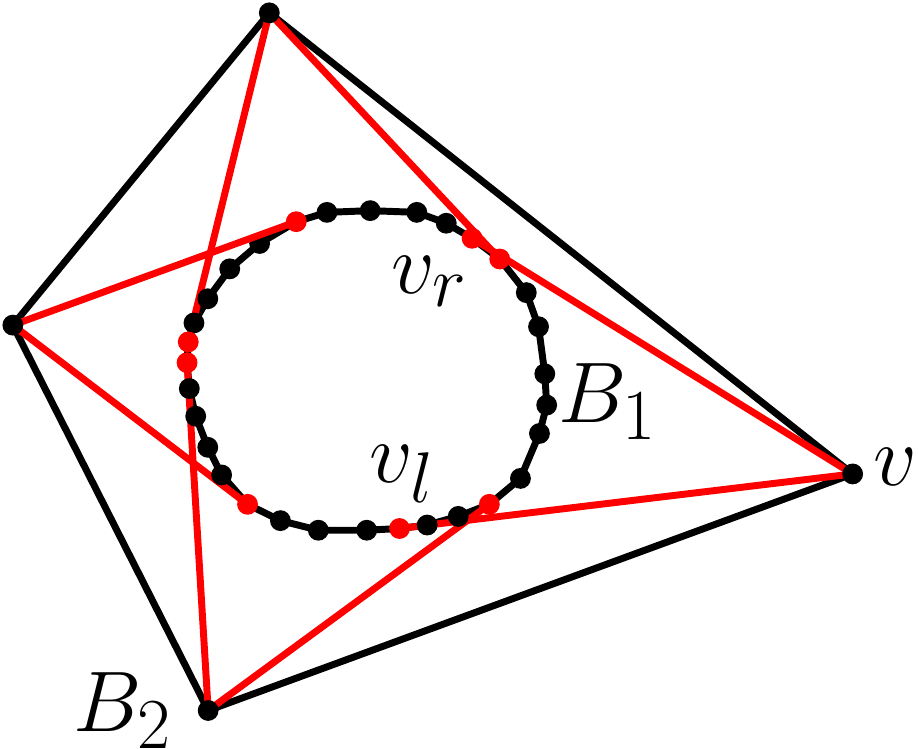}\\
      {\small (a) {chambers}} &&      
      {\small (b) For $m=2$ a similar}&&
      {\small (c) Avoiding {chambers}}\\
      {\small are not part of a}&&
      {\small structure may be part of a } &&
      {\small and similar structures by}\\
      {\small simple polygon if $m=2$.}&&
      {\small simple polygon.} &&
      {\small tangent points.}

    \end{tabular}
  \end{center}
  \vspace*{-12pt}
  \caption{Possible chambers in case of two shells and how we avoid them.}
  \label{fig:nopseudotr}
\end{figure}

\noindent\textbf{The upper bound approach for two shells:} The high-level idea 
is to avoid areas that are unguarded by structures similar to chambers. In particular, in the case of $m=2$, a chamber cannot be part of a simple polygon; otherwise, the boundary of~$P$ meets points at least twice, see Figure~\ref{fig:nopseudotr}(a). However, there is another structure that {has an effect similar to that of} {chambers} and that also may cause unguarded areas, see Figure~\ref{fig:nopseudotr}(b). In the example of Figure~\ref{fig:nopseudotr}(b), our approach guarantees that $p_2$ or $p_6$, $p_2$ or $p_4$, and $p_4$ or~$p_6$ is guarded.

	More generally, for a point $p$ on the outer shell, a point $q$ on the inner shell is a \emph{tangent point} of $p$ if all points from the inner shell lie on the same side with respect to the line induced by $p$ and $q$. Each point on the outer shell has two tangent points on the inner shell. In our approach we guarantee that two unguarded points on the inner shell are not separated by tangent points corresponding to a point from the outer shell, see Figure~\ref{fig:nopseudotr}(c).

	Our approach makes a case distinction as follows: Let $B_1 \subset S$ be the points on the inner shell and $B_2 \subset S$ be the points on the outer shell of the input point set $S$. If $|B_2| \geq \sqrt{|B_1|}/2$ we take $B_1$ as the guard set $G$. Otherwise, we compute for each $v \in B_2$ the two corresponding tangent points $v_l$ and $v_r$ on $B_1$, see Figure~\ref{fig:nopseudotr}(c). Next, we compute a longest sequence~$\langle v_1,\dots,v_k  \rangle$ of points from the inner shell such that $\langle v_1,\dots,v_k  \rangle$ does not contain any tangent points. Finally, we fix every second point from $\langle v_1,\dots,v_k  \rangle$ as unguarded and choose all other points from $S$ as guarded.
	
	\noindent\textbf{Analysis of the approach for two shells:} For the constructed point set $G$, we can guarantee that $G$ is a guard set for $\simplepolygons{S}$ with $|G| \leq (1 - \frac{1}{\sqrt{6 |S|}}) |S|$:
			
\begin{theorem}\label{lem:upperBoundTwoshells}
	For each point set $S$ that lies on two convex hulls, we can compute in $\mathcal{O}(|S| \log |S|)$ time a guard set $G$ with $|G| \leq (1 - \frac{1}{\sqrt{6 |S|}}) |S|$.
\end{theorem}

For the proof of Theorem~\ref{lem:upperBoundTwoshells}, we first show $|G| \leq (1 - \frac{1}{\sqrt{6 |S|}}) |S|$, see Lemma~\ref{lem:twoShellsUpperBoundSize} followed by showing that $G$ is a guard set for
$\simplepolygons{S}$, see the partition of $P$ described below and Lemma~\ref{lem:upperBoundTwoShellsTrianglesAdj}.

\begin{lemma}\label{lem:twoShellsUpperBoundSize}
	$|G| \leq \left(1 - \frac{1}{\sqrt{6|S|}}\right)|S|$.
\end{lemma}
\begin{proof}
	For simplified presentation we denote $n_1:=|B_1|$ and $n_2 := |B_2|$. We consider the two cases $n_2 \geq \frac{\sqrt{n_1}}{2}$ and $n_2 < \frac{\sqrt{n_1}}{2}$ separately:
	\begin{itemize}
		\item Assume that $n_2 \geq \frac{\sqrt{n_1}}{2}$ holds. This is equivalent to $4 n_2^2 \geq n_1$, which implies $4n_2^2 + n_2 \geq n_1 + n_2 = |S|$. This yields $5n_2^2 \geq |S|$ and thus we obtain $n_2 \geq \frac{\sqrt{|S|}}{\sqrt{5}}$. Furthermore, we know that the number $|G|$ of guarded points is equal to $n_1$ becauase our approach sets $G := B_1$. Thus, we can upper-bound $|G|$ by $\frac{n_1}{|S|} |S| \leq \frac{n_1 + n_2 - n_2}{|S|}|S| \leq (1 - \frac{\sqrt{|S|}/\sqrt{5}}{|S|}) |S| \leq (1 - \frac{1}{\sqrt{5 |S|}})|S|$.
		\item Assume that $n_2 < \frac{\sqrt{n_1}}{2}$ holds. In that case we upper-bound $|G|$ as follows:  $n_2 < \frac{\sqrt{n_1}}{2}$ implies that there are at most $\sqrt{n_1}$ tangent points because for each point on the outer shell there are two tangent points on the inner shell. Thus, a longest sequence $\langle v_1,\dots,v_k \rangle$ on the inner shell that does not contain any tangent points has a length of at least $\sqrt{n_1}-1$. Thus, we obtain that at least $\frac{\sqrt{n_1}-1}{2}$ points are unguarded because we only choose every second point from $\langle v_1,\dots,v_k \rangle$ as guarded. 
		
		Furthermore, by combining $\frac{\sqrt{n_1}}{2} > n_2$ with $|S| = n_1 + n_2$, we get $|S| \leq \frac{4}{3} n_1$. This implies that the number of guarded points is upper-bounded by $|S| - \frac{\sqrt{\frac{3}{4}|S|}}{2}+\frac{1}{2}$, which is no larger than $(1 - \frac{1}{\sqrt{6|S|}}) |S|$.
	\end{itemize}
\end{proof}

	In order to prove that $G$ is a guard set for $\simplepolygons{S}$, we consider an arbitrarily chosen but fixed polygon $P \in \simplepolygons{S}$ and construct a partition $T$ of $P$ into convex regions, such that each region $t \in T$ is adjacent to a guarded point $v \in G$. This implies that $G$ guards the polygon $P$ because each convex region $t$ is guarded be an arbitrarily chosen corner point from $t$.\\

\noindent\textbf{Partition of $P$:} For simplification, we denote by $H_1$ and $H_2$ the convex hulls of $B_1$ and $B_2$. Below, we first describe how to determine the regions (triangles) from $P$ that are incident to points from the boundary of the convex hull of $S$, i.e. incident to $\partial H_2 \cap P$, see blue bounded regions in Figure~\ref{fig:uBtS}(b). After that we argue that the remaining parts of $P$ are convex regions $A \subseteq H_1$ that do not intersect each other, see red bounded regions in Figure~\ref{fig:uBtS}(b):
	
\begin{figure}[ht]
  \begin{center}
    \begin{tabular}{ccccc}
      \includegraphics[height=3.5cm]{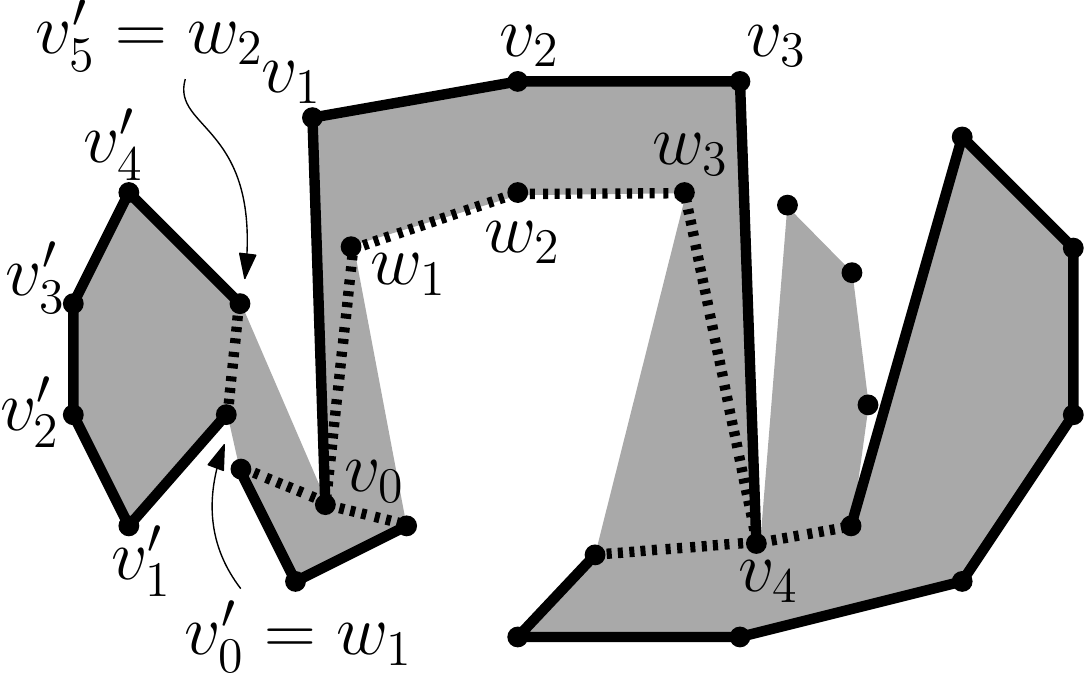}&&
      \includegraphics[height=3.5cm]{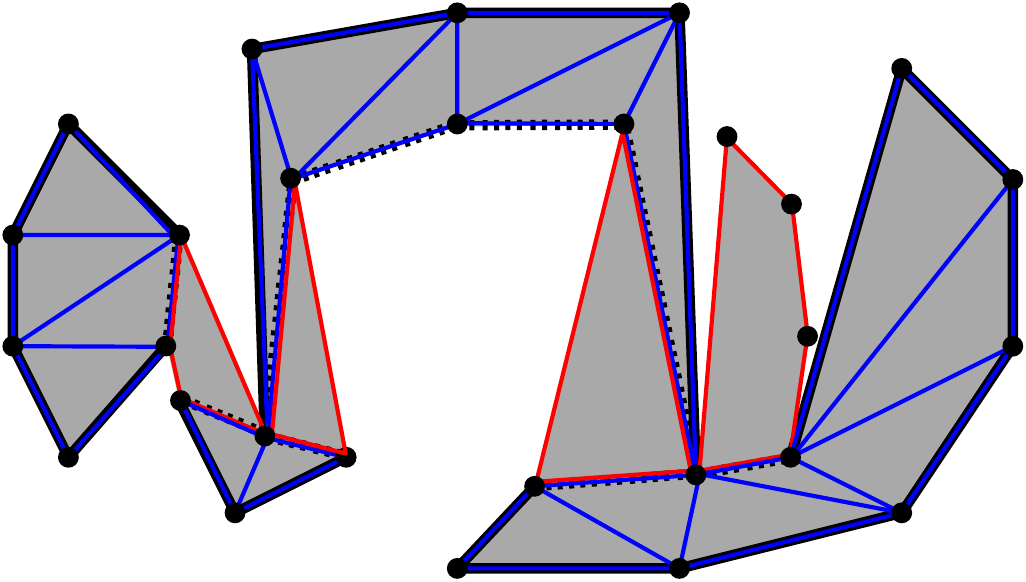}  \\
      {\small (a) Construction of $\langle v_0,...,v_{k+1} \rangle$}&&
      {\small (b) Triangles incident to $\partial H_2 \cap P$: blue} \\
      {\small and $w_1,...,w_{\ell}$ for $k=3$ and $\ell=3$.}&&
      {\small Triangulation of $H_1 \cap P$: red.}
    \end{tabular}
  \end{center}
  \vspace*{-12pt}
  \caption{Stepwise construction of $T$.}
  \label{fig:uBtS}
\end{figure}
	
			\begin{enumerate}
		 		\item Triangles that are incident to $\partial H_2 \cap P$: Let $\langle v_1,...,v_k \rangle$ be a maximal sequence of points from $B_1$ that are connected by segments from $\partial P$, see Figure~\ref{fig:uBtS}(a). The predecessor $v_0$ and successor $v_{k+1}$ of $v_1$ and $v_k$ on $\partial P$ do not lie on $H_2$, which implies that $v_0$ and $v_{k+1}$ lie $H_1$. Otherwise, $\langle v_1,...,v_k \rangle$ would not be maximal or another point $p \in P$ would be isolated such that $p$ cannot be part of $P$. Let $\langle w_1,...,w_{\ell} \rangle$ be the sequence of points that lie on $H_1$ between the segments $v_0v_1$ and $v_kv_{k+1}$, see Figure~\ref{fig:uBtS}(a). By walking simultaneously from $v_1$ to $v_k$ and from $w_1$ to $w_{\ell}$, we triangulate the polygon that is bounded by $\langle v_0,...,v_{k} \rangle$ and $\langle w_1,...,w_{\ell} \rangle$. We call the resulting triangles \emph{type 2} regions.
				\item Partition of the remaining parts: As no point from $S$ lies in the interior of $H_1$ it follows that the remaining areas of $P$ that are not yet triangulated are convex polygons $t \subseteq H_1$ that do not intersect each other, see Figure~\ref{fig:uBtS}(b). We call the resulting convex polygons \emph{type 1} regions.
			\end{enumerate}
			
	\begin{lemma}\label{lem:upperBoundTwoShellsTrianglesAdj}
		Each region $t \in T$ is adjacent to a point $v \in G$.
	\end{lemma} 
	\begin{proof}
		We distinguish if the region $t$ is of type 1 or of type 2:
			\begin{itemize}
				\item $t$ is of type 2: $t$ is adjacent to a point $v_1 \in H_1$ and adjacent to a point $v_2 \in H_2$. Because our approach ensures that all points from $H_1$ or all points from $H_2$ are guarded, it follows $v_1 \in G$ or $v_2 \in G$.
				\item $t$ is of type 1: The region $t$ is given via a sequence $\langle w_1,...,w_\ell = w_1 \rangle$ of points from $H_1$, see Figure~\ref{fig:ottfried}. In the first case of our approach, we choose all points from the inner shell $B_1$ as guarded. Thus we obtain that $w_1,\dots,w_{\ell}$ are guarded, which implies the lemma.
				
				 Next, we consider the situation achieved in the second case of our approach. In particular, we show that at least one point from $w_1,\dots,w_{\ell}$ is guarded. For the sake of contradiction, we assume that $w_1,\dots,w_{\ell}$ are unguarded. At least one edge from the boundary of $t$ is not an edge of the boundary of $P$ because otherwise the resulting circle of edges would imply that no point from $S$ lies on the outer shell. Let $w_iq$ be an edge from the boundary of $t$ such that $w_iq$ is not an edge of~$\partial P$. This implies that the edge $w_iq$ is shared by $t$ and another type~2 triangle~$\triangle$, see Figure~\ref{fig:ottfried}. Let $v$ be the third corner of $\triangle$. As $\triangle$ is of type 2, it follows that $v$ lies on the outer shell of $S$. As type 2 triangles are constructed such that no point from $S$ lies in the interior of $\triangle$ it follows that even $qv$ or $w_iv$ intersects the boundary $\partial H_1$ of the convex hull $H_1$ of the inner shell in an edge $w_ip$ or~$qp$. Without loss of generality, we assume that $qv$ intersects $\partial H_1$ in an edge $w_iq \subset \partial H_1$, see Figure~\ref{fig:ottfried}. This implies that the two unguarded points $w_i$ and $q$ are separated on~$H_1$ by the two tangent points $v_l$ and $v_r$ of $v$. Thus, our approach ensures that~$w_i$ or $q$ is guarded, which is a contradiction to the assumption that $w_1,...,w_\ell$ are unguarded. This concludes the proof.

\begin{figure}[ht]
  \begin{center}
      \includegraphics[height=3.5cm]{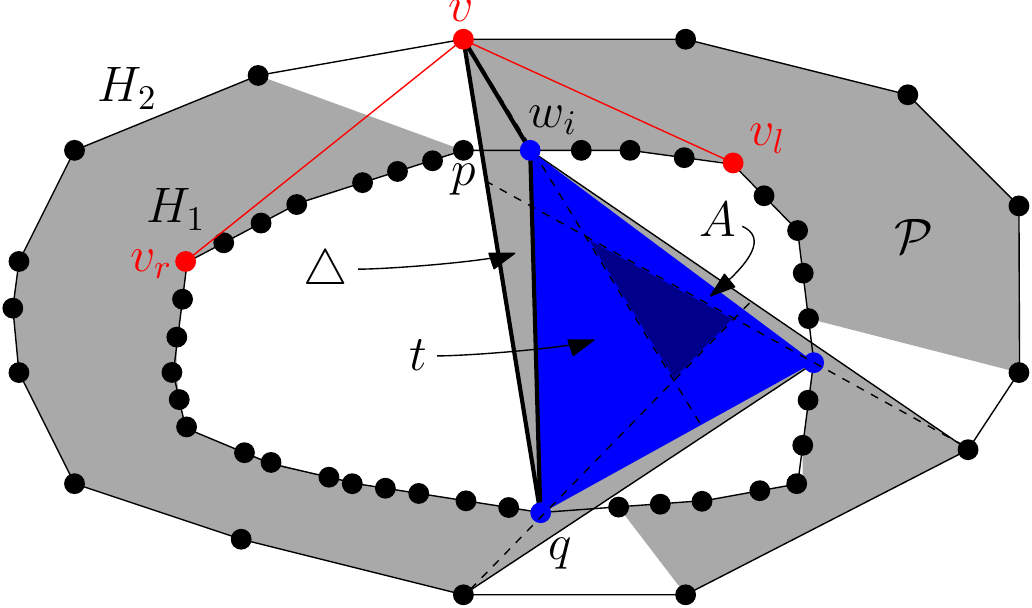} 

  \end{center}
  \vspace*{-12pt}
  \caption{A polygon $P$ causing a region $t \subset P$ of type 2 needed in the contradiction proof of Lemma~\ref{lem:upperBoundTwoShellsTrianglesAdj}. If the corners of $t$ are not guarded, there is an area $A \subseteq t$ that is not guarded. However, we prevent that all corners from $t$ are unguarded by avoiding that unguarded points on $H_1$ are separated by tangent points.}
  \label{fig:ottfried}
\end{figure}
			\end{itemize}
	\end{proof}

	We obtain Theorem~\ref{lem:upperBoundTwoshells} by combining Lemma~\ref{lem:twoShellsUpperBoundSize} and Lemma~\ref{lem:upperBoundTwoShellsTrianglesAdj}. Finally, Theorem~\ref{lem:upperBoundTwoshells} implies Corollary~\ref{cor:uniguardtwoshellsupperbound}:

\begin{corollary}\label{cor:uniguardtwoshellsupperbound}
	$\uniguardshells{n}{2} \leq \left(1 - \frac{1}{\sqrt{6 n}} \right) n$
\end{corollary}

\subsubsection{Upper Bounds for $\uniguardshells{n}{m}$ for $m \geq 3$}\label{sec:upperboundm3}

	In this section we generalize the approach for two shells to the case of $m \geq 3$. 

	Let $B_1,\dots,B_m$ be the pairwise disjoint subsets of $S$ that lie on the $m$ shells of $S$. The high-level idea of the approach is a generalization of the approach for $m=2$ and described as follows. In particular, instead of one inner shell, we now consider~$m-1$ inner shells $B_1,\dots,B_{m-1}$ that may have tangent points from points of the outer shell $B_m$. 
	
	If $|B_m|$ is ``large enough'' (larger than a value $\lambda$), we set $G =B_{1} \cup \dots \cup B_{m-1}$. Otherwise, we carefully choose one shell~$B_j$ for $j \in \{1,\dots,m-1 \}$ and select partially its points as unguarded. All the remaining points are selected as guarded.
	
	In particular, we first compute the tangent points on $B_j$ for all points from $B_{j+1} \cup \dots \cup B_m$. Next, we compute a longest sequence $\langle v_1,\dots,v_k \rangle$ of points from $B_j$ between to tangent points. Finally, we fix every second point from $\langle v_1,\dots,v_k \rangle$ as unguarded and all remaining points from $S$ as guarded.
	
	It still remains to describe how to choose $B_j$ in the second case of our approach. In particular, we choose $B_j$ as the shell such that the number of unguarded points is maximized in the worst case for the above described approach. In particular, we choose~$j$ such that $
\frac{|B_{j}|}{2(|B_{j+1}| + \dots + |B_m|)} - 1$ is maximized. This maximizes the number of unguarded points in the worst case because for each point from $B_{j+1},\dots,B_m$ there are at most two tangents on $B_j$. Furthermore, we decide if ``$|B_m|$ is large enough'' by applying worst case balancing. In particular, we set $\lambda$ to the lower bound for the number of unguarded points in the worst case, i.e. 
$\lambda:=\frac{|B_{j}|}{2(|B_{j+1}| + \dots + |B_m|)} - 1$. 
	
	By applying a similar argument as for the case of $m=2$, we can show that the computed point set $G \subseteq S$ is a guard set for $\simplepolygons{S}$. The details are developed in the rest of the subsection.

\begin{theorem}\label{thm:UpperBoundMShells}
	For any point set $S$ that lies on $m$ convex hulls we can compute in $\mathcal{O}(n \log n)$ time a guard set $G$ with $|G| \leq \left( 1 - \frac{1}{16 |S|^{\left(1-\frac{1}{2m}\right)}} \right) |S|$.
\end{theorem}

	This leads to our generalized upper bound for $\uniguardshells{n}{m}$ for $m \geq 3$:

\begin{corollary}\label{cor:uniguardupperboundsmshells}
	$\uniguardshells{n}{m} \leq \left( 1 - \frac{1}{16 n^{\left( 1-\frac{1}{2m} \right)}} \right)n$.
\end{corollary}

\paragraph{Analysis.}

	In the following we establish an upper bound for $|G|$ and show that $G$ is a guard set for $\simplepolygons{S}$. For a simplified presentation we define $n_1:=|B_1|,\dots,n_m:=|B_m|$.
	
	The following lemma is the key technical ingredient in our proof that the number of guarded points is bounded above by $\left( 1 - \frac{1}{16 n^{\left( 1-\frac{1}{2m} \right)}} \right)n$.

\begin{lemma}\label{lem:mShellsUpperBoundSizeHelpingLemma}
	The maximum of $\frac{n_{j}}{2 (n_{j+1}+\dots+n_m)} - 1$ and $n_m$ is lower-bounded by $\frac{1}{16} n^{\frac{1}{2^m}}$.
\end{lemma}
\begin{proof}
	For the sake of contradiction, assume that both values $\frac{n_{j}}{2 ( n_{j+1} + \dots + n_m)} - 1$ and $n_m$ are smaller than $\frac{1}{16} n^{\frac{1}{2^m}}$. This implies that $\frac{n_{m- \ell-1}}{2 ( n_{m-\ell} + \dots + n_{m})} - 1 < \frac{1}{16} n^{\frac{1}{2^m}}$ $(\star)$ holds for all $\ell \in \{ 0,...,m-2 \}$. Based on that, we show that $n_{m-\ell} < \frac{1}{16}n^{2^{\ell-m}}$ holds for all $\ell \in \{ 0,...,m-1 \}$. Thus we can upper-bound $n_1 + \dots + n_m$ as follows:
		\begin{eqnarray}
			 n_1 + \dots + n_m = n_{m-0} + \dots + n_{m-(m-1)}&\leq & \frac{1}{16}n^{2^{-m}} + \dots + \frac{1}{16}n^{2^{-1}} < n.
		\end{eqnarray}
	This is a contradiction because $n = n_1+ \dots+n_m$, concluding the proof.
	
	It still remains to prove that $n_{m-\ell} < \frac{1}{16}n^{2^{\ell-m}}$ holds for all $\ell \in \{ 0,...,m-1 \}$, which we do in the following. In particular, we show the stronger inequality $n_{m-\ell} + \dots + n_m < \frac{1}{16} n^{2^{\ell-m}}$ by induction over $\ell$, which implies $n_{m-\ell} < \frac{1}{16}n^{2^{\ell-m}}$, as required.
	
	For $\ell = 0$ we know by assumption that $n_m < \frac{1}{16} n^{\frac{1}{2^m}}$ holds. Assume that $n_{m-\ell} + \dots + n_m < \frac{1}{16}n^{2^{\ell-m}}$ $(\dagger)$ holds. Based on that we show $n_{m-\ell-1} + \dots + n_m < \frac{1}{16}n^{2^{\ell+1-m}}$ as follows:
	
	By the assumption $(\star)$, we know that $\frac{n_{m- \ell-1}}{2 ( n_{m-\ell} + \dots + n_{m})} - 1 < \frac{1}{16} n^{\frac{1}{2^m}}$ holds. Combining this with the assumption $n_{m-\ell} + \dots + n_m < \frac{1}{16}n^{2^{\ell-m}}$ $(\dagger)$ of the induction yields $\frac{n_{m-(\ell+1)}}{\frac{2}{16} n^{2^{\ell-m}}} -1<\frac{1}{16} n^{\frac{1}{2m}}$. This implies $n_{m-\ell-1} < \frac{6}{256} n^{2^{\ell+1-m}}$. A final application of the assumption $(\star)$ yields $n_{m-\ell-1} + \dots + n_m < \frac{6}{256}n^{2^{\ell+1-m}} + \frac{1}{16}n^{2^{\ell-m}}$, which, in turn, is smaller than $\frac{1}{16}n^{2^{\ell+1 - m}}$.
\end{proof}

	By applying Lemma~\ref{lem:mShellsUpperBoundSizeHelpingLemma} we can upper-bound $|G|$ as required:

\begin{corollary}\label{cor:mShellsUpperBoundSize}
	$|G| \leq \left(1 - \frac{1}{16 |S|^{\frac{2m-1}{2m}}} \right) |S|$.
\end{corollary}
\begin{proof}
	Our approach guarantees that the number of unguarded points is lower-bounded by the maximum of $\frac{n_j}{2 (n_{j+1} + \dots + n_m )} - 1$ and $n_m$. By Lemma~\ref{lem:mShellsUpperBoundSizeHelpingLemma}, this is lower-bounded by $\frac{1}{16}|S|^{\frac{1}{2m}}$. Thus, the number of guarded points can be upper-bounded by $|S| - \frac{1}{16}|S|^{\frac{1}{2m}}= \left(1 - \frac{1}{16 |S|^{\frac{2m-1}{2m}}} \right) |S|$.
\end{proof}

	Finally, we show that $G$ is a guard set for $\simplepolygons{S}$. In particular, we consider an arbitrarily chosen but fixed polygon $P \in \simplepolygons{S}$ and construct a partition $T$ of $P$ into convex regions, such that each region $t \in T$ is adjacent to a vertex $v \in G$.\\
	
	Roughly speaking, we extend the approach for determining a partition in the case of two shells to the case of $m$ shells for $m \geq 3$. In particular, we repeatedly apply the first step of the above approach and remove the corresponding triangles from the polygon until the remaining points lie on one shell. Finally, we apply the second step of the approach for two shells to the area that is given by the remaining regions. In the following, we give the details of this approach.
	
\noindent\textbf{Partition of $P$:} For $i \in \{ 1,\dots,m \}$, let $H_i$ be the convex hull of $B_i$. The basic idea for the construction of the partition of $P$ is the following. Consecutively, for each $i = m,...,2$ we compute the triangles that are incident to $\partial H_i \cap P$ just like we do for $H_2$ in the case of two shells, see Figure~\ref{fig:upperBoundMShells}(b)--(e). Finally, we argue that the remaining parts of $P$ are convex regions $t \subseteq H_1$ that do not intersect each other, see Figure~\ref{fig:upperBoundMShells}(f).
	
\begin{figure}[ht]
  \begin{center}
    \begin{tabular}{ccccc}
      \includegraphics[height=2.4cm]{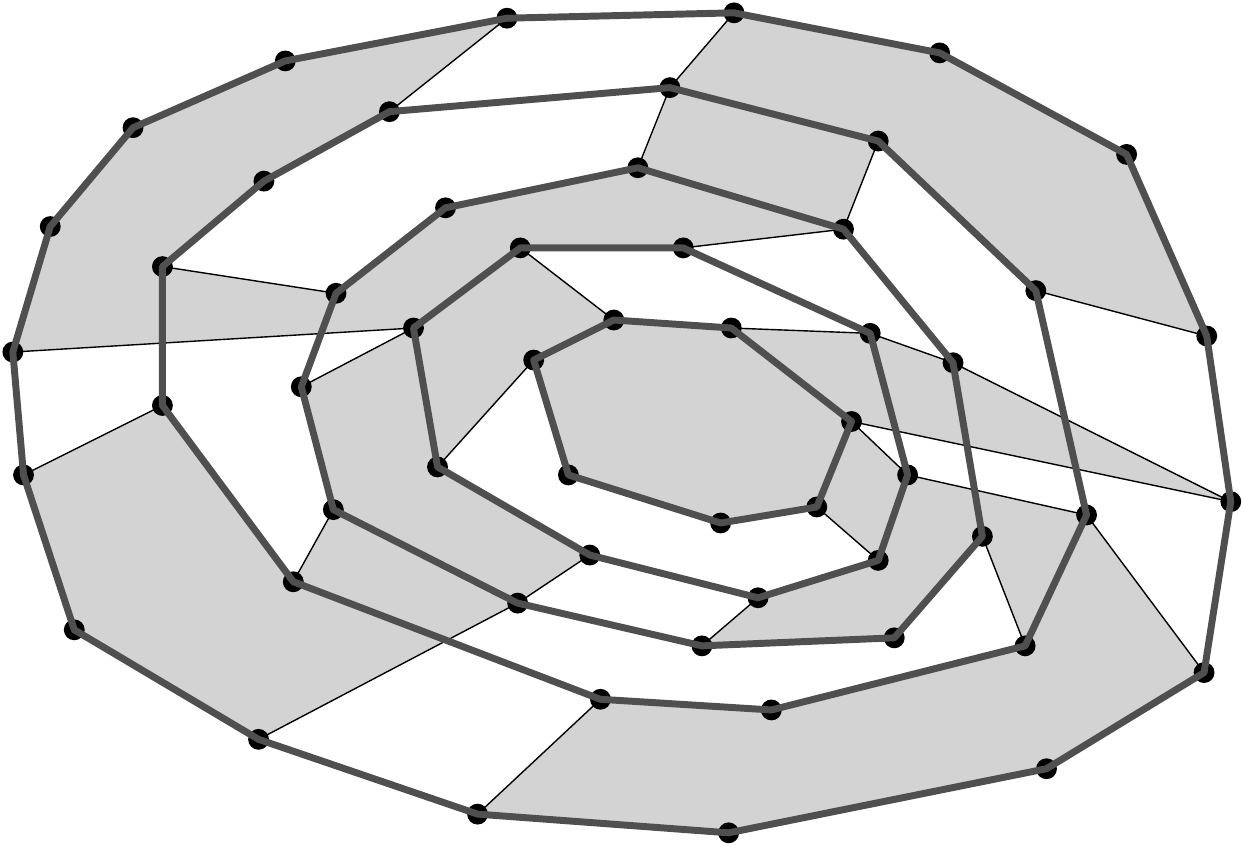} &&
      \includegraphics[height=2.4cm]{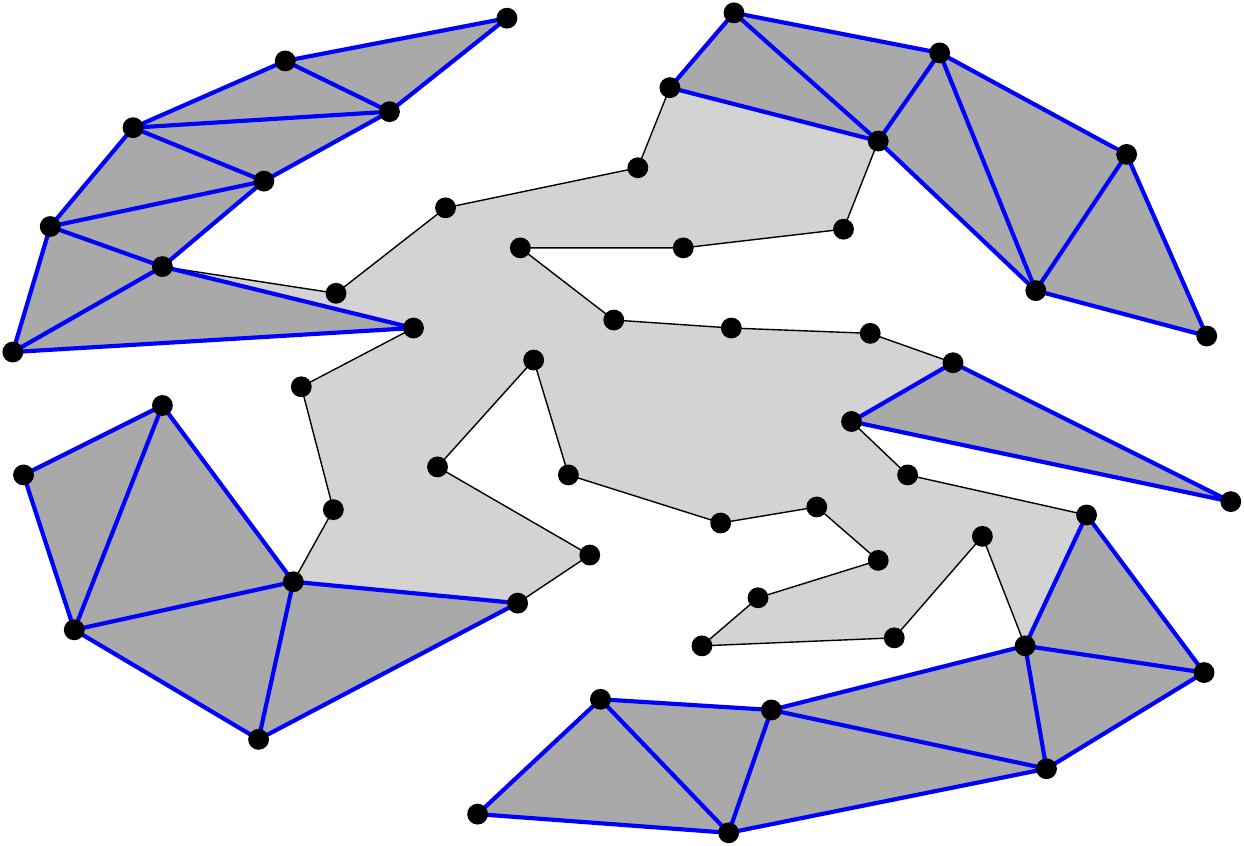}&&
      \includegraphics[height=2.4cm]{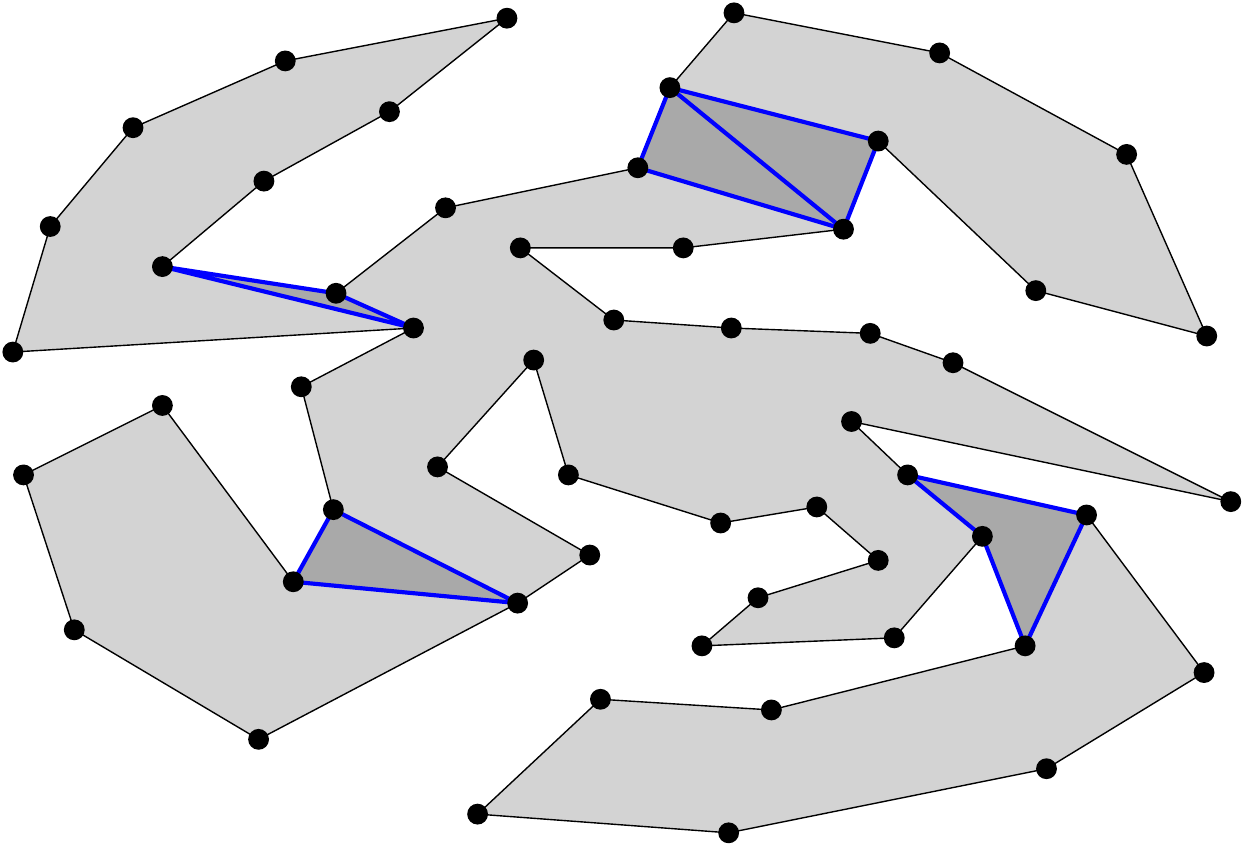}  \\
      {\small (a) The point set and} &&      
      {\small (b) Triangles that are}&&
      {\small (c) Triangles that are} \\ 
      {\small a possible polygon.} &&      
      {\small incident to $\partial H_5$.}&&
      {\small incident to $\partial H_4$.} \\   
      \includegraphics[height=2.4cm]{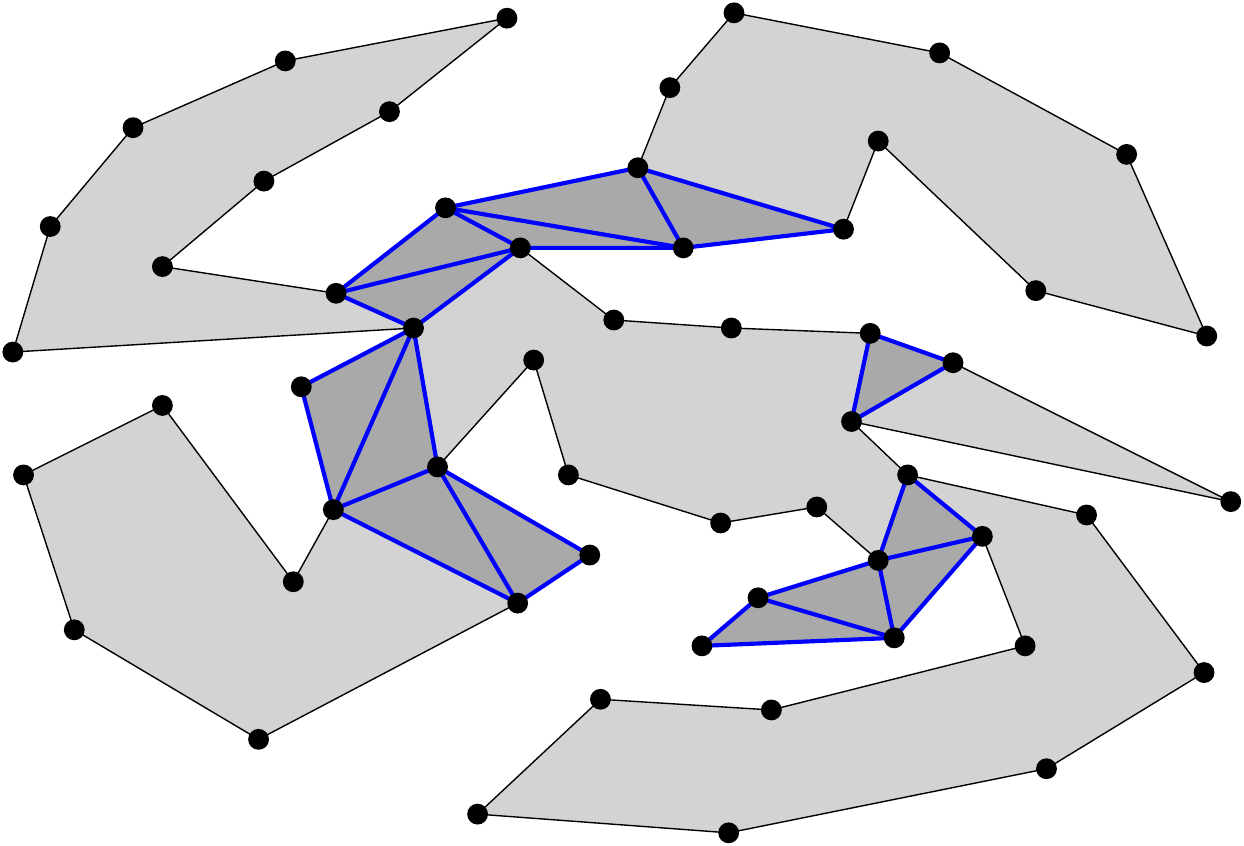} &&
      \includegraphics[height=2.4cm]{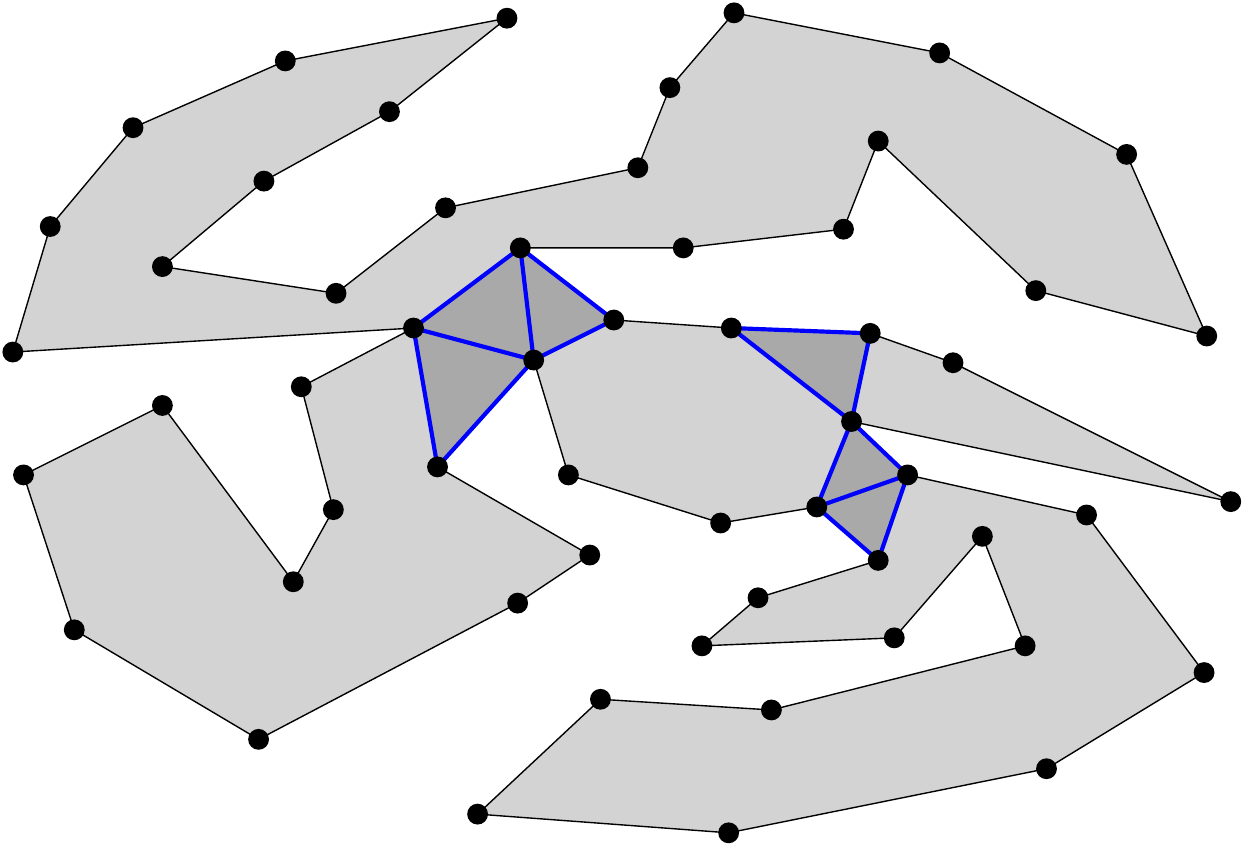}&&
      \includegraphics[height=2.4cm]{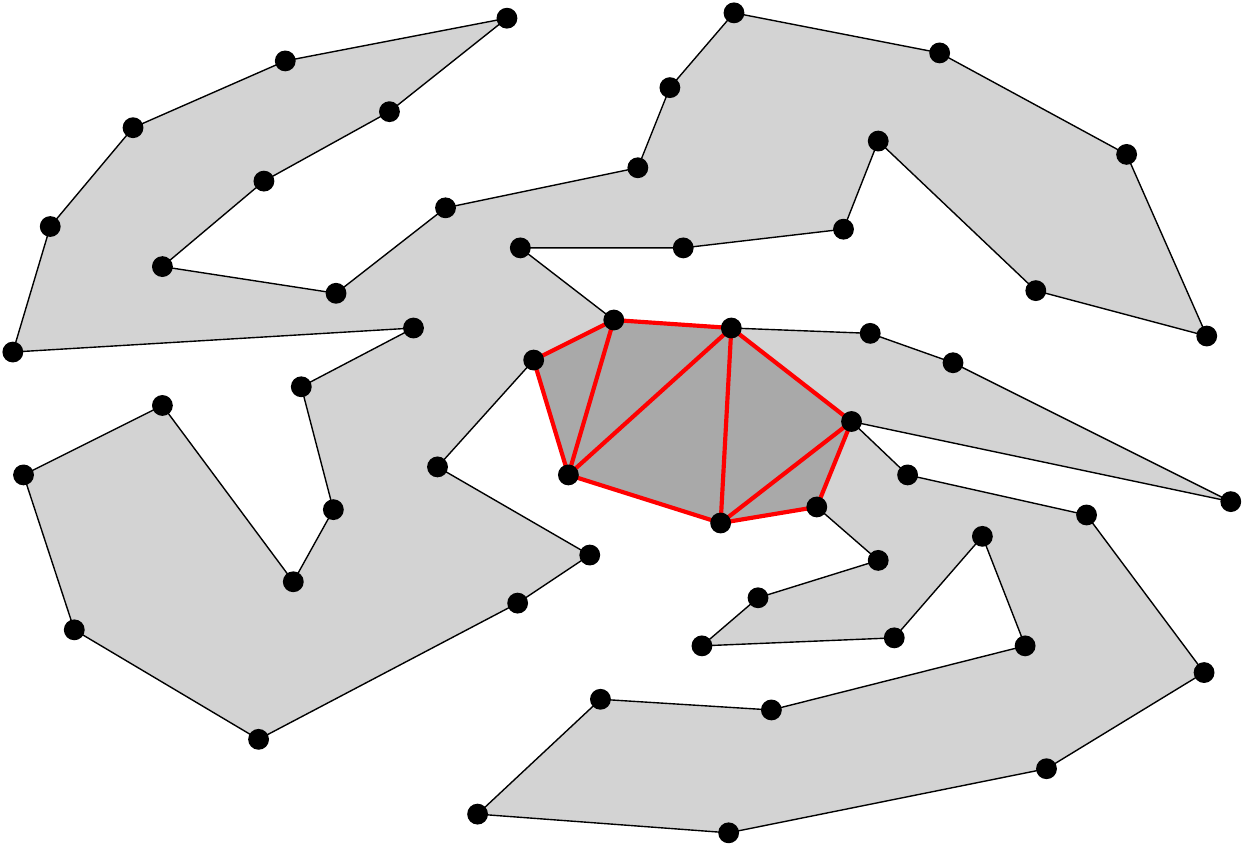}  \\
      {\small (d) Triangles that are} &&      
      {\small (e) Triangles that are}&&
      {\small (f) Regions that are} \\
      {\small incident to $\partial H_3$} &&      
      {\small incident to $\partial H_2$}&&
      {\small incident to $\partial H_1$} 
    \end{tabular}
  \end{center}
  \vspace*{-12pt}
  \caption{Stepwise construction of the partition of $P$ for the case of five shells.}
  \label{fig:upperBoundMShells}
\end{figure}
	
	\begin{itemize}
		\item Triangles that are incident to outer shells: The construction of the triangles proceeds from $H_m$ to~$H_2$. In particular, we iterate the following construction for $i = m,...,2$: Let $\langle v_1,...,v_k \rangle$ be a maximal sequence of points on $\partial H_i$ that are connected by segments from $P$, such that no segment $v_{j}v_{j+1}$ intersects the interior of $H_{i-1}$. Let $v_0$ and $v_{k+1}$ be the points before and after $v_1$ and $v_k$ on the boundary of $P$. Let $\langle w_1,...,w_{\ell} \rangle$ be the sequence of vertices on $H_{i-1}$ that lies between the segments $v_0v_1$ and $v_kv_{k+1}$. By walking simultaneously from $v_1$ to $v_k$ and from $w_1$ to $w_{\ell}$, we triangulate the polygon that is bounded by $\langle v_0,...,v_k \rangle$ and $\langle w_1,...,w_{\ell} \rangle$. We call the resulting triangles \emph{type $i$} regions. 
		
		We remove all type $i$ regions from $P$ and repeat the above construction for $i := i-1$ until $i=1$.
		\item Partition of the remaining parts: By the same argument as in the case of two shells we know that the remaining parts of $P$ are convex polygons $t \subseteq H_1$ that do not intersect each other. We call the resulting convex polygons \emph{type 1} regions.
	\end{itemize}
	
\begin{lemma}\label{lem:mShellsAdjacentToGuarded}
	Each region $t \in T$ is adjacent to a point $v \in P$ such that $v \in G$.
\end{lemma}
\begin{proof}
	All triangles that are not of type $j$ are adjacent to a point $v \in G$. Thus we assume, without loss of generality, that $t$ is of type $j$. By the same argument we are allowed to assume that all points of $t$ lie on $\partial H_j$; by the same argument as applied for type~1 regions in the case of two shells, it follows that at least one vertex of $t$ is guarded. This concludes the proof.
\end{proof}

\statement{Theorem}{thm:UpperBoundMShells}
		{\em
	For each point set $S$ that lies on $m$ convex hulls we can compute in $\mathcal{O}(n \log n)$ time a guard set $G$ with $|G| \leq \left( 1 - \frac{1}{16 |S|^{\frac{2m-1}{2m}}} \right) |S|$.
}

%% file: 04-boundskuniversal.tex
\section{Bounds for the $k$-Universal Guard Numbers}\label{sec:kuniguards}

In the following we state several lower and upper bounds for various $k$-universal guard numbers.  

\subsection{Lower bounds for $\kuniguard{n}{k}$}\label{sec:lowerboundskuniguards}

\begin{theorem}\label{lem:lowerBounds2UGP}
	$\kuniguard{n}{2} \geq \lfloor \frac{3n}{8} \rfloor$
\end{theorem}

\begin{proof}

        For each $n \in \mathbb{N}$ we give a pair of simple polygons that have a common set of vertices of size $n$, such that each guard set for $\{ P_{n,1},P_{n,2} \}$ has a size of at least $\lfloor \frac{3n}{8} \rfloor$. This implies $\kuniguard{n}{2} \geq \lfloor \frac{3n}{8} \rfloor$.

        \begin{figure}[h!]
  \begin{center}
    \begin{tabular}{ccc}
      \includegraphics[height=1.3cm]{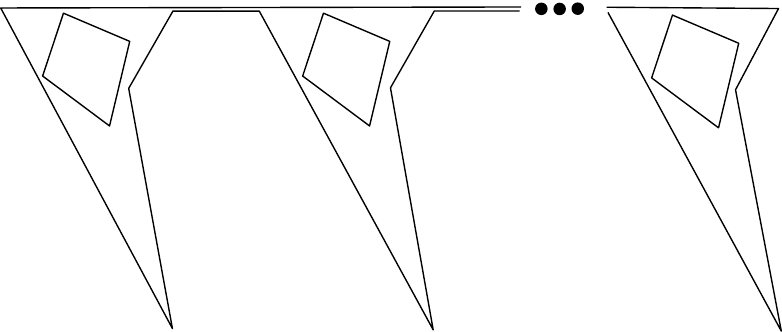} &&
       \includegraphics[height=1.3cm]{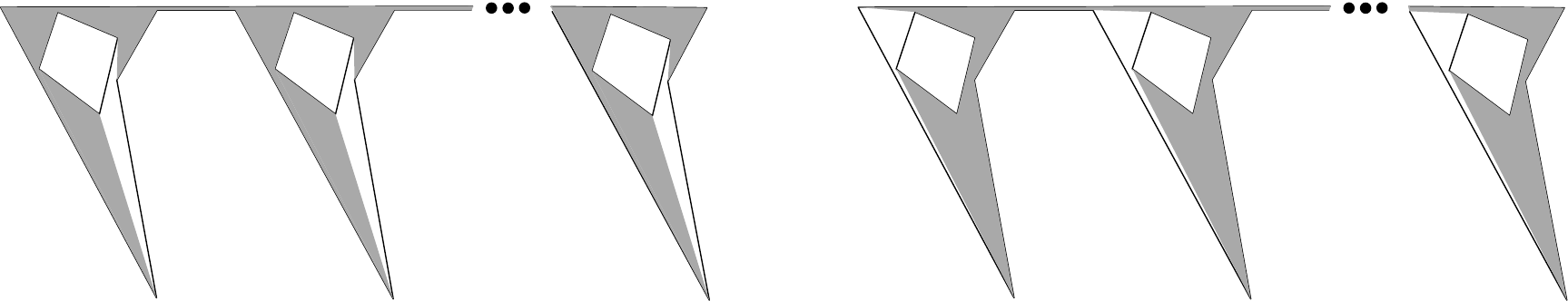}\\
      {\small (a) The polygon $P$.} &&
      {\small (b) The polygons $P_1$ and $P_2$.}
    \end{tabular}
  \end{center}
  \vspace*{-12pt}
  \caption{A $\frac{3n}{8}$ lower-bound construction for $\kuniguard{n}{2}$: Covering a $\frac{3n}{8}$ lower-bound construction for $\kuniguardholes{n}{1}$.}
  \label{fig:lowerBound2UGP}
\end{figure}

        Consider the polygon $P$ that is illustrated in Figure~\ref{fig:lowerBound2UGP}(a). Each guard set for $P$ has size at least $\lfloor \frac{3n}{8} \rfloor$, where $n$ is the number of vertices of $P$. We construct two polygons $P_1$ and $P_2$, as illustrated in Figure~\ref{fig:lowerBound2UGP}(b). We have $P_1 \cup P_2 = P$ at which $P_1$, $P_2$, and $P$ have the same vertices. Furthermore, we have $a \sees{P} b$ if $a \sees{P_1} b$ and $a \sees{P_2} b$. Thus, a guard set for $\{ P_1,P_2 \}$ is at least as large as a guard set for $P$. This concludes the proof.

\end{proof}

\begin{theorem}\label{lem:lowerBounds3UGP}
	$\kuniguard{n}{3} \geq \lfloor \frac{4n}{9} \rfloor$.
\end{theorem}

\begin{proof}
For each $n \in \mathbb{N}$ we give a set of three simple polygons
that have a common set of vertices of size $n$, such that each guard
set for $\{ P_{n,1},P_{n,2}, P_{n,3} \}$ has a size of at least
$\lfloor \frac{4n}{9} \rfloor$. This implies $\kuniguard{n}{3} \geq
\lfloor \frac{4n}{9} \rfloor$.

First, consider an example (see
Figure~\ref{fig:lowerBound3UGP1}), with three simple polygons on a
set of $n=9$ points.  By a brute-force check of all
$\binom{9}{3}$ possible triples of points, we see that three guards do not suffice
to guard all three polygons; however, four guards easily
do. 

\begin{figure}[ht]
  \begin{center}
    \begin{tabular}{ccccc}
      \includegraphics[height=2.8cm]{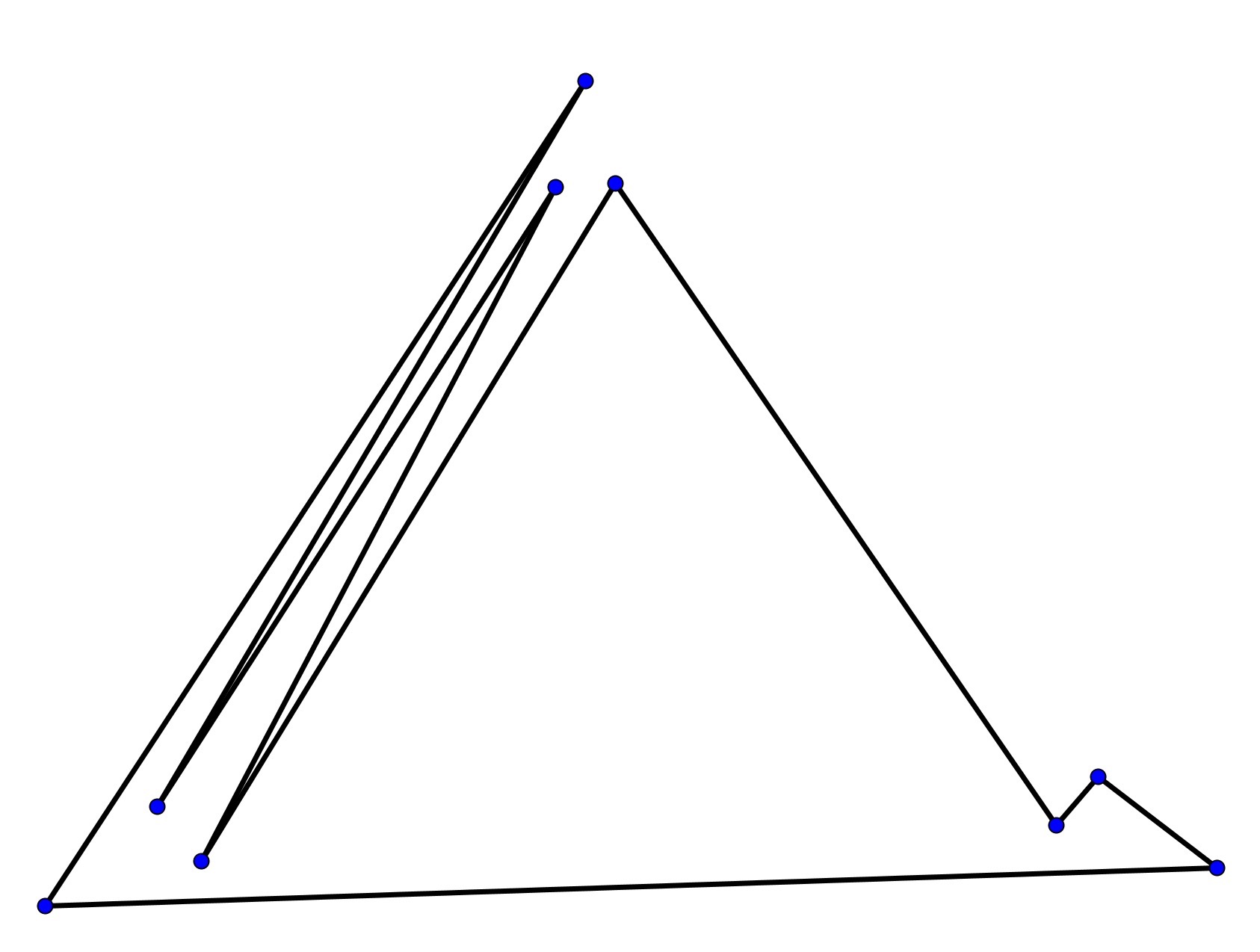} &&
      \includegraphics[height=2.8cm]{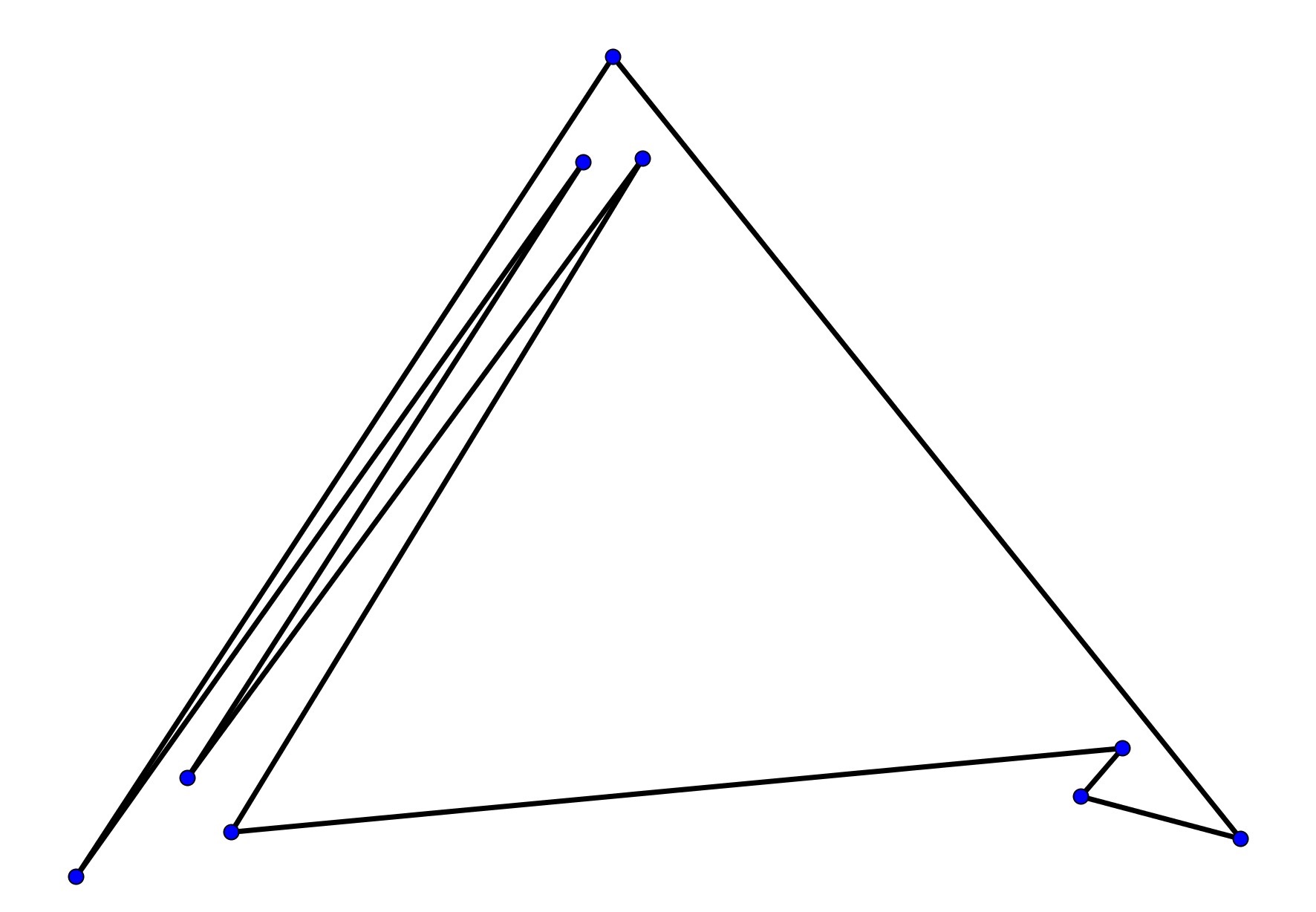} &&
      \includegraphics[height=2.8cm]{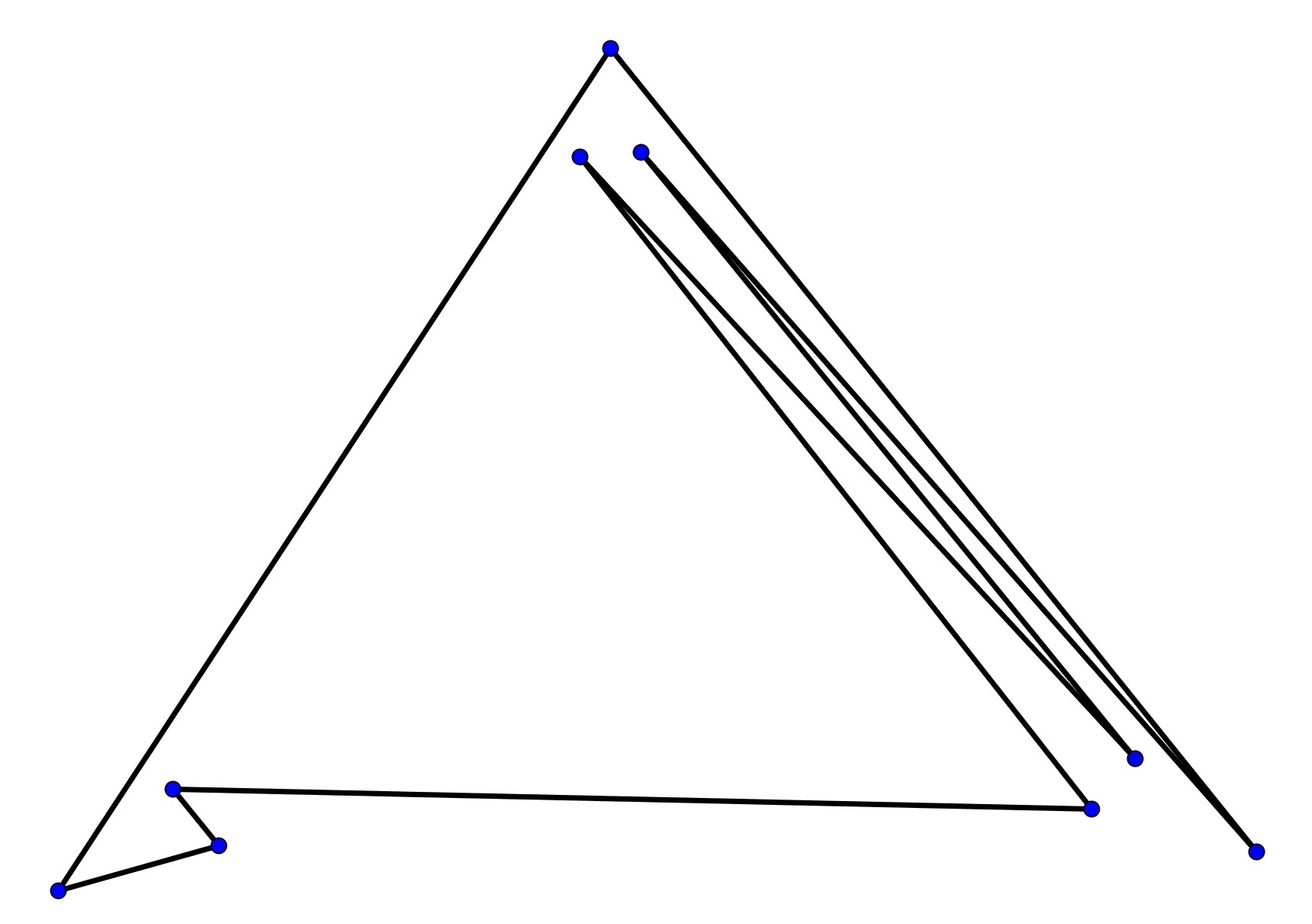}\\
    \end{tabular}
  \end{center}
  \vspace*{-12pt}
  \caption{The polygons $P_1$, $P_2$, $P_3$ require four 3-universal guards for $\kuniguard{n}{3}$ when $n=9$.}
  \label{fig:lowerBound3UGP1}
\end{figure}

We extend the example 
(Figure~\ref{fig:lowerBound3UGP2}), by connecting copies of the
polygons in Figure~\ref{fig:lowerBound3UGP1} with the vertices of a
much larger bounding triangle. 
In this way, for each
point set of size nine, we need at least four guards; for large enough $n$,
we can ignore the three vertices of the outer big triangle. This concludes
the proof.

\begin{figure}[ht]
  \begin{center}
    \begin{tabular}{ccccc}
      \includegraphics[height=2.8cm]{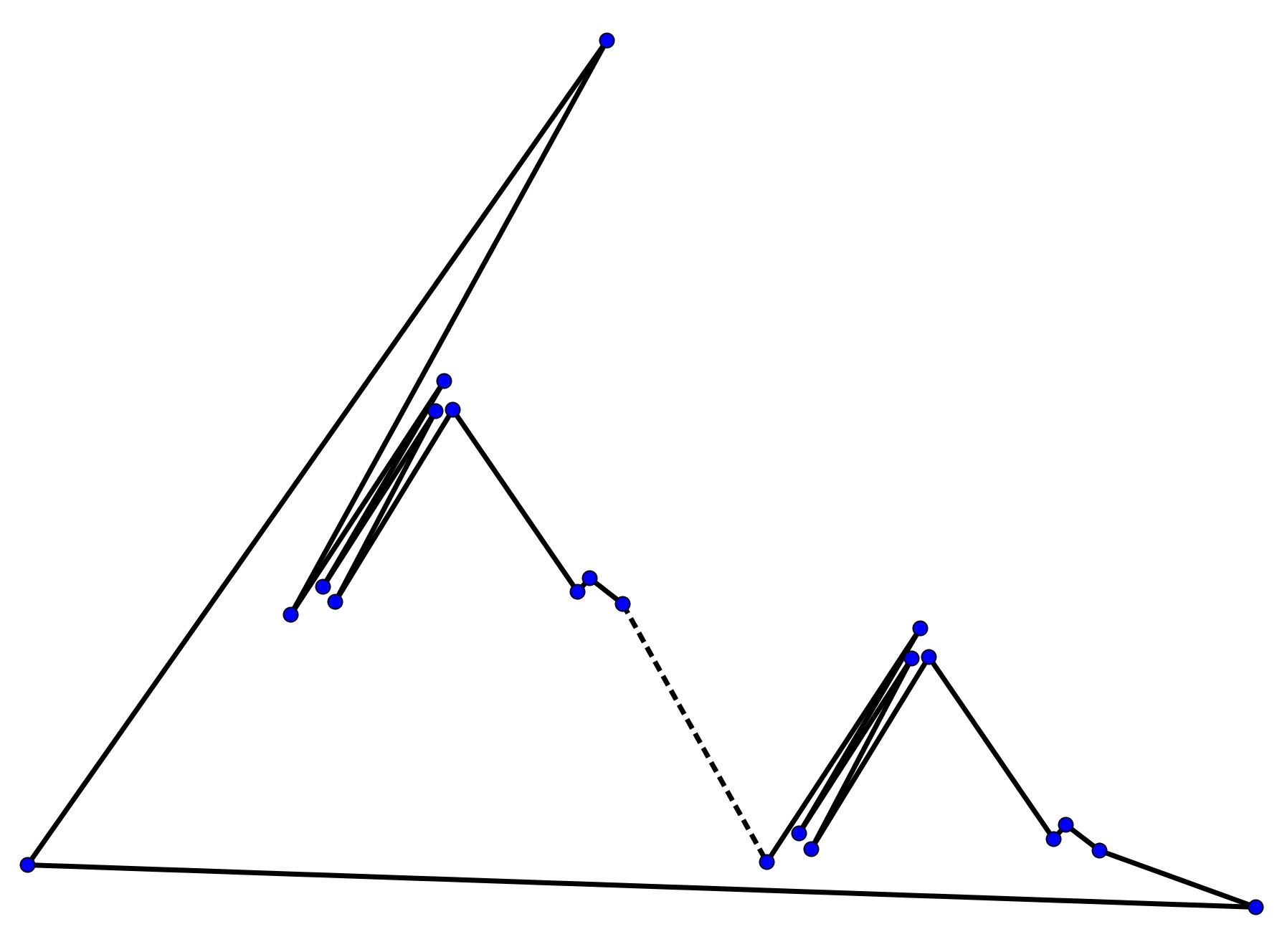} &&
      \includegraphics[height=2.8cm]{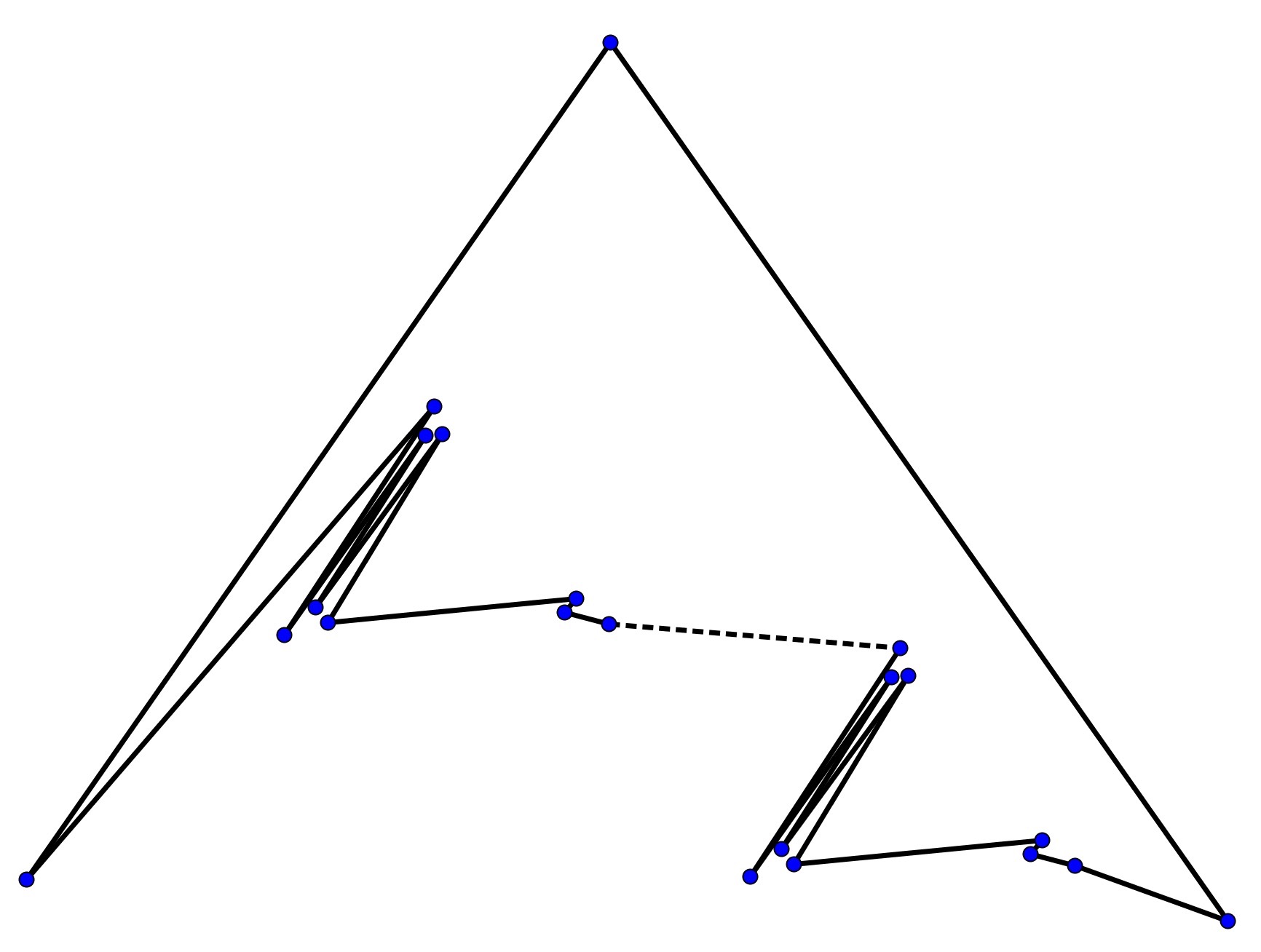} &&
      \includegraphics[height=2.8cm]{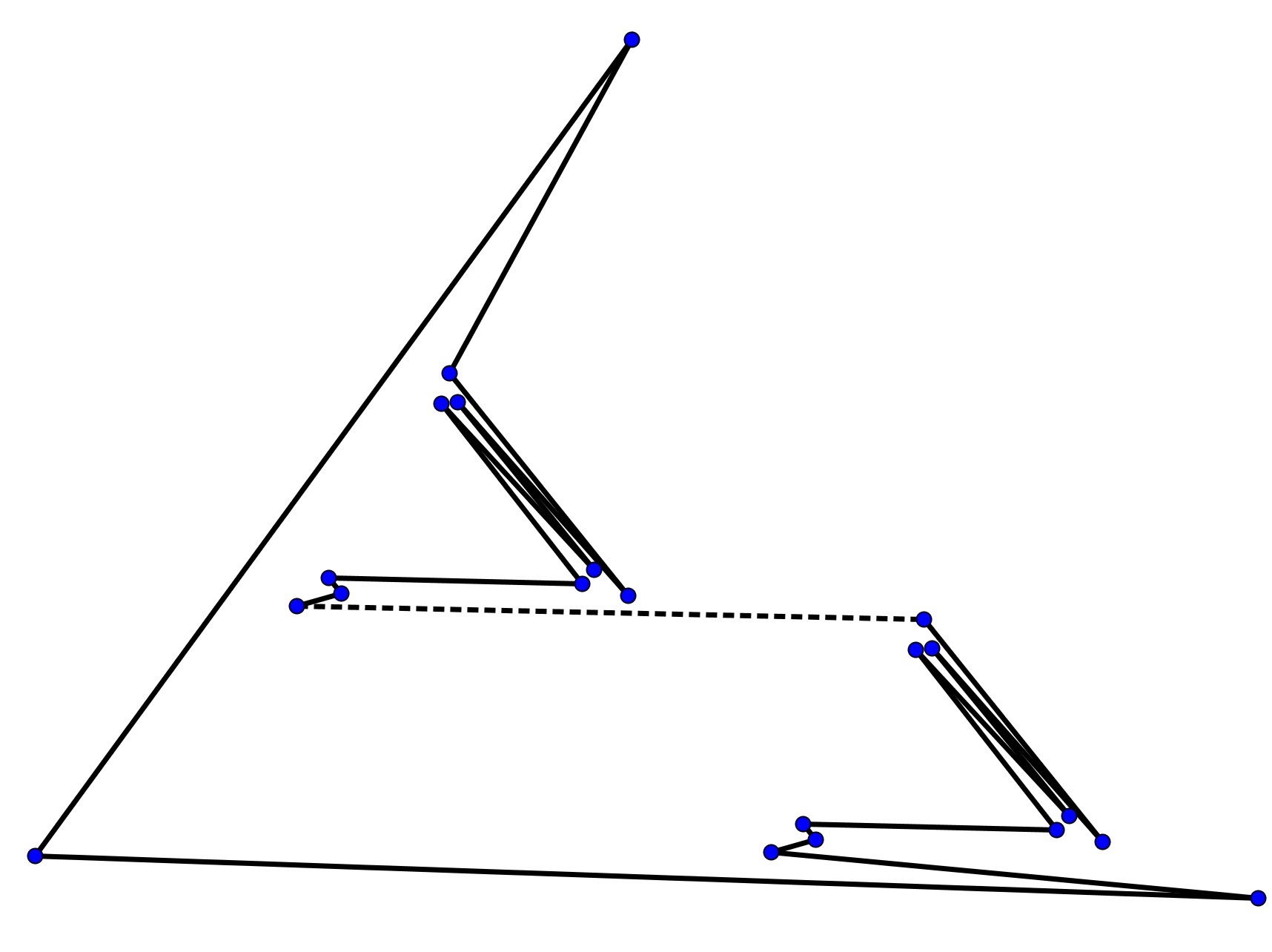}\\
    \end{tabular}
  \end{center}
  \vspace*{-12pt}
  \caption{The general polygons $P_{n,1}$, $P_{n,2}$, $P_{n,3}$ require $\lfloor \frac{4n}{9} \rfloor$ 3-universal guards for $\kuniguard{n}{3}$.}
  \label{fig:lowerBound3UGP2}
\end{figure}
\end{proof}

\begin{theorem}\label{lem:lowerBounds4UGP}
	$\kuniguard{n}{5} \geq \kuniguard{n}{4} \geq \frac{n}{2} - 8 \sqrt{n}-23$.
\end{theorem}

\begin{proof}

For each $n \in \mathbb{N}$ we give a set of four of simple polygons $P_1$, $P_2$, $P_3$, and $P_4$ that have a common set of $n$ vertices. Let $G$ be a guard set for $\{ P_1,P_2,P_3,P_4 \}$. We show $|G| \geq \frac{1}{2} n - 16 \sqrt{n}- 4$.

        First, we give the required construction of $P_1$, $P_2$, $P_3$, and $P_4$, such that $n = (4 \ell)^2 + 16 \ell +4$ for $\ell \in \mathbb{N}$ and show that a corresponding guard set needs at least $\frac{n}{2} - 16 \sqrt{n}-4$ guards. Next we show how to extend the construction and the corresponding argument appropriately to an arbitrary $n \in \mathbb{N}$.

        We construct $P_1$, $P_2$, $P_3$, and $P_4$, as illustrated in Figure~\ref{fig:lowerbound4ugp}. The vertices in the middle block are structured in groups of size four. Assume that one of these groups has only one guarded point. This implies that the other points are unguarded and thus build an unguarded area in $P_1$, $P_2$, $P_3$, or $P_4$, as illustrated by the dark gray cones. Hence, each of these groups has two guarded points. This implies $|G| \geq \frac{1}{2}(4 \ell)^2 = \frac{n}{2} - 8 \ell - 2 \geq \frac{n}{2} - 8 \sqrt{n} + 4$, because $16 \ell^2 + 16 \ell + 4 =n$ implies $\ell \leq \sqrt{n} - \frac{1}{2}$.

        \begin{figure}[ht]
  \begin{center}
    \begin{tabular}{ccccccc}
      \includegraphics[height=2.5cm]{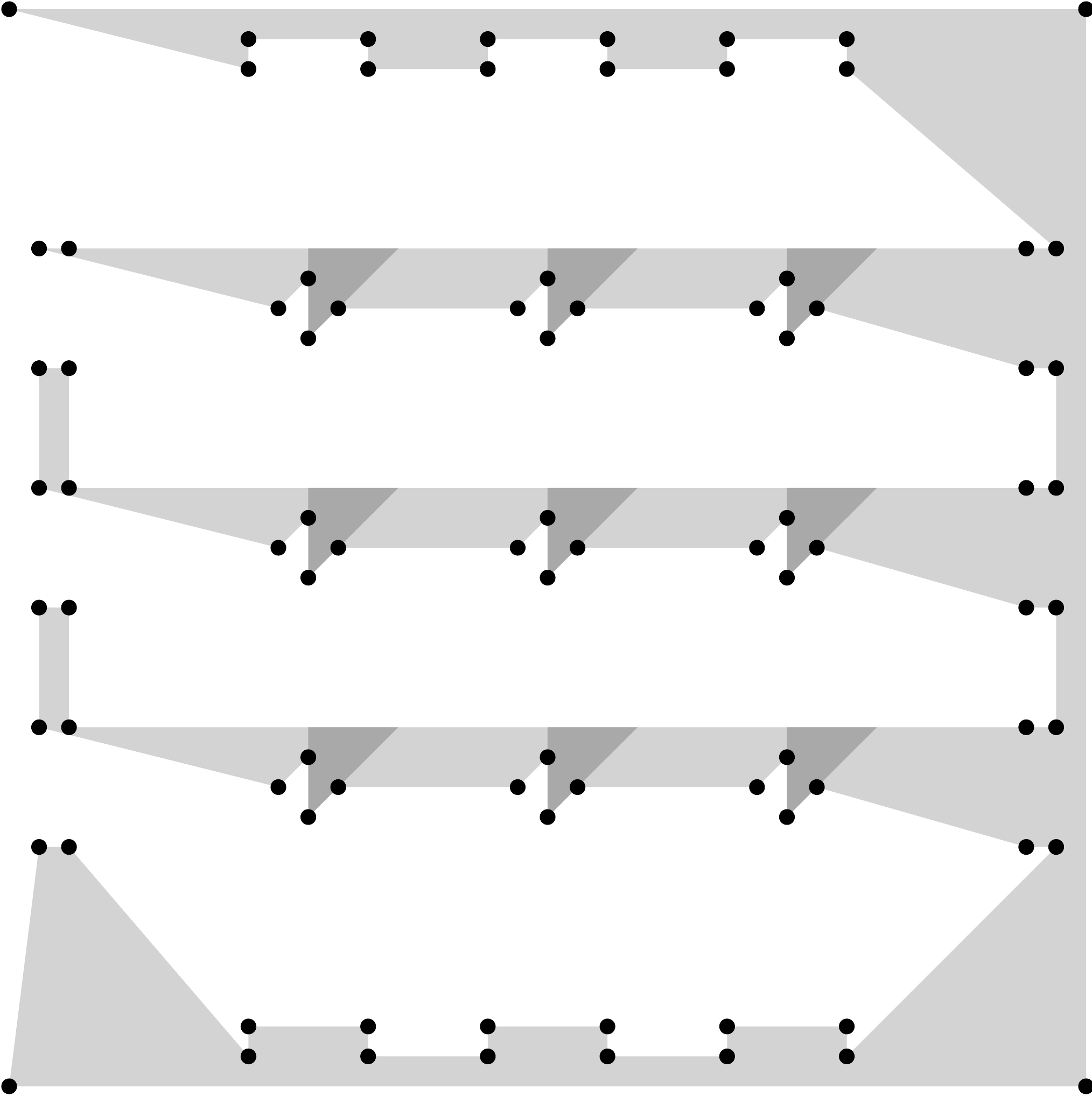} &&
       \includegraphics[height=2.5cm]{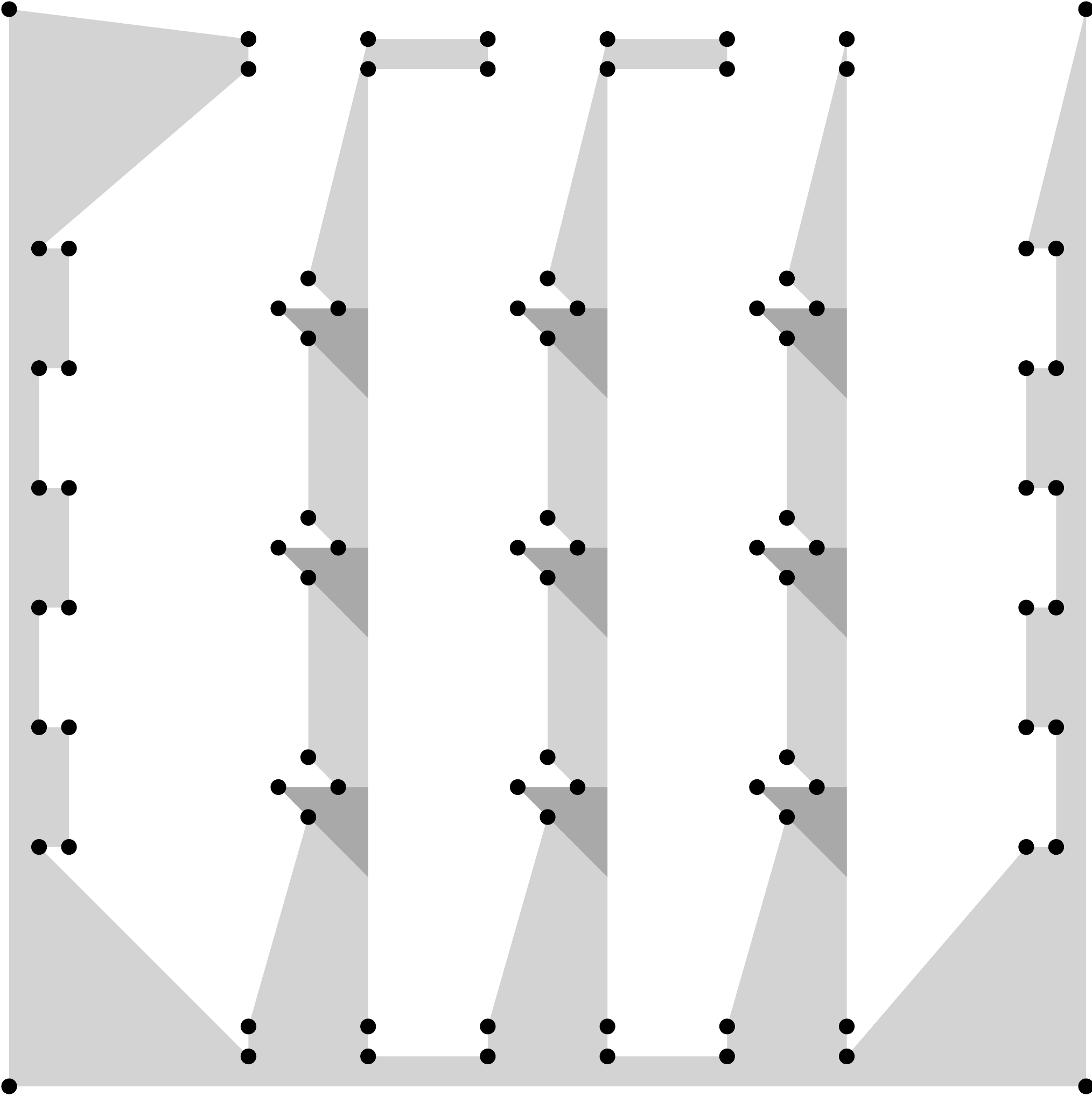}&&
       \includegraphics[height=2.5cm]{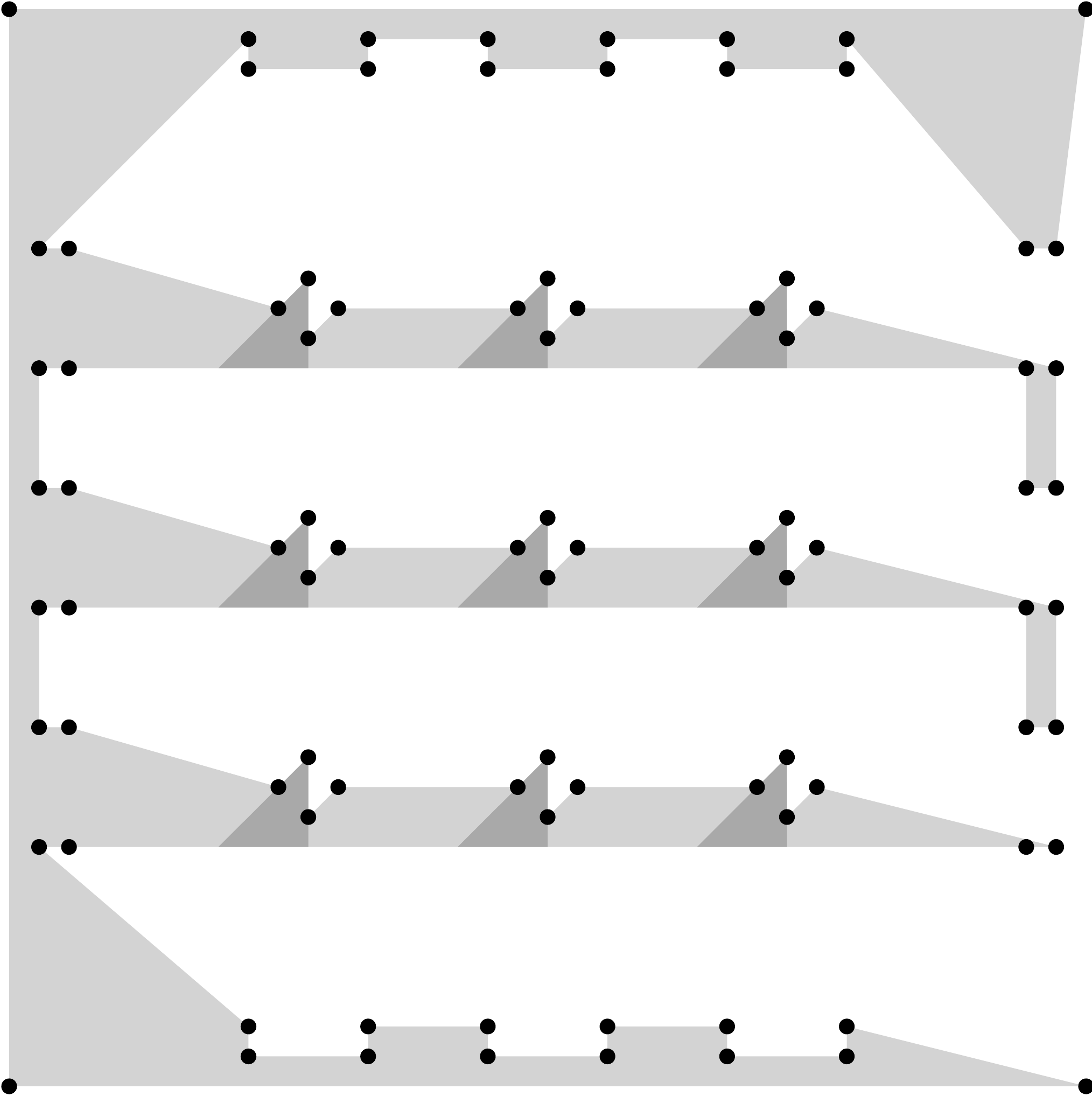} &&
       \includegraphics[height=2.5cm]{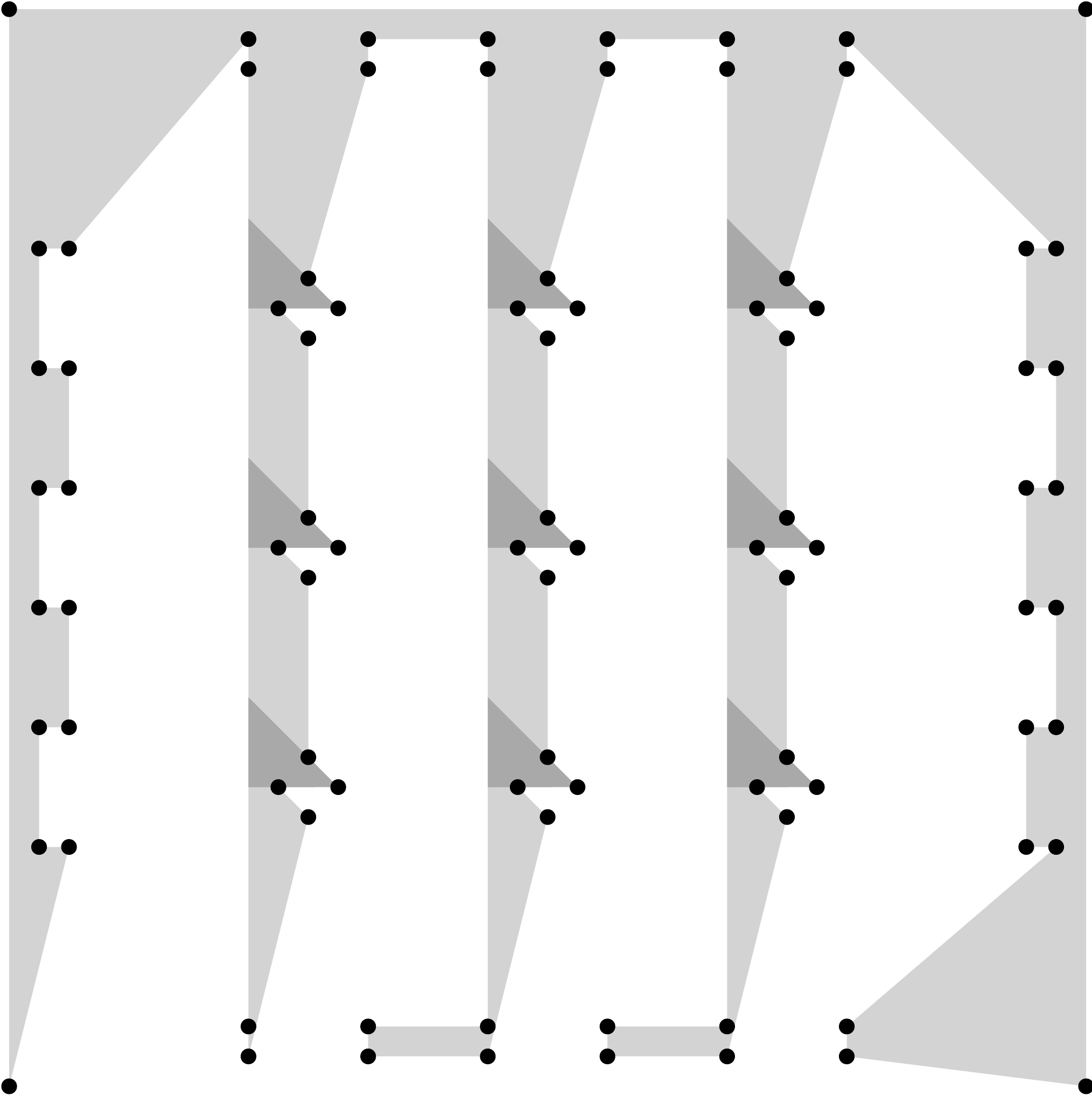}\\
    \end{tabular}
  \end{center}
  \vspace*{-12pt}
  \caption{Lower-bound construction of $\frac{1}{2} n - 8 \sqrt{n}- 4$ for $k$-universal guard numbers.}
  \label{fig:lowerbound4ugp}
\end{figure}

Finally, we give an extension of the above approach to an arbitrary $n \in \mathbb{N}$. Let $\ell_{0}$ be the largest value such that $16 \ell^2 + 16 \ell + 4  \leq n$. We apply the above construction for $16 \ell^2 + 16 \ell + 4$ points. Alle the remaining points are added in a new row and column of the above construction. The worst case for that approach is $n = 16 \ell^2 + 16 \ell + 4 + 19$, as $19$ additional points are needed until the first new guard is enforced. This concludes the proof.

\end{proof}

\begin{theorem}\label{lem:lowerBounds5UGP}
	$\kuniguard{n}{k} \geq \lfloor \frac{5n}{9} \rfloor$ for $k \geq 6$.
\end{theorem}

\begin{proof}
This proof is similar to the proof of
Theorem~\ref{lem:lowerBounds3UGP}. In addition to $P_1$, $P_2$, $P_3$ in
Figure~\ref{fig:lowerBound3UGP1}, we add three more polygons $P_4,
P_5, P_6$; refer to Figure~\ref{fig:lowerBound3UGP3}.
\begin{figure}[ht]
  \begin{center}
    \begin{tabular}{ccccc}
      \includegraphics[height=2.4cm]{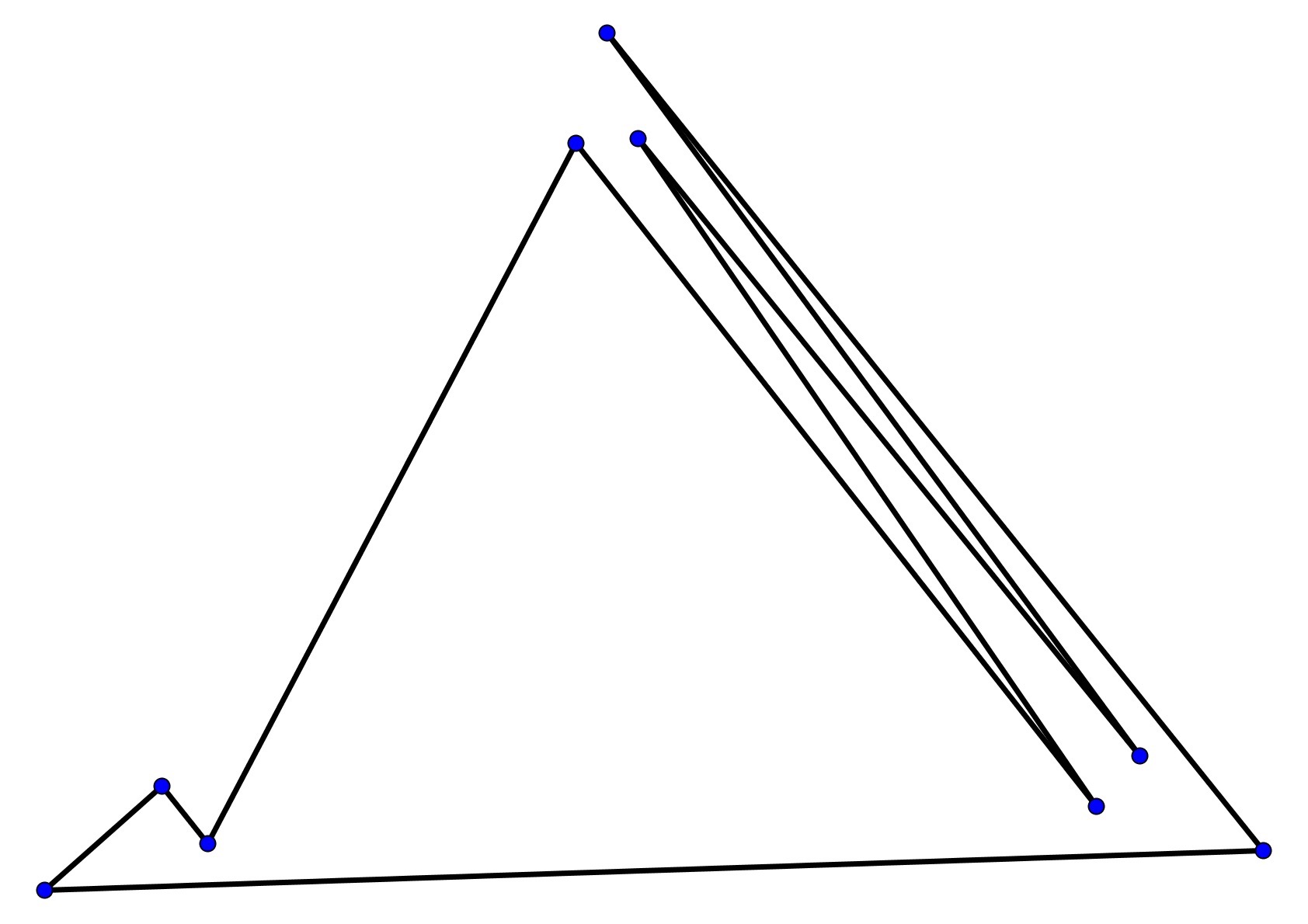} &&
      \includegraphics[height=2.4cm]{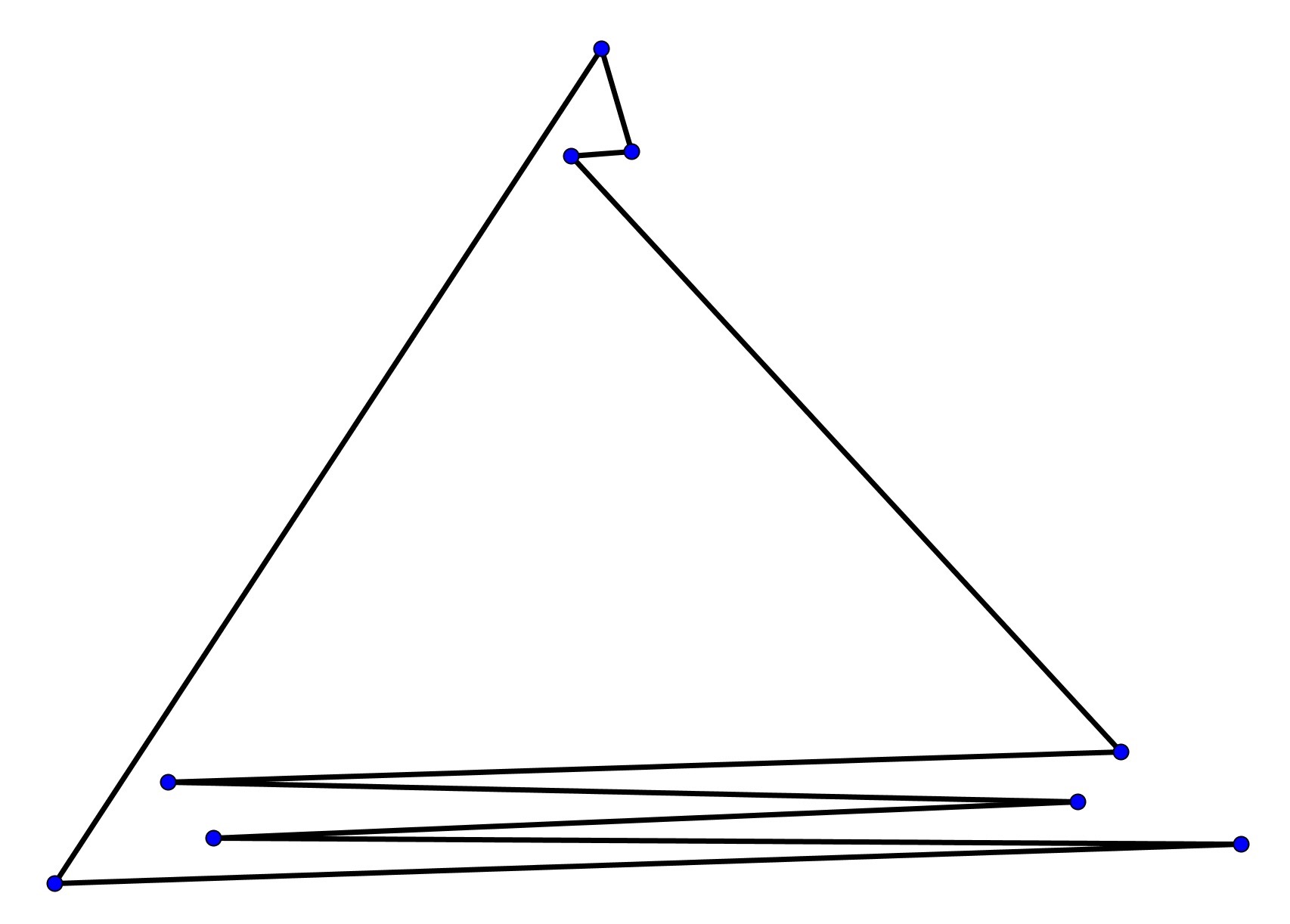} &&
      \includegraphics[height=2.4cm]{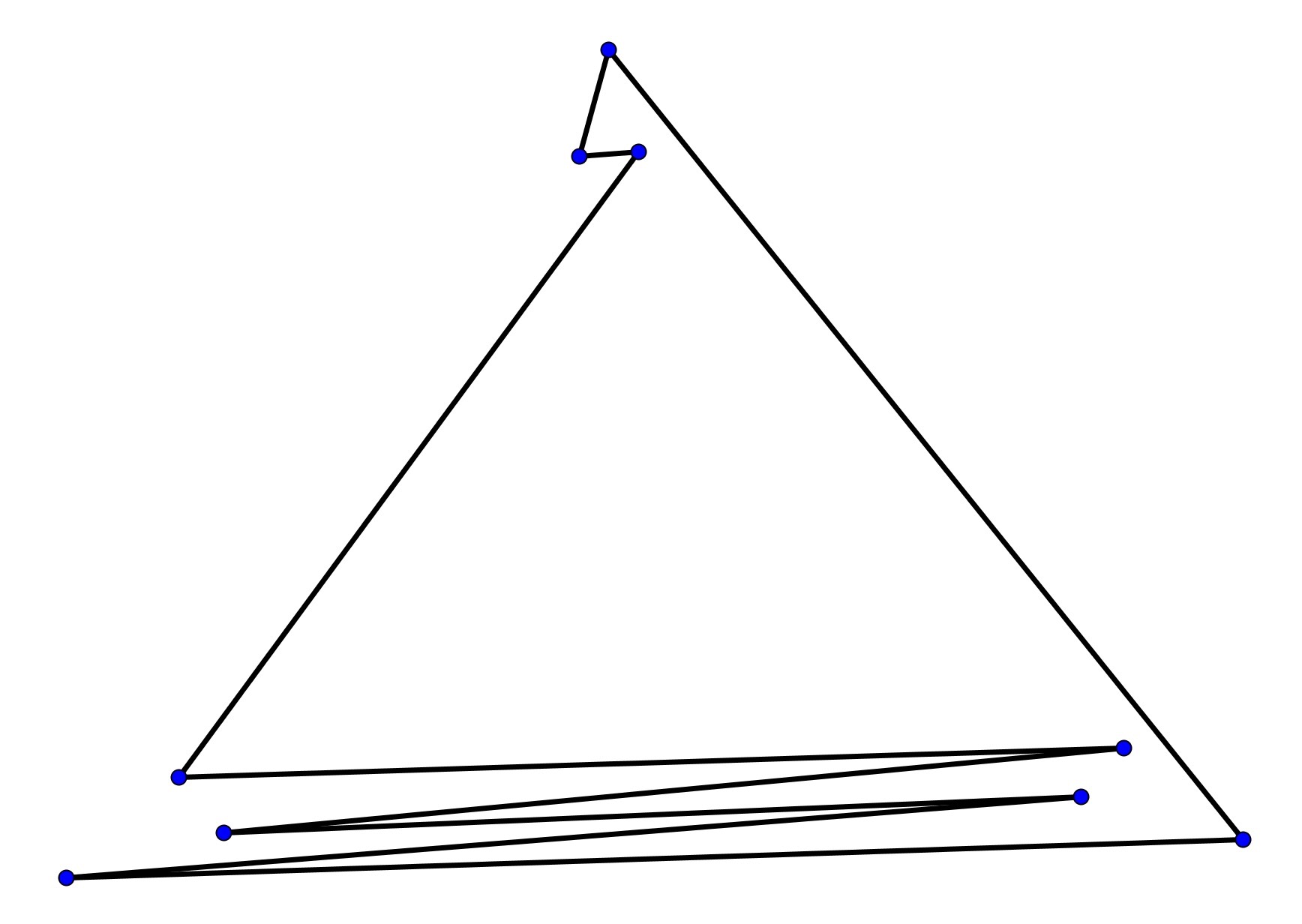}\\
    \end{tabular}
  \end{center}
  \vspace*{-12pt}
  \caption{The polygons $P_4$, $P_5$, $P_6$, together with  $P_1$, $P_2$, $P_3$, require five 6-universal guards for $\kuniguard{n}{6}$ when $n=9$.}
  \label{fig:lowerBound3UGP3}
\end{figure}
Then, by the same argument, the polygons $P_4$, $P_5$, $P_6$, together
with $P_1$, $P_2$, $P_3$, require 5 6-universal guards.  The
extensions are also similar since they are essentially symmetric to the
three polygons in Theorem~\ref{lem:lowerBounds3UGP}. This concludes the
proof.
\end{proof}

\subsection{Upper Bounds for $k$-Universal Guard Numbers}

We give non-trivial upper bounds for
$\kuniguard{n}{k}$ and $\kuniguardholes{n}{k}$, for all values $n,k
\in \mathbb{N}$. In particular, we provide algorithms that efficiently
compute guard sets for $\simplepolygons{S}$ and $\polygons{S}$ for any
given $S \in \points{n}$ and analyze the computed guard sets.
	
\begin{theorem}\label{thm:upperBoundskugp}
	$\kuniguard{n}{k} \leq \left( 1 - \left( \frac{2}{3} \right)^k \right)$.
\end{theorem}

	
	Hoffmann et al.~\cite{hoffmann:holes} showed $\kuniguardholes{n}{1} \leq \lfloor \frac{3n}{8} \rfloor$. Our approach implies for the traditional guard number $\kuniguardholes{n}{1} \leq \lfloor \frac{n}{2} \rfloor$.
	
	The following theorem shows that we can combine our approach with the method from~\cite{hoffmann:holes}.

\begin{theorem}\label{thm:upperBoundhugph}
	$\kuniguardholes{n}{k} \leq \left( 1 - \left( \frac{5}{8} \right)^k \right)n$
\end{theorem}


%% file: 05-other.tex
\section{Other Variants}
\label{sec:other}

In this section, we consider two variants of the Universal Art Gallery
Problem: the case in which guards are allowed to be placed only at
input points $S$ that are interior to the convex hull of $S$, and the
case in which the input set $S$ is a regular grid of points.  In both
cases we obtain tight bounds on the universal guard number.

\subsection{Interior Guards}

In the Interior Universal Guards Problem (IUGP) we allow guards to be
placed only at points of $S$ that are not convex hull vertices of $S$.
Note that placing guards at
{\em all} interior points is sufficient to guard any polygonalization
of $S$, since the $CH(S)$ vertices are convex vertices in any
polygonalization of $S$; it is a simple fact is that the reflex vertices
of any simple polygon see all of the polygon.  Our main result in this
section is a proof that it is sometimes necessary to place guards at
all interior points, in order to have a universal guard set.


\old{
\begin{theorem}\label{thm:iugp}
	$\uniguardinterior{n}{} = n - \Theta (1)$
\end{theorem}

\bigskip
Furthermore we get a stronger result.
}

\begin{theorem}\label{thm:UGPI}
The interior universal guard number satisfies 	$\uniguardinterior{n}{} = n - \Theta (1)$.
In particular, there exist configurations of $n$ points $S$, for arbitrarily large $n$, for which $CH(S)$ is a triangle,
and the only universal guard set using only interior guards is the set of {\em all} $n-3$ interior points.
\end{theorem}

\begin{proof}
Figure~\ref{fig:main} shows the structure of the construction.
The set $S$ consists of the following $n=9+3k$ points:
\begin{itemize}
\item $a,b,c$, which are the vertices of the convex hull of $S$; in the example in the figure, the triangle $\Delta abc$ is equilateral; 
\item three pairs of points, with $a_1,a_2$ very close to $a$, $b_1,b_2$ very close to $b$, and $c_1,c_2$ very close to $c$; these 6 points are in convex position;
\item three sets of $k$ points, with each set of points collinear, and the set of $3k$ points in convex position; denote the points
by  $p_1,\ldots,p_k$, $q_1,\ldots,q_k$, and $r_1,\ldots,r_k$, with the points indexed in order along the segments $p_1p_k$, $q_1q_k$, and $r_1r_k$.
\end{itemize}

In more details, the properties of the point configuration are as follows.
\begin{description}
\item[(1)] All of the points $p_i$ lie to the right of the oriented
  line through $aa_1$; similarly, points $q_i$ are to the right of
  $bb_1$ and points $r_i$ are to the right of $cc_1$.
\item[(2)] The line $ap_i$ passes between points $q_{k-i+1}$ and
  $q_{k-i+2}$.  (A similar statement holds for lines $bq_i$ and
  $cr_i$.)  The line $aq_i$ passes between points $p_{k-i+1}$ and
  $p_{k-i+2}$.  (A similar statement holds for lines $br_i$ and
  $cp_i$.) We call this the ``interleaving rays property''.  
See Figure~\ref{fig:main}, right.
\end{description}

In order to argue that such a configuration exists, for arbitrarily
large $k$ (and thus for arbitrarily large $n=9+3k$, we give a
procedure for placing the points $p_i, q_i, r_i$ along their
respective segments.
We begin with a placement of points with $k=2$, as shown, zoomed in,
in Figure~\ref{fig:main-zoom}.  (The point $q_1$ is shown collinear
with $a$ and $p_2$, and $r_2$ is shown collinear with $b$ and $q_1$;
however, the point $q_1$ is just to the left (by an arbitrarily mall
amount) of the oriented line $ap_2$, and $r_2$ is just to the left of
oriented line $bq_1$. Similarly, $q_2$ is just left of $ap_1$ and
$r_1$ is just left of $bq_2$, etc.)
Then, we claim that we can place new points $p$ between $p_1$ and
$p_2$, $q$ between $q_1$ and $q_2$, and $r$ between $r_1$ and $r_2$,
while preserving the interleaving rays property. See
Figure~\ref{fig:main-zoom}, right.  (The existence of such a point $p$
along the segment $p_1p_2$ follows from the intermediate value
theorem: as $p$ varies from $p_1$ to $p_2$, the corresponding position
of $r$ (on $r_1r_2$, just to the left of where line $bq$ intersects
$r_1r_2$, where $q$ is the point on $q_1q_2$, just to the left of
where line $ap$ intersects $q_1q_2$) varies from $r_1$ (which is below
$cp_1$) to the point $r_2$ (which is above $cp_2$).)  We then reindex
the points to be $p_1,p_2,p_3$, $q_1,q_2,q_3$, and $r_1,r_2,r_3$.
%
%
We then apply this argument recursively to place 2 new points in the 2
gaps (along segments $p_1p_2$, $p_2p_3$, $q_1q_2$, $q_2q_3$, $r_1r_2$,
and $r_2r_3$), and repeat, placing 4 new points in in 4 gaps, then 8
new points, etc.  Doing so allows the instance size to grow
(exponentially) with each iteration, showing that the construction
yields arbitrarily large instances.

We claim that every point of $S$ interior to the convex hull of $S$
must have a guard in any universal guarding that is not allowed to
place guards at the convex hull vertices ($a,b,c$).  To see this, we
show polygonalizations that would have some portion of the polygon
unguarded if not all interior points of $S$ were guarded. In
Figure~\ref{fig:poly} (left) we give a polygonalizaton of $S$ showing that
if $a_1$ is not guarded, then, even if all other interior points are
guarded, a portion of the polygon (shown in gray) is not seen.  In
Figure~\ref{fig:poly} (right), we give a polygonalizaton of $S$ showing
that if $p_i$ is not guarded, then, even if all other interior points
are guarded, a portion of the polygon (shown in gray) is not seen.
\end{proof}

\begin{figure}[htbp]
   \centering
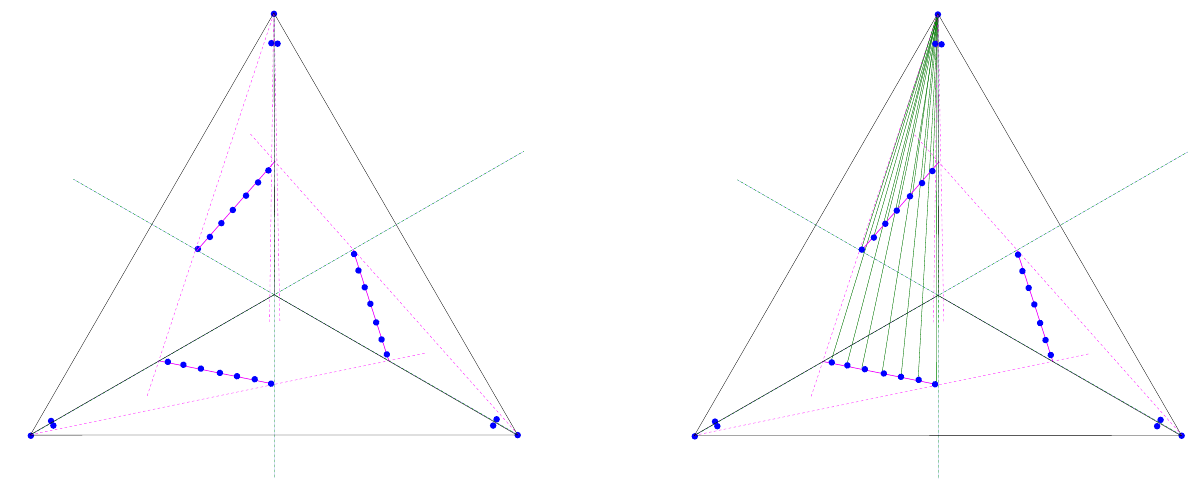
   \caption{The construction of the instance showing that for some input sets $S$ of $n=9+3k$ points, if guards are not allowed to be placed at convex hull vertices, then all interior points of $S$ may be required to be in a universal guard set.}
\label{fig:main}
\end{figure}

\begin{figure}[htbp]
   \centering
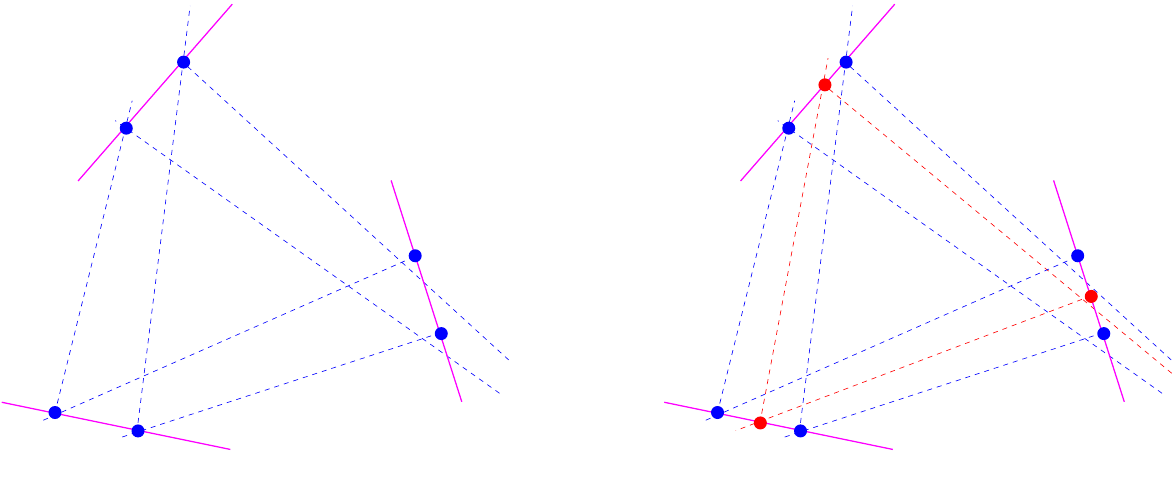
   \caption{A zoomed-in view of the construction in Figure~\ref{fig:main}.  Left: Placement of $k=2$ points on each of the three segments, in order that the interleaving rays property holds for this small instance. Right: Addition of the intermediate points $p,q,r$ on the three segments, while preserving the interleaving rays property.}
\label{fig:main-zoom}
\end{figure}

\begin{figure}[htbp]
   \centering
 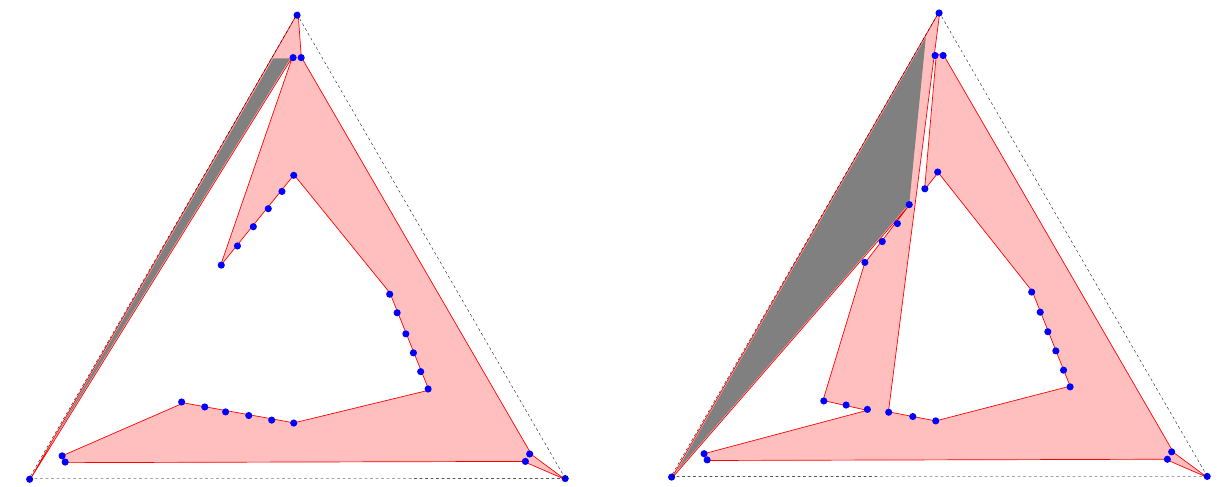
   \caption{Left: A polygonalizaton of $S$ showing that if $a_1$ is not guarded, then, even if all other interior points are guarded, a portion (in gray) of the polygon is not seen. Right: A polygonalizaton of $S$ showing that if $p_i$ is not guarded, then, even if all other interior points are guarded, a portion (in gray) of the polygon is not seen.}
\label{fig:poly}
\end{figure}

We remark that the configuration of points $S$ given in the proof
above can be universally guarded with approximately $n/2$ guards, if
we permit guards at the three convex hull vertices: With guards at
$a,b,c$, and at the 6 points $a_1,a_2,b_1,b_2,c_1,c_2$, we need only
to place guards at every other point in the collinear sequences
$p_1,p_2,\ldots,p_k$, $q_1,q_2,\ldots,q_k$, and $r_1,r_2,\ldots,r_k$,
in order to guard $S$ universally, as one can readily check.  Thus,
the reason so many guards ($|S|-3$) were needed was because of the
requirement to avoid guarding the convex hull vertices.

\old{ 
\subsubsection{Configuration Construction}

%
We say that an ordered triple $(a,b,c)$ of three points $a, b, c \in
S$ form a {\em spike} if there exists a polygonalization of $S$ that
includes the edges $ab$ and $bc$ and such that placing guards at all
of the vertices $S\setminus \{a,b,c\}$, and not at the vertices
$\{a,b,c\}$, causes some portion of the triangle $\triangle abc$ (in
particular, the point $b$) not to be seen (within the polygon) by the
guards.
\old{ prior definition:
We say that three points $a, b, c \in S$ form a {\em spike}
if there exists a subset $S'\subseteq S$ with $a,b,c\in S'$ and a
simple polygonal chain, $\pi$, having vertex set $S'$ such that not
all of $\triangle abc$ is seen by the points $S\setminus \{a,b,c\}$
when treating $\pi$ as a set of opaque edges.  }
Refer to Figure~\ref{fig:spike}.
The concept of a spike is related to that of a chamber introduced earlier.

A point set $S$ is said to be in a {\em safe configuration with respect to
  spikes} if no 3 points of $S$ form a spike. 

\begin{figure}[htbp]
   \centering
\includegraphics[scale=0.25]{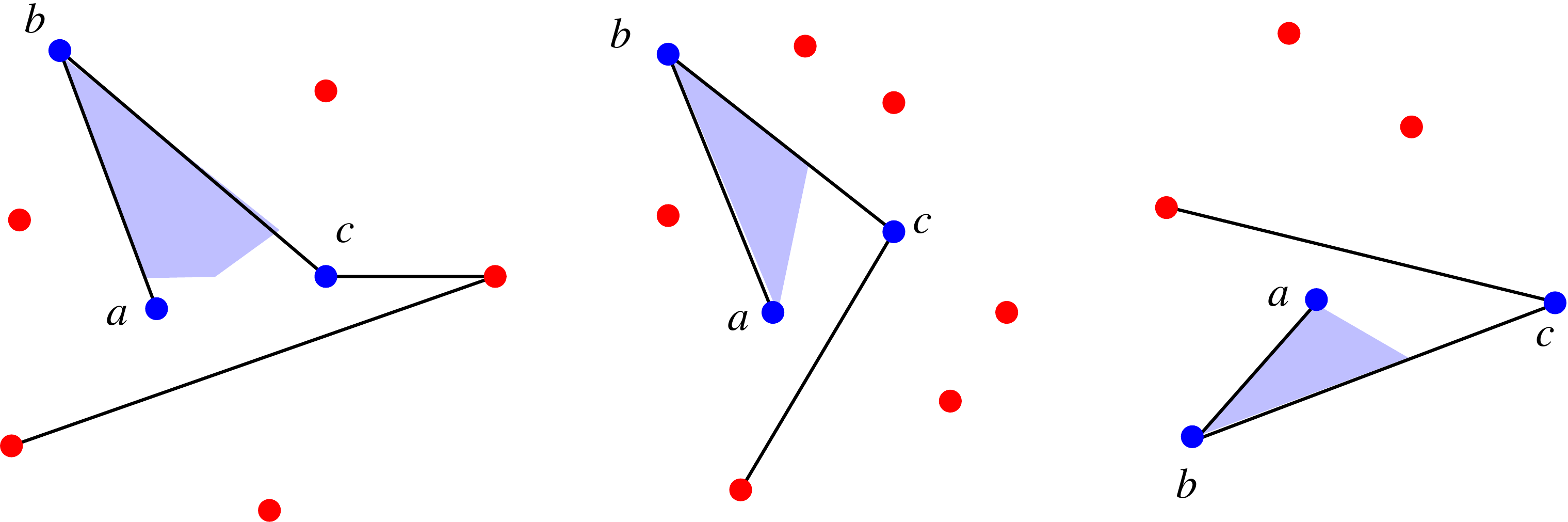}
   \caption{Examples of spikes on $a,b,c$; guards at red points fail to see all of $\triangle abc$.}
\label{fig:spike}
\end{figure}



I think I have a simplification of the construction; it still has the 3 groups of points on 3 segments (that are in a ``Godfried'' like configuration), but it does not need
the original 6 interior points, and the construction is completely symmetric (3-way rotational symmetry).
See attached.
I show just 3 points on each of the green segments, but imagine that there are about n/3 on each.  The points are evenly spaced, with a gap of $\delta$
between any two consecutive points.  There is a point at one end of a green segment, then evenly spaced (spacing $\delta$), leaving spacing $\delta/2$
at the other end of the segment.
This spacing makes it possible to draw a spike if any one point on the interior is not guarded.
I am writing down the exact formulas for the locations of the points now, so that the proof will have the algebra to verify correctness.
Also, there will be a figure and a proof to show that there is a full polygonalization consistent with each spike.
(I modified the definition of spike: see below)

Note that this modified construction also has the nice property that the points have just 2 layers: the CH has 3 vertices, and the other n-3 points are all in convex position
(on the 3 green segments)..

\begin{figure}[htbp]
   \centering
\includegraphics[scale=0.35]{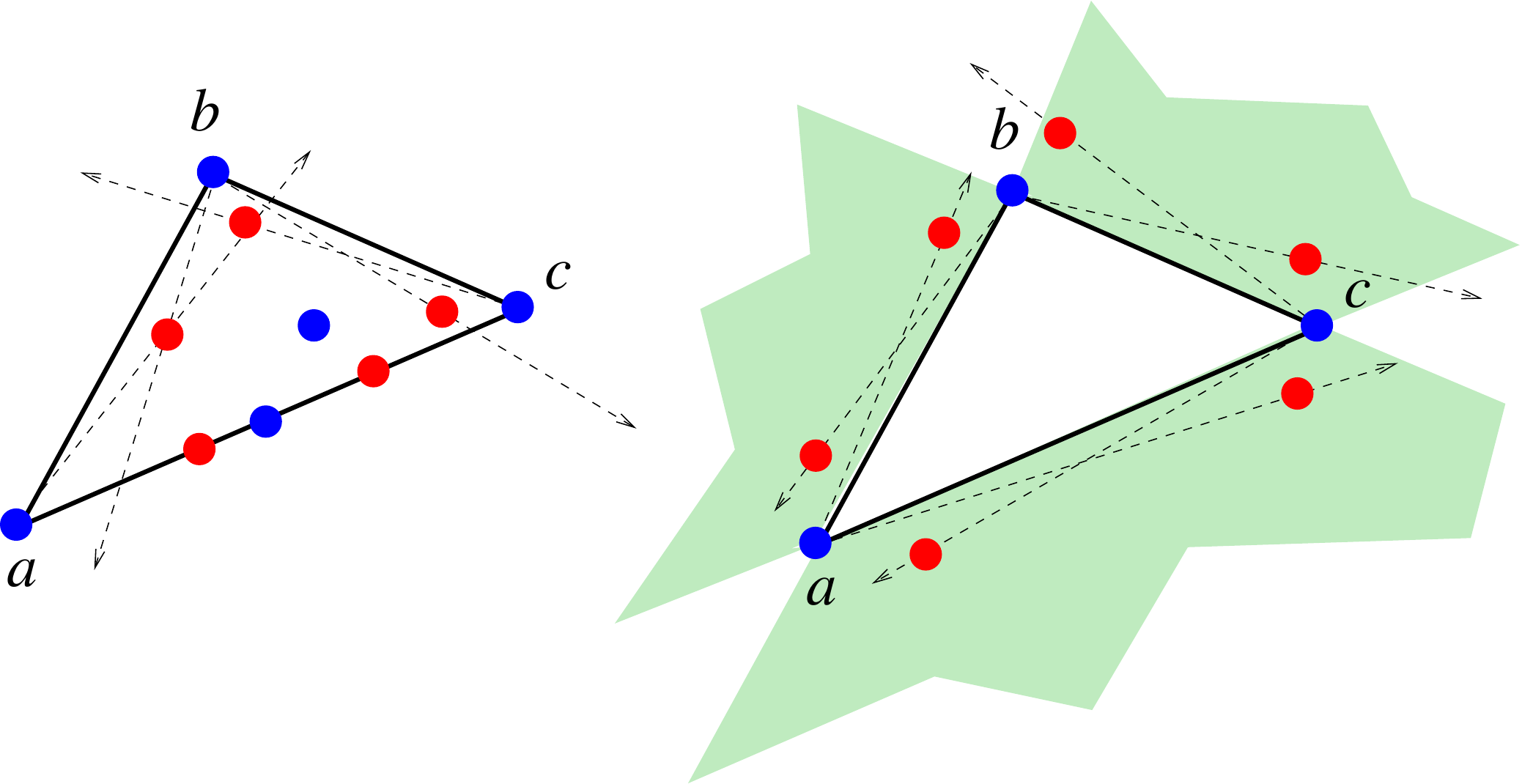}
   \caption{Safe conditions: Rules 1 and 2.}
\label{fig:safe}
\end{figure}

We say an input point $p \in S$ is {\em guarded} if we place a guard at $p$, otherwise it is {\em unguarded}. For three unguarded points $a, b, c\in S$ we say that they {\em
  satisfy the safe condition} if they satisfy either one of the
following rules (refer to Figure~\ref{fig:safe}):

{\em Rule 1:} There are points inside (or on the boundary of) $\triangle abc$,
  and within $\triangle abc$  
a ray with apex in $\{a,b,c\}$ rotated inwards, starting from each incident edge to the apex,
hits a guarded interior point before hitting an unguarded point. 

{\em Rule 2:} There is no point of $S$ inside (or on the boundary of)
  $\triangle abc$, and, further, a ray with apex in $\{a,b,c\}$
  rotated outwards, starting from each incident edge to the apex hits
  a guarded point that is within the corresponding ``wedge'' (shown in
  green in the figure), before hitting an unguarded point.

A key fact is the equivalence:

\begin{lemma}\label{lem:key}
A point set $S$ with guards at $G\subseteq S$ is in safe configuration with respect to spikes if and
only if any three unguarded points of it satisfy the safe conditions.
\end{lemma}

\begin{proof}

If any three unguarded points $a, b, c$ satisfy the safe
conditions, we want to show that no spike can be created on this point set,
which means by satisfying the safe conditions, every empty triangle is guarded.
An empty triangle is a triangle on three unguarded vertices. So consider the
smallest empty triangle units, which are the empty triangles with no other
unguarded point inside of it. If any of these triangles has points inside of it
or on its boundary, by Rule 1 of the safe conditions, we know those points are
guards, and they guard that triangle. This leaves the case when there is no
point lie inside nor on the boundary of the triangle. By Rule 2 of the safe
conditions, the first hitting points outside this triangle are guarded points
inside the corresponding wedges. Below we prove that no spike can form with
this triangle. There are two cases:

\begin{quote}
Case 1: If there is no
segment connecting two points that can block the hitting guards from guarding
the triangle, we are done. These hitting guards can always guard the triangle,
no spike can form.
\end{quote}

\begin{quote}
Case 2: If there are some
segments connecting two points that can block the hitting guards from guarding the
triangle, there are several subcases:

\begin{quote}
Subcase 2.1: If there
are two hitting points of any edge.

\begin{itemize}
\item {Subcase 2.1.1}: Refer to
Figure \ref{5}, two other segments can block guards $1, 2$ from guarding the
$\triangle abc$, and the two segments have a common endpoint $d$ inside the
wedge that contains guards $1,2$.

 If $d$ is an unguarded
point, $\triangle abd$ will violate the safe condition, which is a
contradiction. So at least one of $a, b, d$ has to be a guard, since $a, b, c$
are not guards, $d$ has to be a guard, and in this case $d$ can guard
$\triangle abc$. 

\item{Subcase 2.1.2}:
Refer to Figure \ref{6}, two other segments can block guards $1, 2$ from
guarding the $\triangle abc$, and the two segments do not meet at a common
endpoint inside the wedge containing guards $1,2$. 

  By a similar argument,
among the above endpoints inside the wedge, there exists a point that is closer
to $b$ and can be connected directly to $b$, call this point $d$.  By
contradiction, if $d$ is an unguarded point, consider $\triangle abd$, 
which violates the safe condition, a contradiction. So as before,
$d$ has to be a guarded point and it can guard $\triangle abc$.

\begin{figure}[htbp]
   \begin{minipage}[t]{0.42\textwidth}
\centering
   \includegraphics[height=3.5cm]{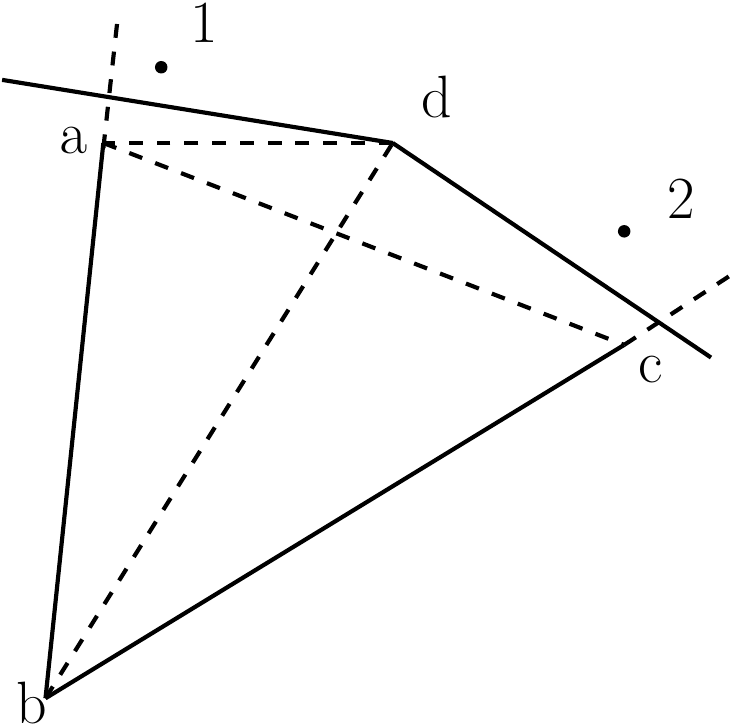}
\caption{Subcase 2.1.1, two blocking segments meet at a common point.}
\label{5}
   \end{minipage}
   \begin{minipage}[t]{0.42\textwidth}
\centering
   \includegraphics[height=3.5cm]{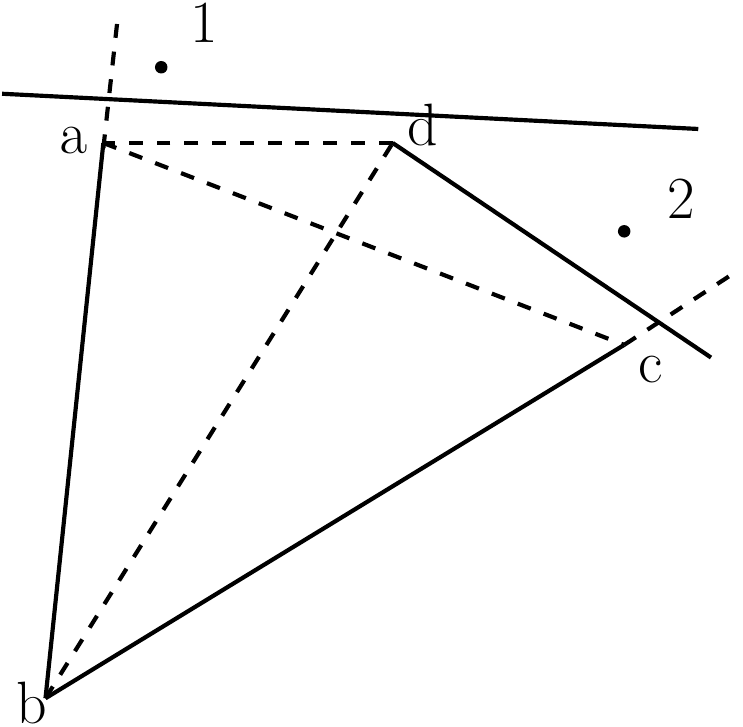}
\caption{Subcase 2.1.2, two blocking segments do not meet at a common point.}
\label{6}
 \end{minipage}

\end{figure}

\item{Subcase 2.1.3}:
A single segment cannot block both hitting guards from guarding the $\triangle
abc$ at the same time. So if only one hitting guard can be blocked, there is
still a hitting guard left such that no spike can form with the $\triangle
abc$. 

\end{itemize}
\end{quote}

\begin{quote}
Subcase 2.2: There is
only one hitting point of any edge.

 Refer to Figure \ref{7}, In this case,
one single segment cannot block the only hitting point $1$ from guarding the
$\triangle abc$. This is because in order for the blocking segment to block the
guard $1$ from guarding the $\triangle abc$, the blocking segment has to have
at least one of its endpoints inside the shaded region, we cannot have a
segment with both endpoints lying outside the shaded region block the guard
$1$. However, by the first hitting point rule, no other point can exist in the
shaded region.  Hence no blocking segment can block this hitting guard $1$ to
guard the $\triangle abc$, so  the $\triangle abc$ will always be guarded and
no spike can form with it.

\begin{figure}[htbp]
\begin{minipage}[t]{0.4\textwidth}
   \centering
\includegraphics[height=3.5cm]{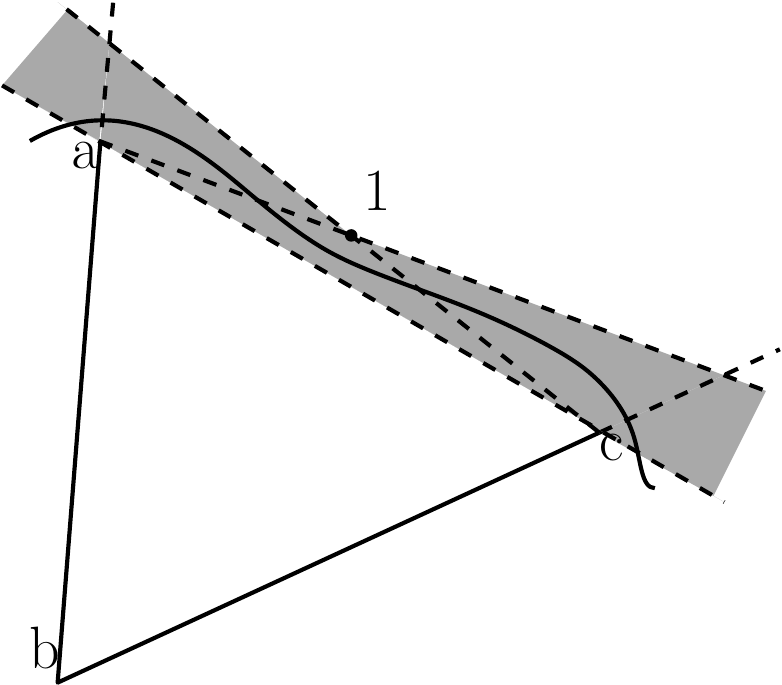}
   \caption{Subcase 2.2,
there is only one hitting point along a edge.}  
\label{7}
\end{minipage}
 \begin{minipage}[t]{0.4\textwidth}
   \centering
\includegraphics[height=3.5cm]{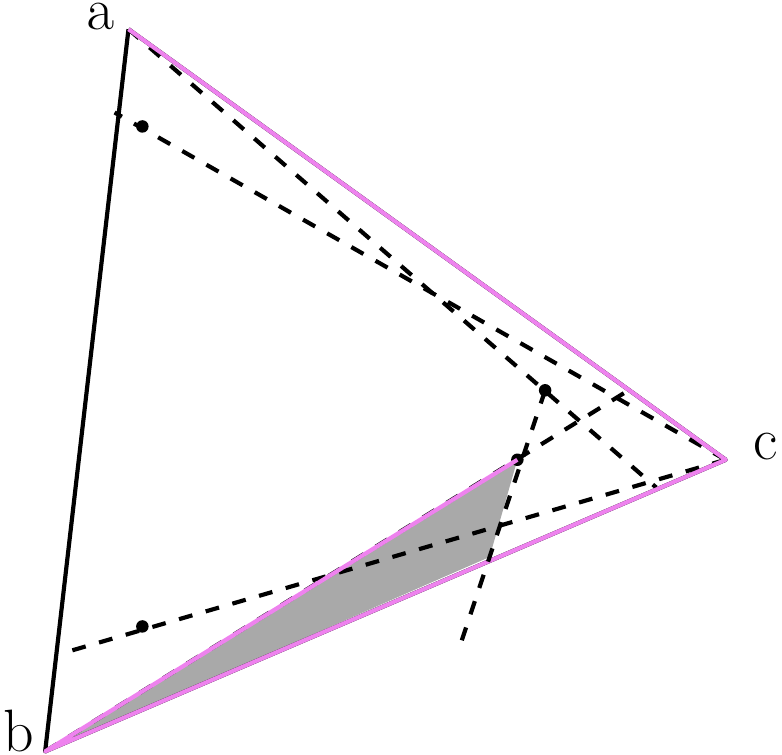}
      \caption{Spike
can form when violating the safe condition Rule 1.}
   \label{111}
\end{minipage}

\end{figure}

\end{quote}

\end{quote}

\vspace{1em}
 For the converse,
assume three points do not satisfy the safe
conditions. Then we can create a spike, which contradicts that the three points
are in the safe configuration with respect to spike. More specifically, if it
satisfies Rule 1 and there is one hitting point that is not guarded, then
using that point and connect the corresponding segment on the ray through it
with the edge that ray starts to rotate from will create a spike with no guard
can see it. If it satisfies Rule 2, and there is either a hitting point
that is unguarded or outside the corresponding wedge, similar to the former
situation, using the wedge and just connect the anchored point to that point
will give us a spike. Refer to Figure \ref{111}, \ref{112}. We can easily
connect the magenta line segments to form a spike, while the shaded areas
are the unguarded regions.

\begin{figure}[htbp]
\begin{minipage}[t]{0.44\textwidth}
   \centering
\includegraphics[height=3.5cm]{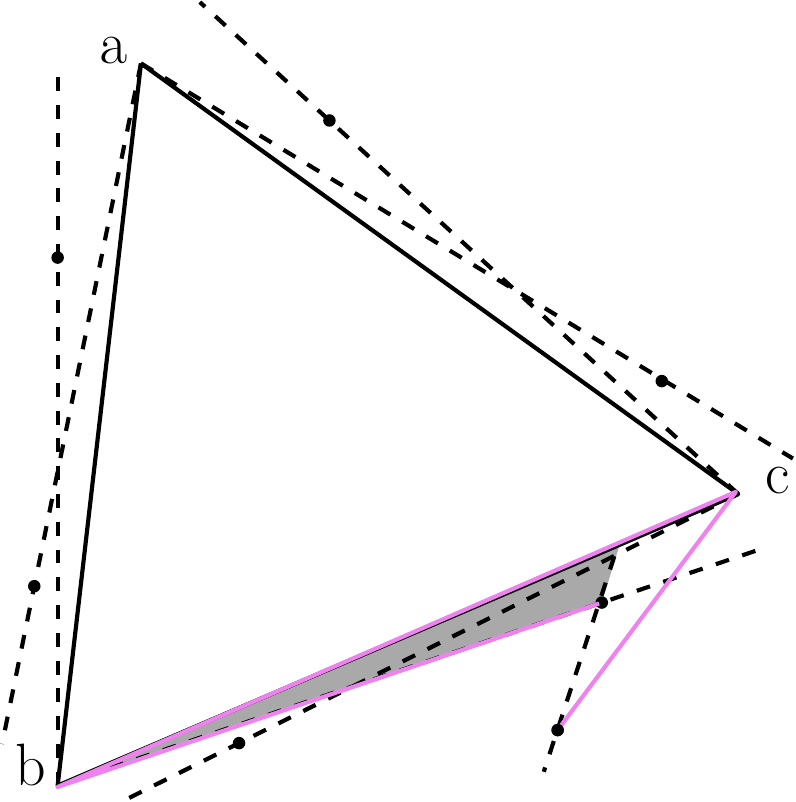}
   \end{minipage}
\begin{minipage}[t]{0.44\textwidth}
   \centering
\includegraphics[height=3.5cm]{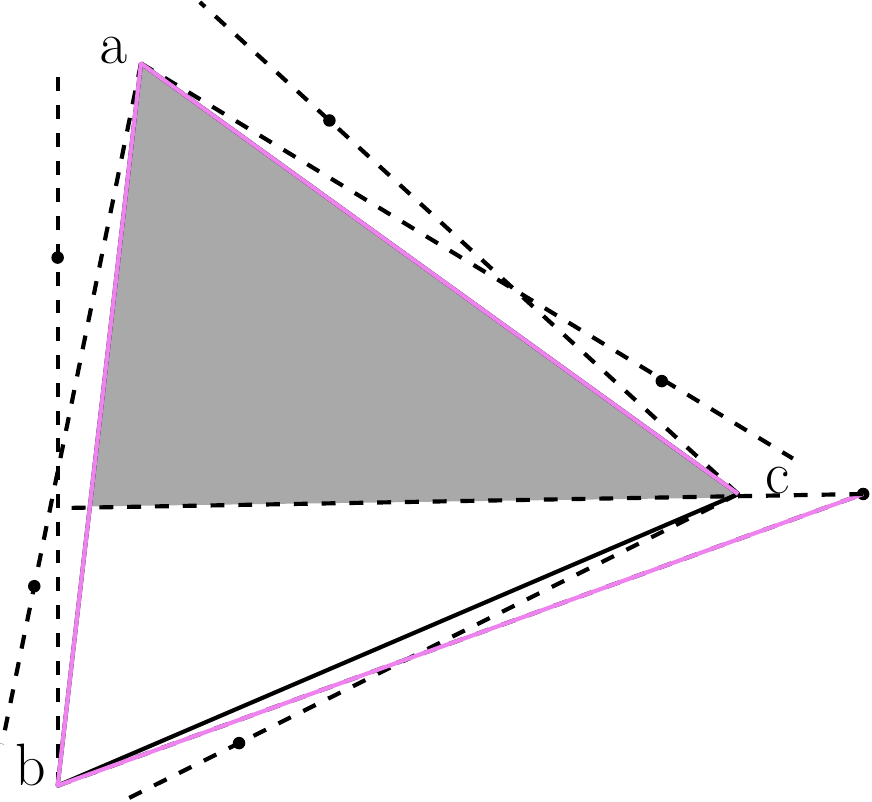}
   \end{minipage}
\caption{(a)-(b) A spike can form when violating the safe condition Rule 2.}
\label{112}
\end{figure}

\end{proof}

{\em Configuration Construction:} We utilize Lemma~\ref{lem:key} 
and construct a careful
configuration of points whose general structure is shown in
Figure~\ref{fig:config}: The points $a,b,c\in S$ are the vertices of
$CH(S)$.  Six additional points (in red) are placed just inside each
edge of $\triangle abc$, so that each is first hit by rays rotating
inwards from the edges of $\triangle abc$.  Then, carefully
located points are placed (in a sequence of ``rounds'', will be 
explained later) along each of
three line segments (thick green in the figure), in such a way that
all of these interior points must be guarded in order to avoid a spike
(created by the unguarded point, together with two vertices of
$\triangle abc$). (Each of the potential spikes is such that, in this
configuration $S$, we can argue that there exists a polygonalization
of $S$ that includes the spike.)

\begin{figure}[htbp]
   \centering
\includegraphics[scale=0.45]{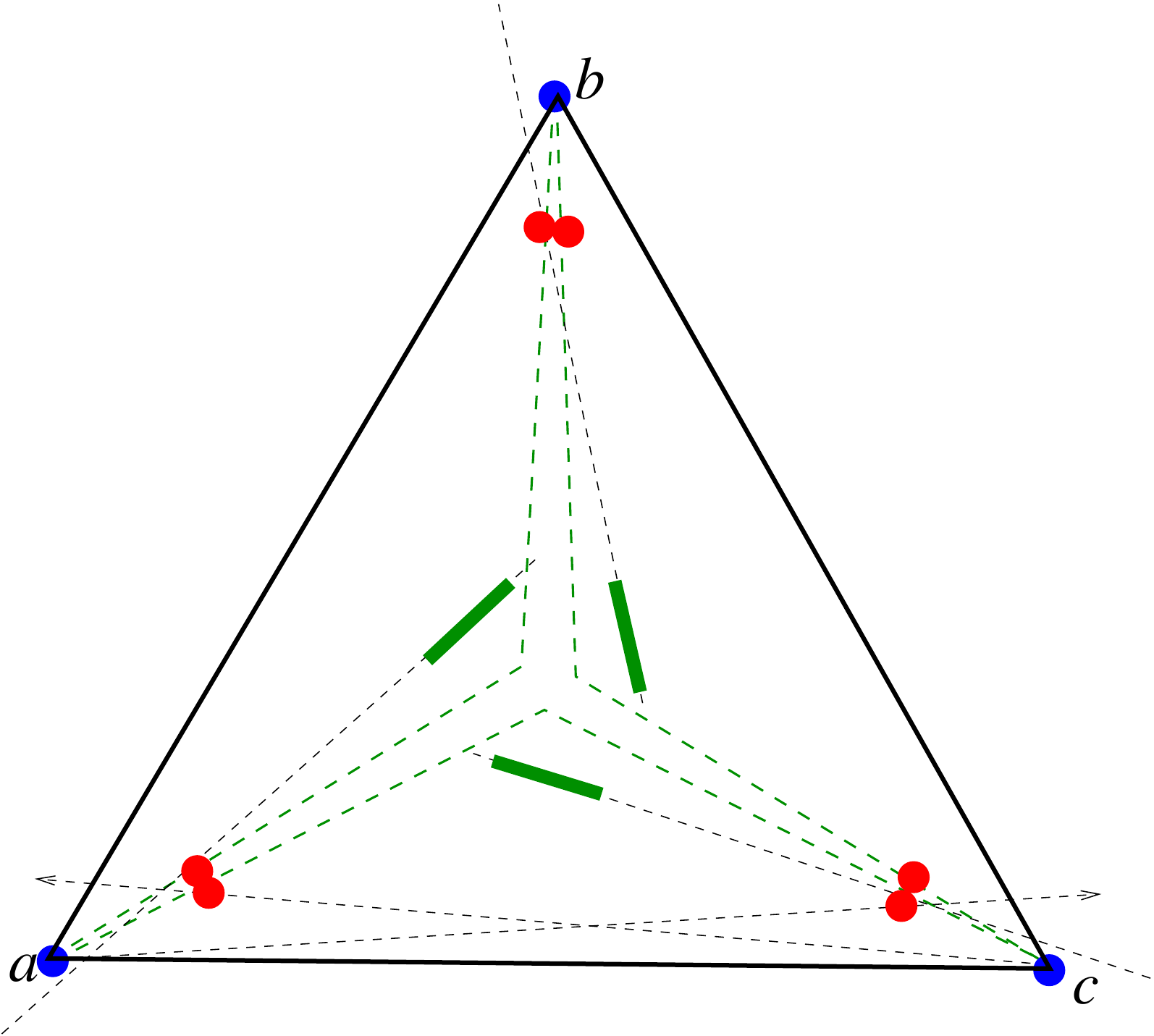}
   \caption{The overall configuration for the proof of Theorem~\ref{thm:UGPI}.}
\label{fig:config}
\end{figure}

Given a triangle, first we place six points as the first hitting points for each triangle edge, and they must be guarded, otherwise a polygonalization with spike can be drawn easily. Next connect each triangle vertex to the first hitting points about its adjacent edges, they will intersect and create three regions, call these three regions the three ''pockets'' with respect to each edge, refer to the shaded regions in Figure \ref{8}. 

 Our goal is to keep adding points, while requiring them to be guarded. Note that we can only add such points in the three pockets. To do so, we make the points we add to violate Rule 2 of the safe condition as explained above. Start from the easiest, refer to Figure \ref{8}, we can add one more points $p, q, r$ in each pocket such that point $q$ must be guarded because the existence of point $p$ made it violates RULE 2, and a polygonalization with spike can be easily drawn if violated. Similarly, point $r$ must be guarded because of point $q$, and point $p$ must be guarded because of point $r$. Such triple points $p, q, r$ are necessary for each other to be guarded. 

\begin{figure}[htbp]
   \centering
   \includegraphics[height=3.5cm]{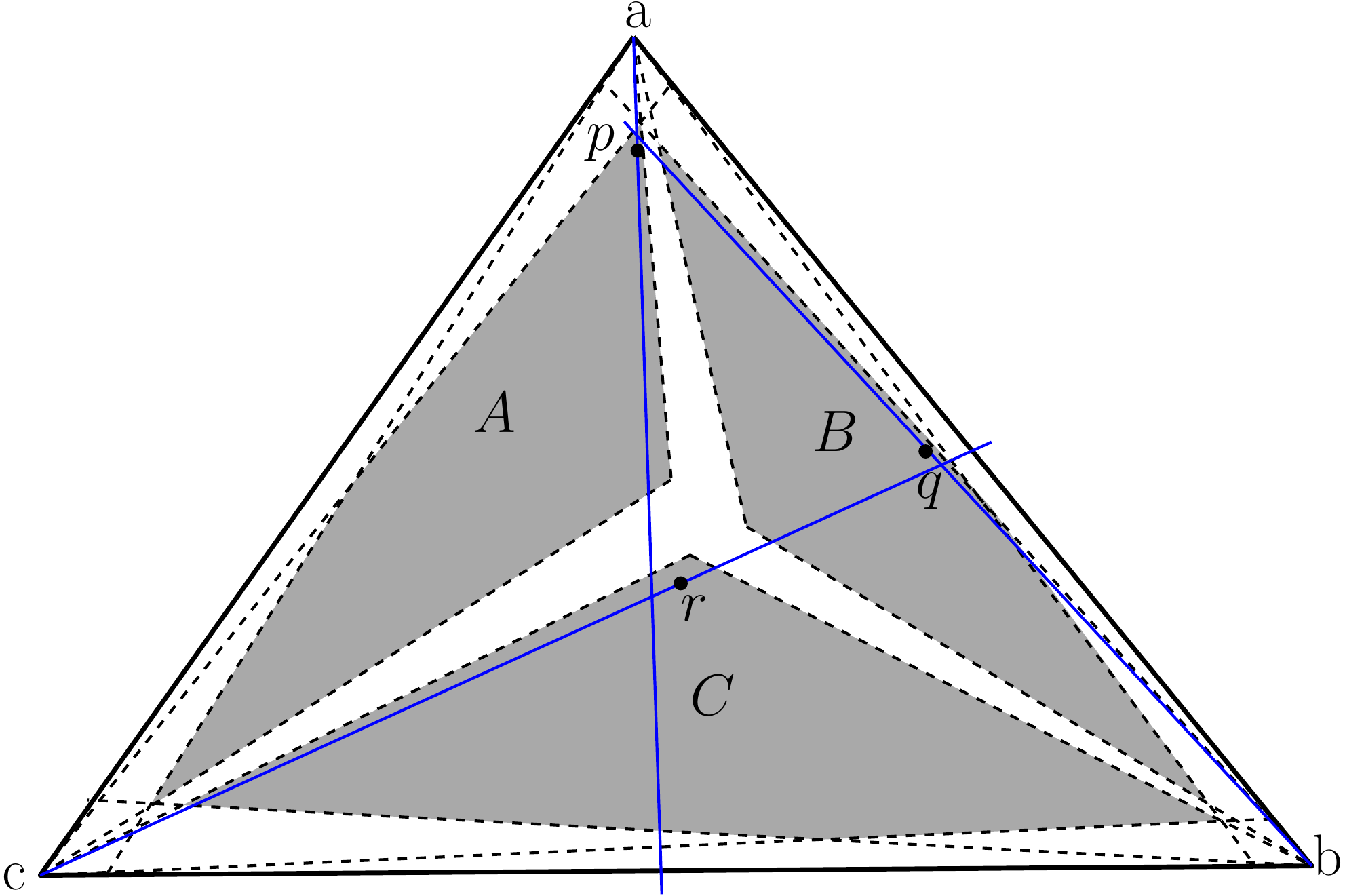}
   \caption{The shaded areas $A, B, C$ in the triangle are the three pockets where we can put more points; in this graph, we put in three points $p, q, r$ needed to be guarded.}
   \label{8}
\end{figure}

When more points are added, the points in the same pocket should maintain a property that each of them form an empty triangle with the corresponding convex hull edge. Notice that both points $p$ and $p'$ form empty triangles with the hull edge means that if we already placed point $p$ in one pocket, connecting point $p$ respectively to the endpoints of the hull edge will create a double wedge, we can only put point $p'$ in the double wedge which do not contain the vertical line. Shown in shaded area in Figure \ref{9}. For simplicity we can always place all the points in one pocket on a straight line ( (drawn in red in later figures)).

\begin{figure}[htbp]
   \centering
   \includegraphics[scale=0.18]{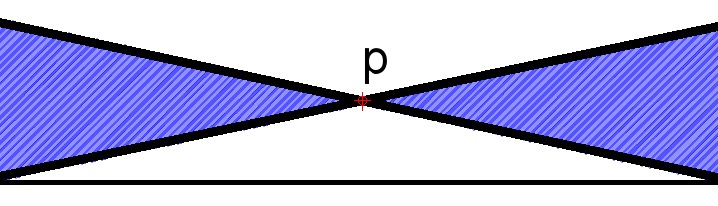}
   \caption{Once we put a point $p$ in the pocket, we can only put the next point in the shaded double wedge created by that point.}
   \label{9}
\end{figure}

Refer to Figure \ref{12}. We add five points in each pocket $A, B$ and
$C$ inside $\triangle abc$, each point has to be guarded because of
the existence of other points. We can add $2^0, 2^1, 2^2, 2^3, 2^4,
2^5, ...$ points in each pocket in this graph at each round , and all
of these points must be guarded, otherwise a polygonalization with a
spike can exist. 

\begin{figure}[htbp]
  \begin{minipage}[t]{0.48\textwidth}
  \centering
  \includegraphics[scale=0.1]{figure10(b).png}
  \end{minipage}
  \begin{minipage}[t]{0.48\textwidth}
  \centering
  \includegraphics[scale=0.1]{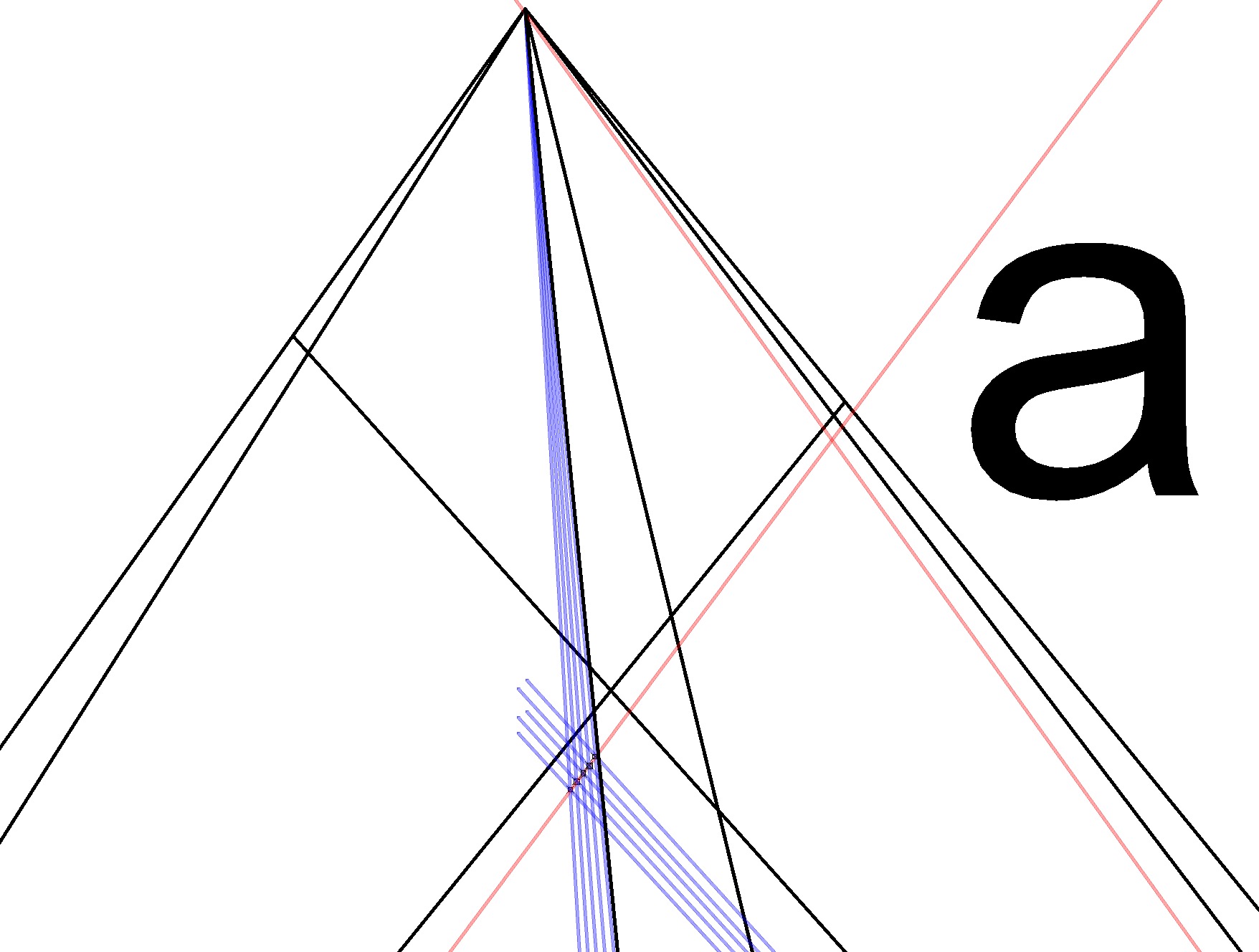}
  \end{minipage}
\ \\
  \begin{minipage}[t]{0.48\textwidth}
  \centering
  \includegraphics[scale=0.1]{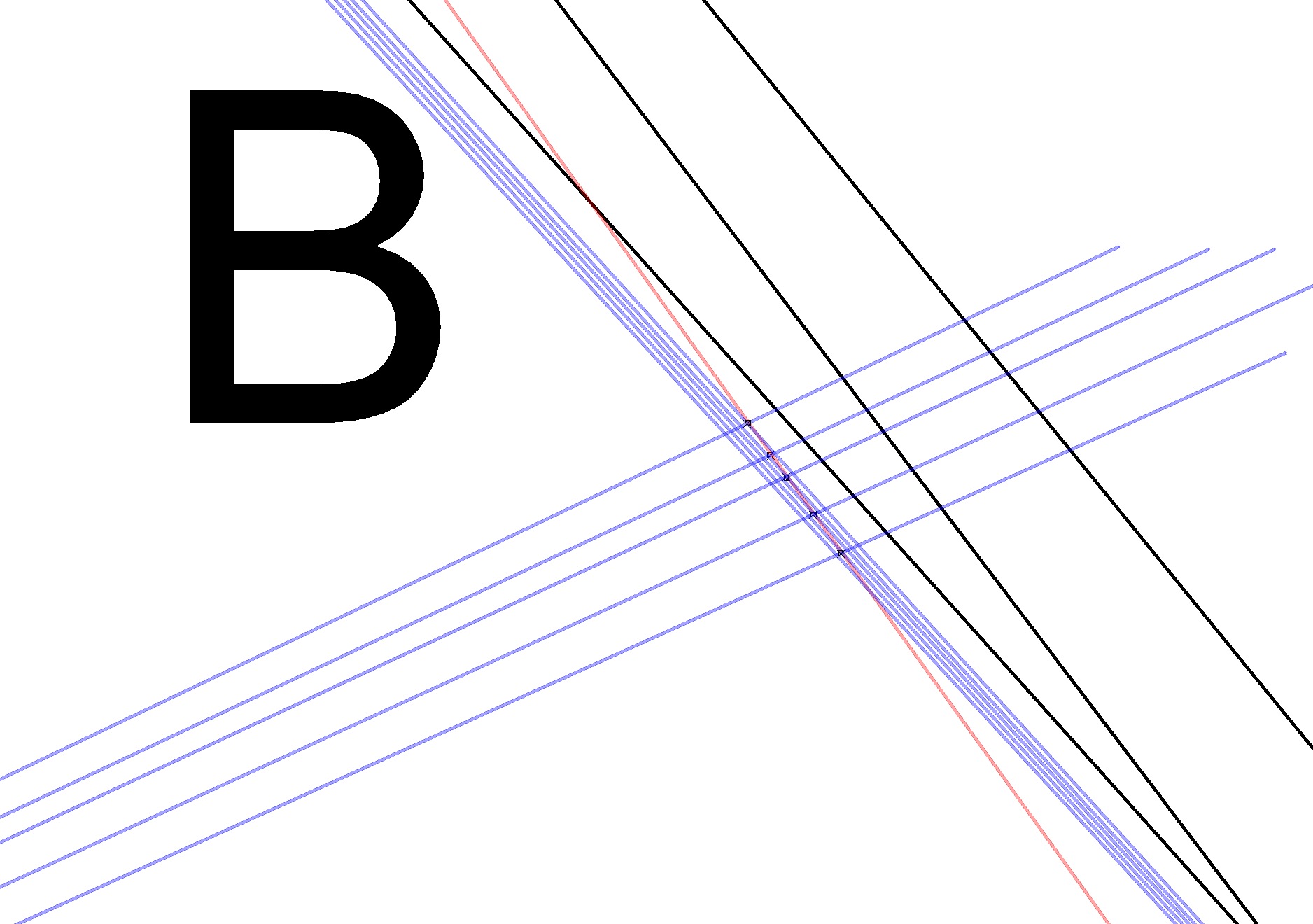}
  \end{minipage}
  \begin{minipage}[t]{0.48\textwidth}
  \centering
  \includegraphics[scale=0.1]{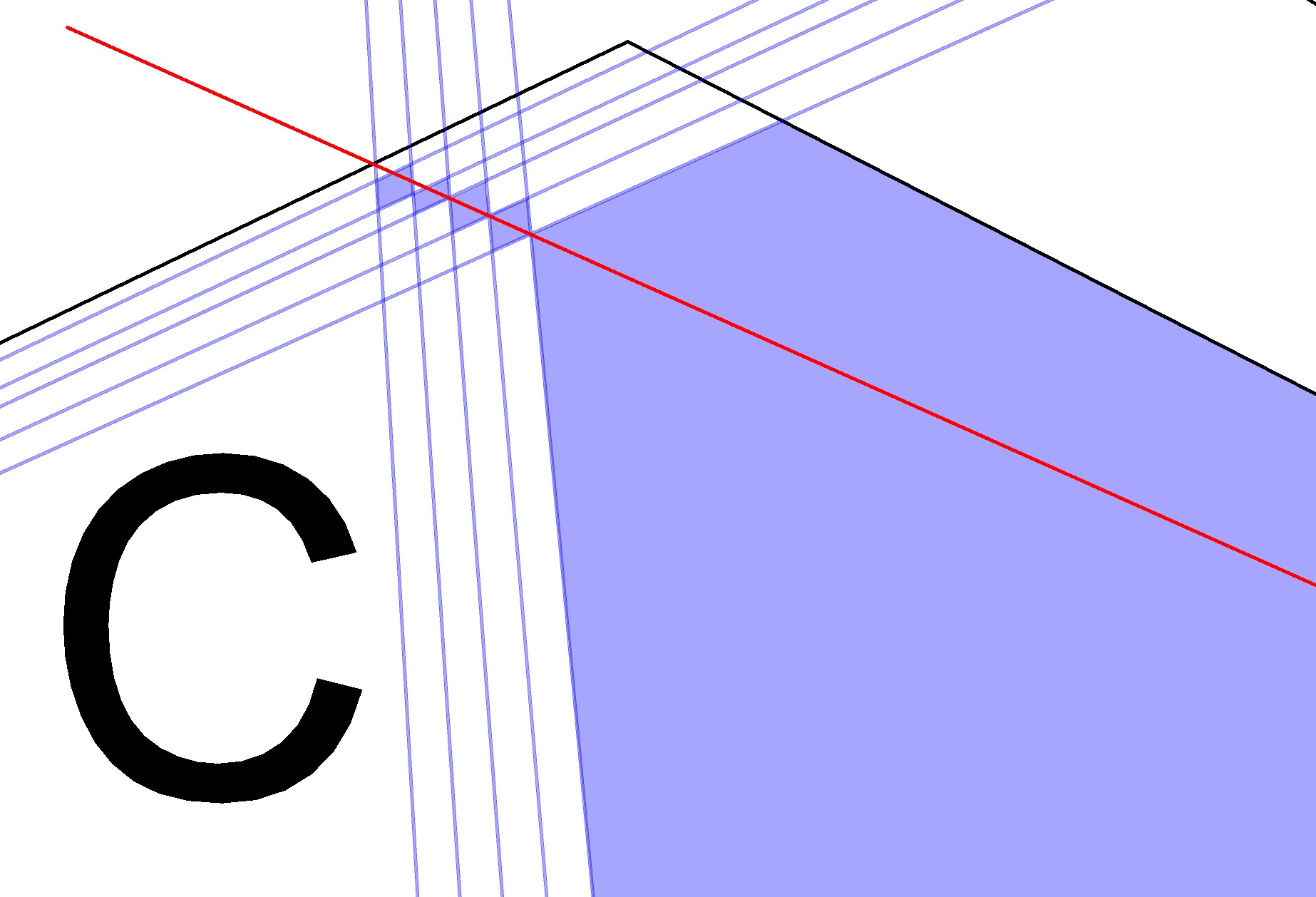}   
  \end{minipage}
   \caption{(a)-(d). The first figure shows an example of five points in each
pocket; the following three figures shows the amplified corners of $a, b$ and
$c$. The blue shaded regions are the candidate regions for choosing the five
points in pocket $C$. In this graph, we choose to put points on the red lines.}
\label{12}
\end{figure}

The graph has such properties: 
\begin{quote} 
1. The lines in pocket $A$ connect the vertex $a$ and all the added points, refer to Figure \ref{12}(b) .

2. The lines in pocket $B$ connect vertex $b$ and all the points in $A$, which create many stripes in pocket $B$. The points in pocket $B$ need to lie in these stripe (and if it's the highest point, it needs to lie above the first line). Without loss of generality, starting from the third point, we assume it's just below (say $\epsilon$, with $\epsilon \to 0$) the intersection of the red line and the upper stripe line the points should lie in, refer to Figure \ref{12}(c).

3. We always choose points on the red lines, and the chosen points satisfy that each of them form an empty triangle with the corresponding hull edge. In pocket $C$, the points are chosen one in every shaded regions on the red line.
\end{quote}

Our method to add more points in this graph is basically as follows.
\begin{quote}
Each round add $1, 2, 2^2, 2^3, 2^4...$ points on the red line containing existing points in each pocket, we add point(s) in $A$ (between the third and fourth lines in Figure \ref{12}) first, then find position to add corresponding points in pocket $B$ and $C$ as follows.

\begin{figure}[htbp]
  \begin{minipage}[t]{0.48\textwidth}
  \centering
  \includegraphics[scale=0.15]{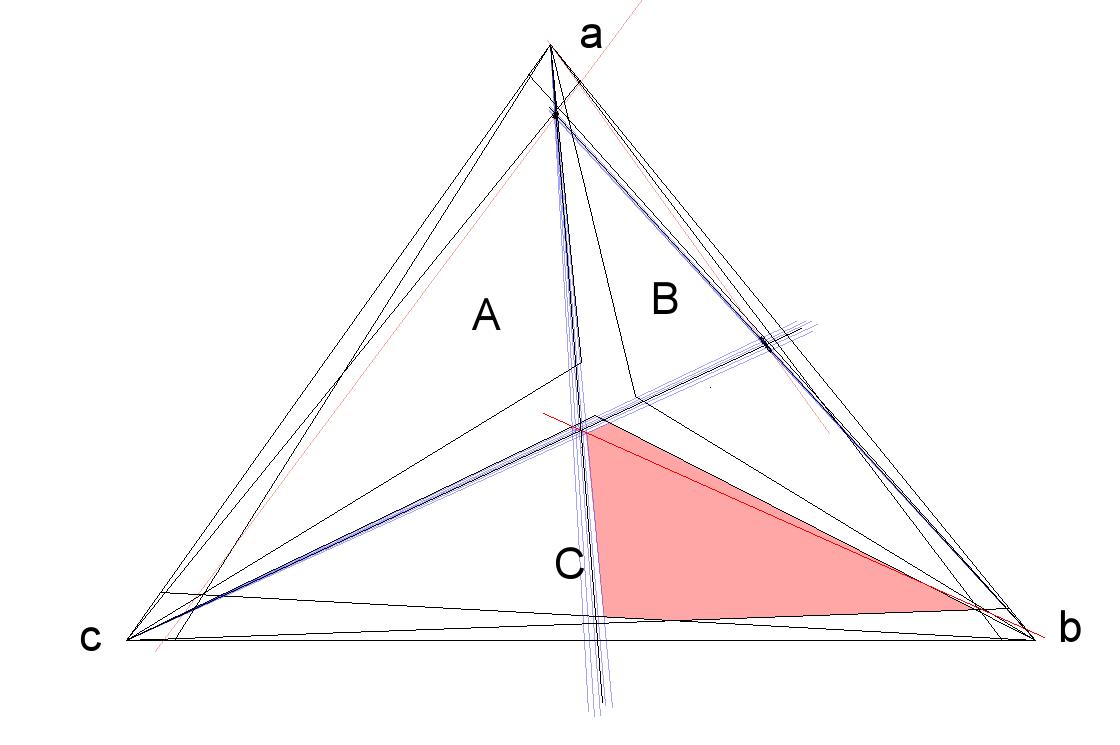}
  \end{minipage}
  \begin{minipage}[t]{0.48\textwidth}
  \centering
  \includegraphics[scale=0.15]{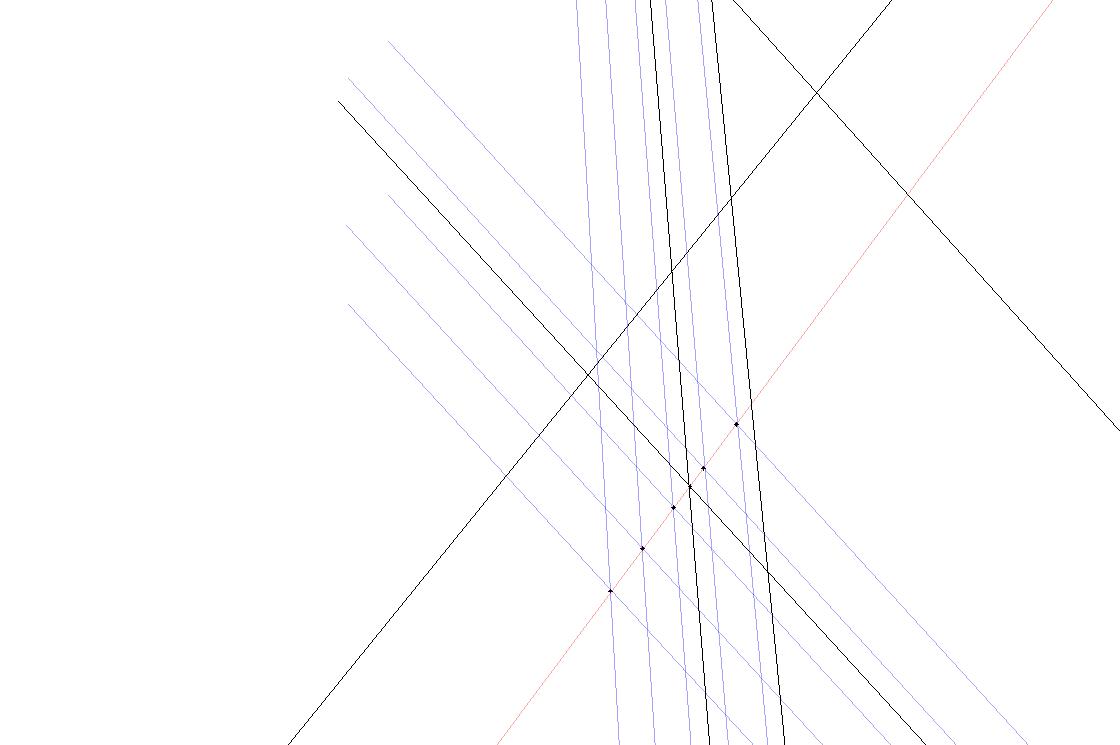}
  \end{minipage}
\vspace{2mm}
  \begin{minipage}[t]{0.48\textwidth}
  \centering
  \includegraphics[scale=0.15]{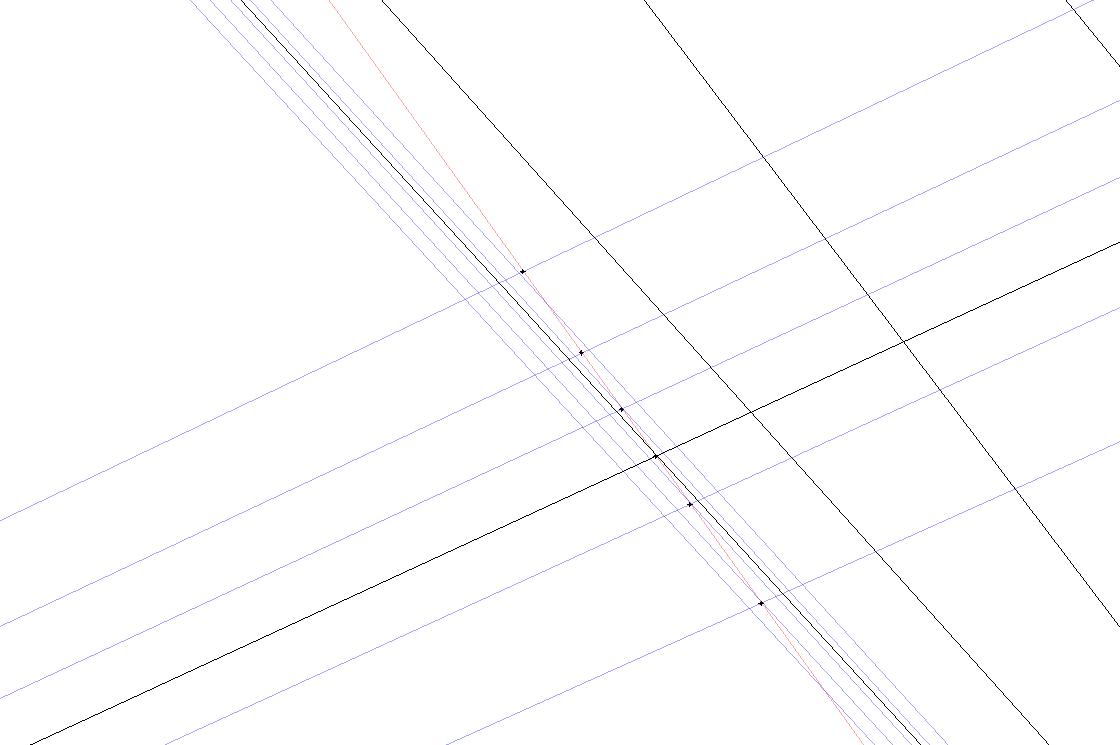}
  \end{minipage}
  \begin{minipage}[t]{0.48\textwidth}
  \centering
  \includegraphics[scale=0.15]{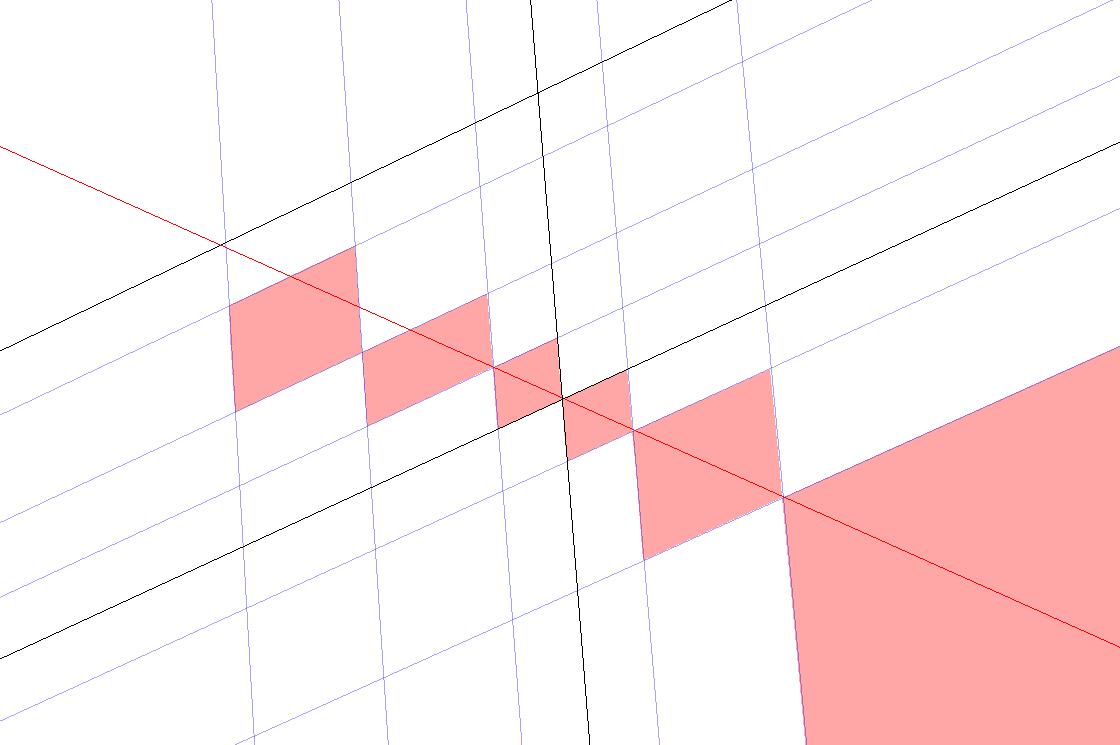}   
   \end{minipage}
   \caption{The first figure shows the corresponding changes in region B and C when we add one point in A on the black line.
 The following three figures show the amplified corners of $a, b$ and $c$. The red shaded regions are the new candidate regions for us to choose the points; then we can still find the points on the red line. The added points lie on the black lines.}
   \label{13}
\end{figure}

 \begin{figure}[htbp]
   \centering
   \includegraphics[scale=0.2]{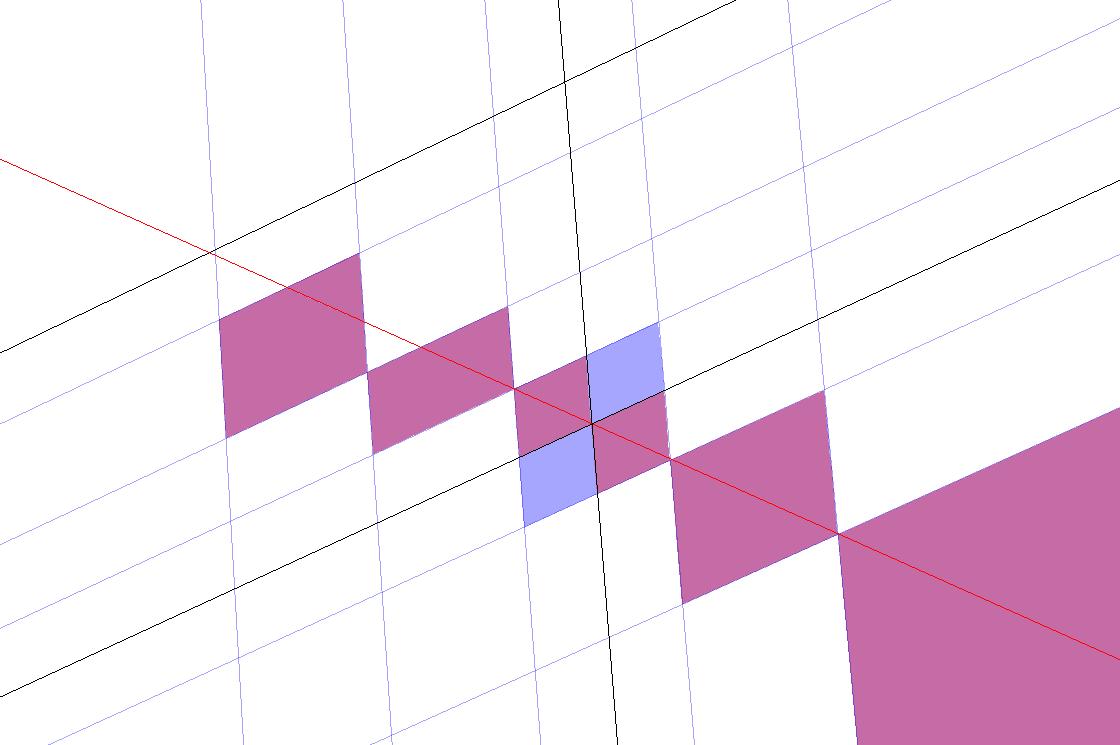}
   \caption{Difference of the shaded regions before and after we add a point in pocket $A$.}
   \label{14}
\end{figure}

\textbf{First round}: Refer to Figure \ref{13}, we add a point (say $a'$) lies between the third and fourth point on the red line in pocket $A$, connect $a'$ with vertex $a$ and vertex $b$ and extend to get two rays.Then the new extended ray $ba'$ will separate one of the earlier stripes in pocket $B$ into 2 smaller stripes, one of which has no point in it. We then add a point (say $b'$) lies between that stripe (having no point) on the red line in pocket $B$, connect $b'$ with vertex $c$. The line connecting $a$ and $a'$ will intersect with the red line in pocket $C$, so is the line connecting $b'$ and $c$. Since the positions of the connecting lines change continuously depends on the positions of the added points $a'$ and $b'$, if we adjust the positions of $a'$ and $b'$ carefully, we can locate both such that the intersection of $aa'$ and $bb'$ will meet at the red line in pocket $C$. Pick points on the red line in shaded regions of pocket $C$ accordingly will satisfy the properties. So in this manner, we add one point in each pocket, and if we compare it with Figure \ref{12}, the only difference is that the middle shaded region is divide into four smaller subregions, we are choosing two opposite regions as the new shaded regions, refer to Figure \ref{14}, note that the red line pass through the two diagonals.

\textbf{Further Rounds}:  The procedure is just similar to the first round. We add points in pocket $A$ first, by adjusting the connecting lines carefully, we can add $2, 2^2, 2^3, 2^4, 2^5, ...$ points in this graph with the property that the red line in pocket $C$ will always pass through all the regions, therefore we can always choose points on this line to make all of the points guarded. 
\end{quote}

Since we are adding infinitely many more points in the graph each round,  and all of these points must be guarded or a polygonalization with spike can exist, we get our main result about the Interior Universal Guards Problem.
}

\old{
\subsubsection{Relation with UGP}
The example we construct in the above section shows that no $<1$ constant ratio guards exist for IUGP, however, it does not mean no $<1$ constant ratio guards exist for UGP. Because note that in above example, the points we add to the triangle are on three straight lines, if we can put guards at the triangle vertices (convex hull vertices), then we do not have to guard all of the interior points, instead, putting guards on alternate point along the straight line where the points are added will suffice to guard all the polygonalization since every empty triangle will be in safe conditions. So in this case, it is enough to put guards on about half the points for UGP.
}

\old{
\subsubsection{Greedy Heuristic to get universal guards for the IUGP}

While, in the worst case, a point configuration may require guards at
essentially all of the points $S$, we investigate experimentally, with
a heuristic, how many interior guards are needed for randomly
generated points, uniform in a square.  Our heuristic starts with a
placement of guards at all interior points of $S$ for the IUGP or at
all points of $S$ for the UGP.  We then iteratively remove guards,
provided that they each obey the safe conditions with respect to other
pairs of unguarded points.  This results in a set that is minimal with
respect to the safe conditions; it need not be minimum, and it may
well be overly conservative, since it does not take into account
explicitly the possible polygonalizations of $S$.  We implemented the 
heuristic in Visual Studio 2013 (C++) and conducted experiments, 
on up to 1000 randomly generated points, on a Windows 8 PC with
 an Intel Core i7-4790 CPU 3.60GHz, 16.00 GB of RAM. 
 The results are shown in table \ref{aa}. For the IUGP case,
the fraction of points with guards is about 95\% for 100 random points
and 97\% for 1000 random points; for the UGP (unrestricted), the
fractions are about 83\% for 100 random points and 95\% for 1000
random points.

\begin{table}[htp]
\small
\centering
\caption{Numerical results for IUGP on 50, 100, 200, 500, 1000 random points}
\label{aa}
\begin{tabular}{|l|l|l|l|l|l|l|l|l|}
\hline\hline
\multicolumn{1}{|c|}{\multirow{2}{*}{Points \#}} & \multirow{2}{*}{Runs} & \multicolumn{3}{l|}{Ratio 1(IUGP)} & \multicolumn{3}{l|}{Ratio 2($UGP$)} & \multirow{2}{*}{\begin{tabular}[c]{@{}l@{}}Running time\\ (seconds)\end{tabular}} \\ \cline{3-8}
\multicolumn{1}{|c|}{}                            &                               & min        & max        & avg           & min         & max         & avg           &                                                                                   \\ \hline
50                                                & 500                           & 0.900          & 0.980         & 0.9492        & 0.660           & 0.820           & 0.7484         & 0.009                                                                                  \\ \hline
100                                               & 500                           & 0.910          & 0.980         & 0.9519            & 0.790           & 0.880           & 0.8319          & 0.128                                                                                  \\ \hline
200                                               & 500                           & 0.930          & 0.980         & 0.9574        & 0.865           & 0.910          & 0.8881         & 2.761                                                                                  \\ \hline
500                                               & 500                           & 0.954         & 0.976         & 0.9656        & 0.920          & 0.946          & 0.9327        & 146.491                                                                           \\ \hline
1000                                              & 25                             & 0.967         & 0.977         & 0.9716        & 0.951          & 0.957          & 0.9536                & 4553.79                                                                           \\ \hline\hline
\end{tabular}
\end{table}
}

\subsection{Full Grid Sets}

A natural special case arises when considering universal guards for a
set $S$ of points that are the $n=n_x\times n_y$ set of grid points (within a rectangle) on an integer lattice.  
For this case we achieve a tight worst-case bound.


\begin{theorem}\label{thm:gridugp}
	$\uniguardgrid{n}{} = \lfloor \frac{n}{2} \rfloor$, for
  rectangular grids of $n=n_x\times n_y$ grid points, with each of
  $n_x, n_y$ above a constant.
\end{theorem}

\begin{proof}
There are two parts to the proof: First, we must show that $\lfloor
\frac{n}{2} \rfloor$ guards suffice to guard a set $S$ of grid points
(that is sufficiently large).  Second, we must show necessity of
$\lfloor \frac{n}{2} \rfloor$ guards, arguing that fewer guards than
this will result in the lack of full coverage for some
polygonalizations of $S$.

The proof of sufficiency (that $\lfloor \frac{n}{2} \rfloor$ guards
suffice for universal guarding) is based on either of two different
patterns of guard selection: (1) place guards at the odd posititions
on odd-numbered rows and at even positions on even-numbered rows of
the grid (i.e., place guards in the grid according to white squares on
a checkboard); or (2) place guards at all positions on the
even-numbered rows.  Both (1) and (2) place $\left
\lfloor{\frac{n}{2}}\right \rfloor$ guards.  The two methods to place
guards are shown in Figure~\ref{twomethods}.  In order to show that
these placements yield universal guard sets (i.e., guard every
possible polygonalization of the input points), we argue that, for
either of the two placement strategies, any {\em empty grid triangle}
(i.e., a triangle whose vertices are grid points, with no other grid
points interior to the triangle or on its boundary) must have at least
one of its three vertices guarded.  This then implies that any
polygonalization $P$ of the grid points $S$ is guarded, since any such
simple polygon $P$ has a triangulation, whose triangles are empty grid
triangles, every one of which has a guard on at least one corner.
Since $P$ has a triangulation with nondegenerate triangles (having
nonempty interiors), we restrict ourselves to nondegenerate empty grid
triangles.

\begin{figure}[ht]
  \begin{center}
\includegraphics[width=0.7\textwidth]{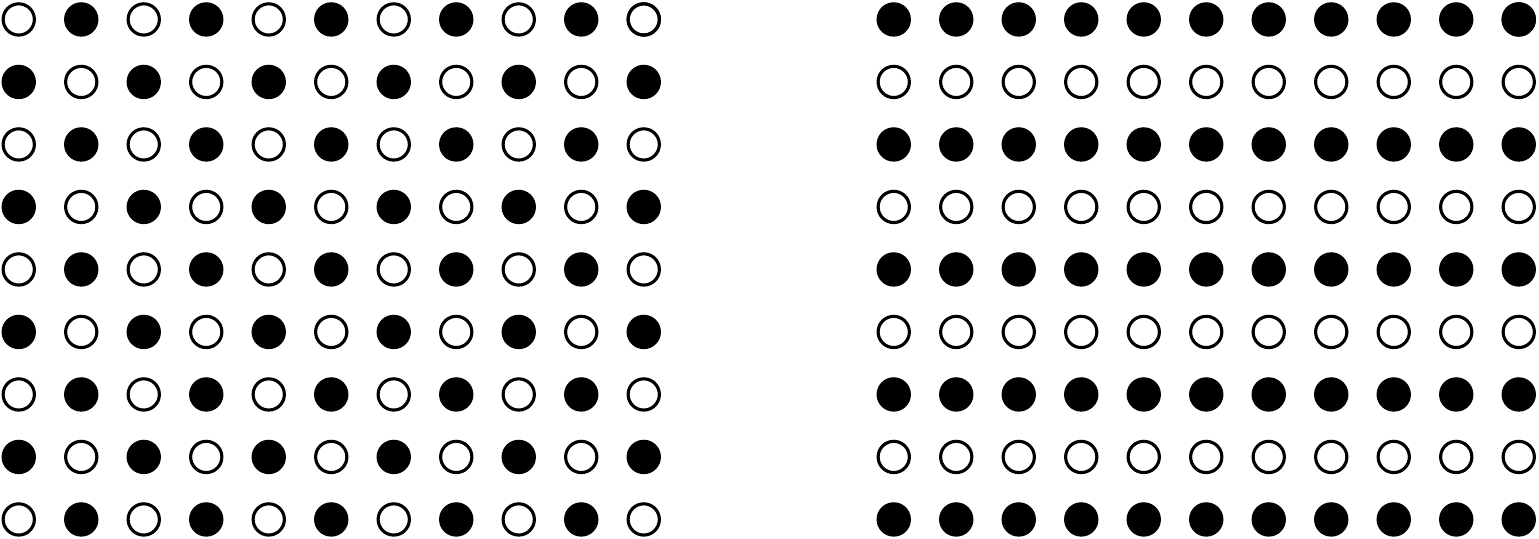}
  \end{center}
   \caption{Placement of guards at grid points according to pattern (1), left, and pattern (2), right.  
   Hollow dots denote unguarded points, and solid dots denote guarded points.}
  \label{twomethods}
\end{figure}

We give the argument for placement method (1); the argument is very
similar for placement method (2).  Consider a (nondegenerate) empty
grid triangle, $\Delta abc$, and assume, for contradiction, that all three
of its vertices $\{a,b,c\}$ are unguarded according to the
placement scheme (1).
\old{
We only need to consider each empty triangle, which means the triangles whose
vertices are all unguarded; all other triangles having at least one of
the three vertices guarded will definitely be guarded by those guards. }
\old{
Assume, without loss of generality, that $a$ has the smallest
$y$-coordinate among $a,b,c$, and that $b$ has the largest
$y$-coordinate.  It cannot be the case that the $y$-coordinate of $c$
is the same as either that of $a$ or of $b$, since segments $ac$ and
$bc$ must not pass through any grid points (other than their
endpoints), and adjacent to $a$ and to $b$, to the left/right on the
horizontal line through these points, the grid points are guarded (by
the specificaion of method (1)), and we have assumed that $a,b,c$ are
unguarded.}
Then, since $\Delta abc$ is an empty grid triangle, the parallelogram
defined by the pair of (integral) vectors $b-a$ and $c-a$ has no grid
points on its interior or on its boundary segments, other than at the
vertices $a$, $b$, $c$, and $b+(c-a)$, all of which are unguarded.
Since these parallelograms tile the plane, this implies that all grid
points are unguarded, a contradiction.

\old{
\paragraph{Case 1: The vertices $\{a,b,c\}$ lie on two distinct vertical lines.} 
Without loss of generality, $a$ lies on vertical line $l_1$ and $b$
and $c$ lie on vertical line $l_2\neq l_1$; see Figure \ref{17}.

Since $b$ and $c$ are unguarded, and they lie on same vertical line
$l_2$, according to placement method (1), we know there exists at
least one other grid point, a guarded point, between $b$ and $c$ on
the vertical segment $bc$; this implies that $\Delta abc$ is not an
empty grid triangle (since it has other grid points, including at
least one guarded grid point, on its boundary).

\begin{figure}[htbp]
\begin{minipage}[t]{0.33\textwidth}
     \centering
     \includegraphics[height=2cm]{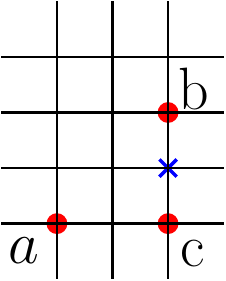}
    \caption{Case 1.}
     \label{17}
   \end{minipage}
   \begin{minipage}[t]{0.33\textwidth}
     \centering
     \includegraphics[height=2cm]{case_2(a)}
    \caption{Case 2(a).}
     \label{18}
   \end{minipage}
   \begin{minipage}[t]{0.32\textwidth}
     \centering
     \includegraphics[height=2cm]{case_2(b)}
    \caption{Case 2(b).}
     \label{19}
   \end{minipage}
\end{figure}

\paragraph{Case 2: The vertices $\{a,b,c\}$ lie on three distinct
  vertical lines.}  Suppose that $a$ is on vertical line $l_1$, $b$ is
on vertical line $l_2$, and $c$ is on vertical line $l_3$.  There are
then two subcases:

\textbf{(a).} Two of the vertices (say, $a$ and $b$) are on the same
horizontal line; see in Figure \ref{18}.  By placement method (1),
there must exist at least one (guarded) grid point between $ab$ on the
horizontal line through $ab$; this implies that $\Delta abc$ is not an
empty grid triangle.

\textbf{(b).} The three points $a,b,c$ are on three distinct
horizontal lines; see Figure \ref{19}.  Without loss of generality,
assume the vertices are in the order $a,b,c$, from left to right.
Since $\Delta abc$ is an empty grid triangle, we know that segment
$ab$ does not pass through any grid points $S$.  Let $l$ be the line
containing $ab$. (Note that $l$ may pass through other grid points of
$S$, just not any that lie interior to segment $ab$.) 

\old{We want to show that $\triangle abc$ cannot be empty, which means $\triangle
abc$ must have at least one guarded point either on its boundary edge or in its
interior. Refer to Figure \ref{20} for an example. We notice that if we connect $ab$,
it is easy to see that $ab$ lies on a line, say $l$, passing through only unguarded
points. Parallel to this line, there will be two adjacent lines (drawn in red)
passing through only guarded points, each is horizontally 1 unit distance away
from $l$. The only way for $\triangle abc$ to be empty is to put point $c$
somewhere in the region between the two red lines, however, no such point $c$
can exist since all possible unguarded points are at least horizontally $2$
unit distance away from $l$. Hence, every $\triangle abc$ has to have at least
one guard on one of the two red lines, thus, every $\triangle abc$ is guarded.}

\begin{figure}[htbp]
     \centering
     \includegraphics[height=3cm]{case_2(b)_further_explanation}
    \caption{Illustration of case 2(b).}
     \label{20}
\end{figure}

This implies that the $\left \lfloor{\frac{n}{2}}\right
\rfloor$ guards see every point in any polygonalization $P$ of $S$,
since any such $P$ can be triangulated, and every triangle in any
triangulation has at least one guard at a vertex.
}

%
The proof of necessity is based on examining local configurations of
unguarded grid points that force certain grid points to be guarded, in
order that every polygonalization is fully guarded.  In particular, we
observe that if a grid point $a\in S$ (that is not one of the 4
corners of the bounding rectangle of $S$) has both an unguarded
horizontal neighbor and an unguarded vertical neighbor, then a
polygonalization of $S$ that connects these three unguarded points in
order can result in an unseen triangular region, even if all other
points of $S$ are guarded.  See Figure~\ref{fig:spike} for an example,
showing locally a portion (three edges) of the polygonalization that
leaves an unguarded portion (shaded gray).  The full polygonalization
of each local configuration is shown by checking each of the cases, to
see that each can be fully polygonalized within a large enough
rectangular grid: a 4-by-5 point grid is sufficient to contain each
local configuration as part of a full polygonalization of the 4-by-5
grid.  Then, if the grid set $S$ is sufficiently large to contain a
4-by-5 subgrid, we claim that the grid points $S$ have a
polygonalization that would leave an unseen (shaded, triangular)
region, if three such grid points (point $a$ and one of its horizontal
and one of its vertical neighbors) are unguarded.  Refer to
Figure~\ref{fig:polygonalization} for an illustration of the
polygonalization of a 4-by-5 subgrid, and its extension to a
polygonalization to the full grid $S$.  Thus, for a sufficiently large
grid $S$, any 2-by-2 subgrid of $S$ (not in one of the 4 corners of
the bounding rectangle of $S$) must have at least two of its four grid
points guarded.

\begin{figure}[htbp]
     \centering
     \includegraphics[width=\textwidth]{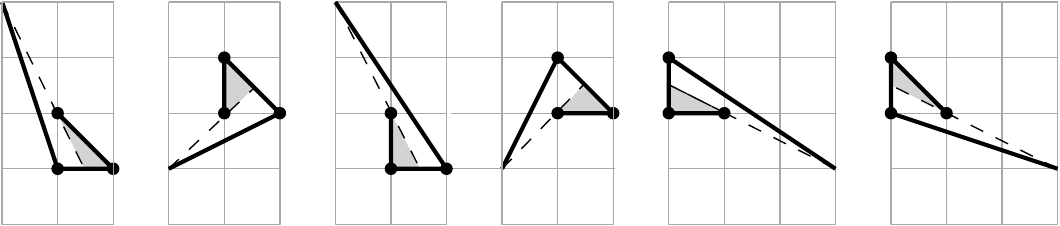}
    \caption{An unguarded grid point and its unguarded neighbors above
      and to the right of it can result in an unseen region (shaded)
      in a polygonalization.  Here, only 3 edges of the
      polygonalization are shown.}
       \label{fig:spike}
 \end{figure}

\begin{figure}[htbp]
     \centering
     \includegraphics[width=\textwidth]{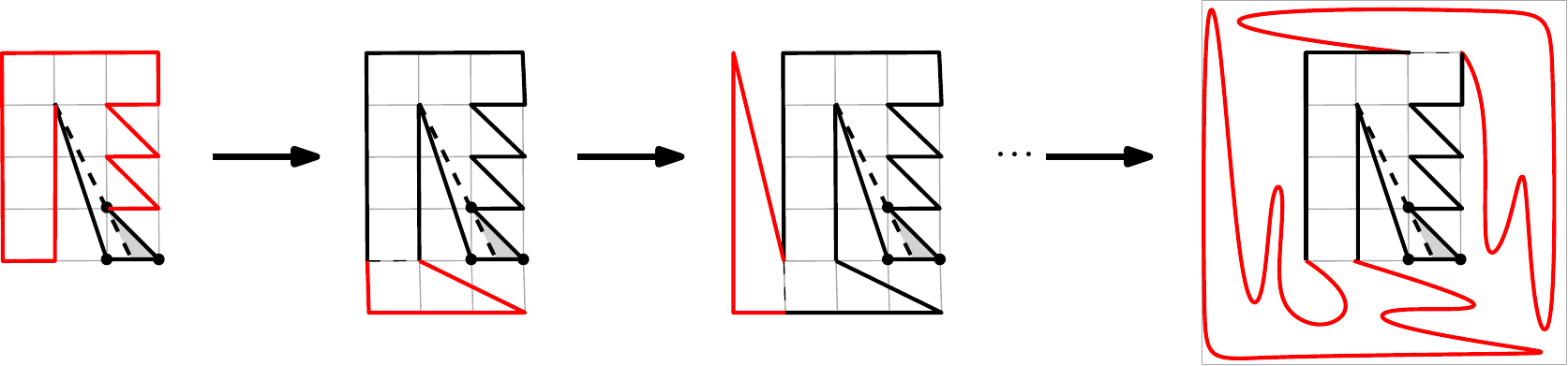}
    \caption{Demonstrating that a local configuration of three
      unguarded points, $a$ and a horizontal and a vertical neighbor
      of $a$, has a polygonalization that leaves an unseen region
      (shown shaded), even if all other grid points are guarded.
      First, we show a polygonalization within a 4-by-5 subgrid, then
      we illustrate how a larger containing grid $S$ admits a
      polygonalization.}
       \label{fig:polygonalization}
 \end{figure}
\end{proof}

\old{
First notice that if an unguarded grid point has both an unguarded
horizontal neighbour and an unguarded vertical neighbour, and it is
not one of the corner(NE, NW, SE, SW) points, then we can always
create a dark region by connecting these points differently.  Figure
\ref{23} shows all possible cases.

 
\begin{figure}[htbp]
     \centering
     \includegraphics[height=3.5cm]{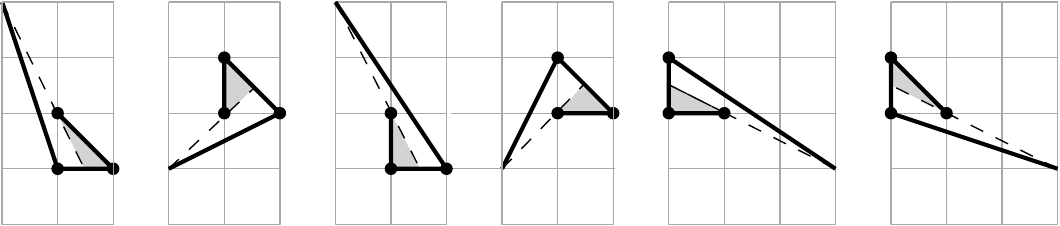}
    \caption{(a)-(f) shows six  different cases of connecting an unguarded point with its unguarded horizontal and vertical neighbours forming a dark region.}
       \label{23}
 \end{figure}
 
When the grid is big enough (at least the size of $5 \times 4$), we
can always fulfill the polygonalization on different cases. This can
be done easily on a $5 \times 4$ grid, and we can use the following
lemma to show the polygonalization exist for larger grid.

\begin{corollary}
If a grid graph $G$ has a sub grid graph $G'$ which can be drawn in a
polygonalization that has a dark region, then $G$ can also be drawn in
some polygonalization with the same dark region exists.
\end{corollary}	

Proof sketch: Notice that there must exist two adjacent (upper/lower
horizontal, left/right vertical) boundary nodes that are connected in
the polygonization on $G'$, say points $p_1, p_2$, so when we add a
new row (or a new column), we can just connect $p_1$ to the closest
endpoint of the new added row, and draw a line segment connecting the
new row, then connect the other endpoint to $p_2$, which will create a
new polygonization with the original spike still exists. Figure
\ref{23detail} shows this idea.

\begin{figure}[htbp]
     \centering
     \includegraphics[height=3.5cm]{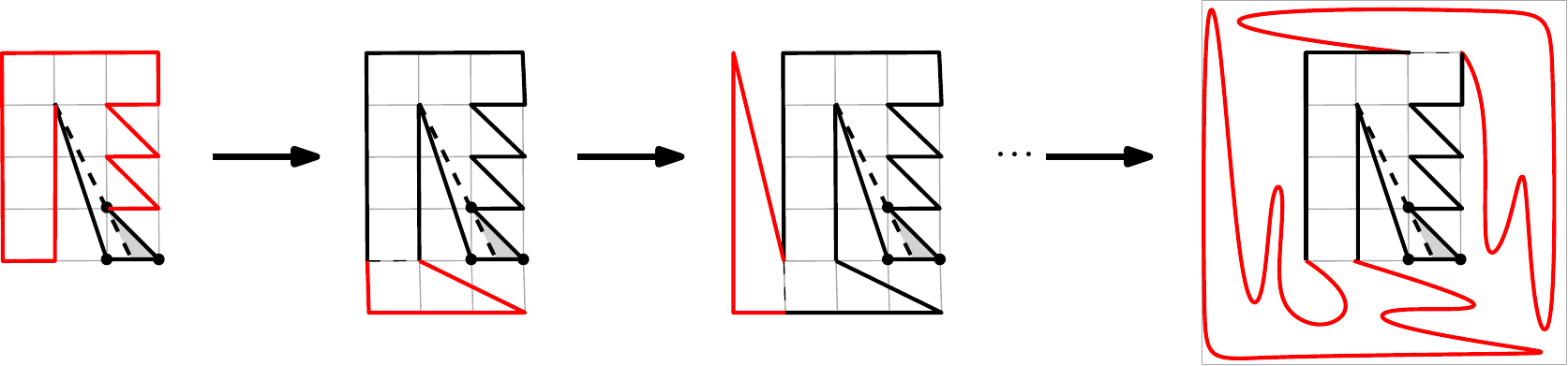}
    \caption{shows a polygonalization with dark region exists for
      first case in Figure \ref{23}, and how to expand the
      polygonalization in larger graph. Red lines shows the adding
      procedure.}
       \label{23detail}
 \end{figure}

Secondly, if two unguarded points are adjacent, w.\,l.\,o.\,g., say on
a vertical line, then there cannot exist any unguarded points on its
adjacent vertical lines.

We can use case analysis based on the two properties mentioned above to show our main result on grid case.
}

\old{
First notice that if an unguarded grid point has both an unguarded horizontal neighbour and an unguarded vertical neighbour, and it is not one of the corner(NE, NW, SE, SW) points, then it can form a spike with those unguarded neighbours, refer to Figure \ref{23} (a)-(f). When the grid is big enough, that spike can exist in some polygonalization.

\begin{figure}[htbp]
\begin{minipage}[t]{0.32\textwidth}
     \centering
     \includegraphics[scale=0.4]{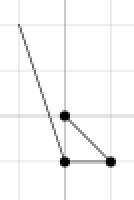}
   \end{minipage}
   \begin{minipage}[t]{0.32\textwidth}
     \centering
     \includegraphics[scale=0.4]{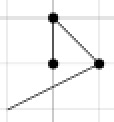}
   \end{minipage}
   \begin{minipage}[t]{0.32\textwidth}
     \centering
     \includegraphics[scale=0.4]{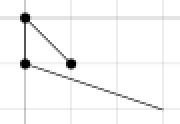}
   \end{minipage}
 \begin{minipage}[t]{0.32\textwidth}
     \centering
     \includegraphics[scale=0.4]{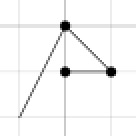}
   \end{minipage}
   \begin{minipage}[t]{0.32\textwidth}
     \centering
     \includegraphics[scale=0.4]{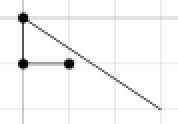}
   \end{minipage}
   \begin{minipage}[t]{0.32\textwidth}
     \centering
     \includegraphics[scale=0.4]{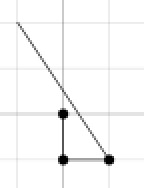}
   \end{minipage}  
    \caption{(a)-(f) shows six  different cases of an unguarded point forming a spike with its unguarded horizontal and vertical neighbours.}
       \label{23}
 \end{figure}

 Secondly, if two unguarded points are adjacent, w.\,l.\,o.\,g., say on a vertical line, then there cannot exist any unguarded points on its adjacent vertical lines. Otherwise a spike can form.

 So in the optimal solution of finding minimized number of guards to guard all the possible polygonalizations, an unguarded point cannot have both of its horizontal and vertical neighbours unguarded at the same time.  Start with an unguarded grid point, there are several possibilities for what the optimal solution can look like.

\begin{quote}
Case 1: Both of its vertical neighbours are unguarded, then its horizontal neighbours must be guarded. Refer to Figure \ref{24}(a).

 In this case, in order to put minimized number of guards, the best way is to use the second method we mentioned above.
\end{quote}

\begin{quote}
Case 2: Only one of its vertical neighbours is unguarded, its horizontal neighbours must be guarded. Refer to Figure \ref{24}(b). 

 In this case, we need more guards than Case 1, so it cannot be optimal.
\end{quote}

\begin{quote}
Case 3: All of its neighbours are guarded. Refer to Figure \ref{24}(c). This case is the first method we mentioned above.
\end{quote}

\begin{quote}
Case 4: It is a corner point. Refer to Figure \ref{24}(d). It cannot use less guards than either Case 1 or Case 3.
\end{quote}

Therefore, we need at least $\left \lfloor{\frac{a}{2}}\right \rfloor \cdot b$ guards to guard every polygonalization on $a\times b$ grid points.

\begin{figure}[htbp]
  \begin{minipage}[t]{0.21\textwidth}
     \centering
     \includegraphics[scale=0.4]{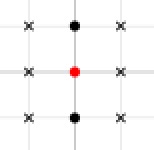}
   \end{minipage}
   \begin{minipage}[t]{0.22\textwidth}
     \centering
     \includegraphics[scale=0.4]{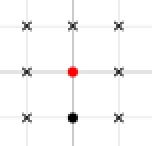}
   \end{minipage}
   \begin{minipage}[t]{0.25\textwidth}
     \centering
     \includegraphics[scale=0.4]{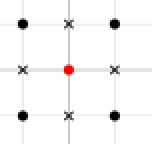}
   \end{minipage}
  \begin{minipage}[t]{0.25\textwidth}
     \centering
     \includegraphics[scale=0.4]{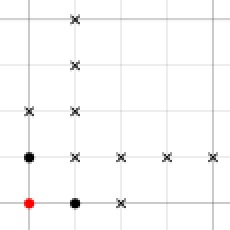}
   \end{minipage}
    \caption{(a)-(d) Four different cases for an optimal solution. Disks denote unguarded points, crosses denote guarded points, a red disk denotes the unguarded point to start with.}
       \label{24}
 \end{figure}

}

\old{
\begin{corollary}
If a grid graph $G$ has a sub grid graph $G'$ which can be drawn in a polygonalization that has a spike, then $G$ can also be drawn in some polygonalization with a spike exists.
\end{corollary}	

}

\old{
\subsection{Interior Guards}

In the Interior Universal Guards Problem (UGPI) we allow guards to be
placed only at points of $S$ that are not convex hull vertices of $S$.
For this case, we obtain an asymptotically tight bound on the number
of universal guards:

\begin{theorem}\label{thm:iugp}
	$\uniguardinterior{n}{} = n - \Theta (1)$
\end{theorem}

\subsection{Full Grid Sets}

A natural special case arises when considering universal guards for a
full set of $n=a\times b$ grid points on an integer lattice.  We are
also able in this case to achieve a tight worst-case bound:

\begin{theorem}\label{thm:gridugp}
	$\uniguardgrid{n}{} = \lfloor \frac{n}{2} \rfloor$.
\end{theorem}
}

%% file: figs/main.pdf_t
\begin{picture}(0,0)%
\includegraphics{figs/main.pdf}%
\end{picture}%
\setlength{\unitlength}{789sp}%
\begingroup\makeatletter\ifx\SetFigFont\undefined%
\gdef\SetFigFont#1#2#3#4#5{%
  \reset@font\fontsize{#1}{#2pt}%
  \fontfamily{#3}\fontseries{#4}\fontshape{#5}%
  \selectfont}%
\fi\endgroup%
\begin{picture}(28612,11468)(476,-10641)
\put(7179,468){\makebox(0,0)[lb]{\smash{{\SetFigFont{5}{6.0}{\familydefault}{\mddefault}{\updefault}{\color[rgb]{0,0,0}$a$}%
}}}}
\put(5170,-5588){\makebox(0,0)[lb]{\smash{{\SetFigFont{5}{6.0}{\familydefault}{\mddefault}{\updefault}{\color[rgb]{0,0,0}$p_1$}%
}}}}
\put(6408,-3096){\makebox(0,0)[lb]{\smash{{\SetFigFont{5}{6.0}{\familydefault}{\mddefault}{\updefault}{\color[rgb]{0,0,0}$p_k$}%
}}}}
\put(7096,-8794){\makebox(0,0)[lb]{\smash{{\SetFigFont{5}{6.0}{\familydefault}{\mddefault}{\updefault}{\color[rgb]{0,0,0}$q_1$}%
}}}}
\put(4385,-8312){\makebox(0,0)[lb]{\smash{{\SetFigFont{5}{6.0}{\familydefault}{\mddefault}{\updefault}{\color[rgb]{0,0,0}$q_k$}%
}}}}
\put(9175,-5285){\makebox(0,0)[lb]{\smash{{\SetFigFont{5}{6.0}{\familydefault}{\mddefault}{\updefault}{\color[rgb]{0,0,0}$r_1$}%
}}}}
\put(9877,-7569){\makebox(0,0)[lb]{\smash{{\SetFigFont{5}{6.0}{\familydefault}{\mddefault}{\updefault}{\color[rgb]{0,0,0}$r_k$}%
}}}}
\put(23136,524){\makebox(0,0)[lb]{\smash{{\SetFigFont{5}{6.0}{\familydefault}{\mddefault}{\updefault}{\color[rgb]{0,0,0}$a$}%
}}}}
\put(16426,-9992){\makebox(0,0)[lb]{\smash{{\SetFigFont{5}{6.0}{\familydefault}{\mddefault}{\updefault}{\color[rgb]{0,0,0}$b$}%
}}}}
\put(29073,-9827){\makebox(0,0)[lb]{\smash{{\SetFigFont{5}{6.0}{\familydefault}{\mddefault}{\updefault}{\color[rgb]{0,0,0}$c$}%
}}}}
\put(21697,-291){\makebox(0,0)[lb]{\smash{{\SetFigFont{5}{6.0}{\familydefault}{\mddefault}{\updefault}{\color[rgb]{0,0,0}$a_1$}%
}}}}
\put(23444,-222){\makebox(0,0)[lb]{\smash{{\SetFigFont{5}{6.0}{\familydefault}{\mddefault}{\updefault}{\color[rgb]{0,0,0}$a_2$}%
}}}}
\put(25812,-7584){\makebox(0,0)[lb]{\smash{{\SetFigFont{5}{6.0}{\familydefault}{\mddefault}{\updefault}{\color[rgb]{0,0,0}$r_k$}%
}}}}
\put(25110,-5300){\makebox(0,0)[lb]{\smash{{\SetFigFont{5}{6.0}{\familydefault}{\mddefault}{\updefault}{\color[rgb]{0,0,0}$r_1$}%
}}}}
\put(23031,-8809){\makebox(0,0)[lb]{\smash{{\SetFigFont{5}{6.0}{\familydefault}{\mddefault}{\updefault}{\color[rgb]{0,0,0}$q_1$}%
}}}}
\put(22343,-3111){\makebox(0,0)[lb]{\smash{{\SetFigFont{5}{6.0}{\familydefault}{\mddefault}{\updefault}{\color[rgb]{0,0,0}$p_k$}%
}}}}
\put(21105,-5603){\makebox(0,0)[lb]{\smash{{\SetFigFont{5}{6.0}{\familydefault}{\mddefault}{\updefault}{\color[rgb]{0,0,0}$p_1$}%
}}}}
\put(20320,-8327){\makebox(0,0)[lb]{\smash{{\SetFigFont{5}{6.0}{\familydefault}{\mddefault}{\updefault}{\color[rgb]{0,0,0}$q_k$}%
}}}}
\put(491,-9977){\makebox(0,0)[lb]{\smash{{\SetFigFont{5}{6.0}{\familydefault}{\mddefault}{\updefault}{\color[rgb]{0,0,0}$b$}%
}}}}
\put(13138,-9812){\makebox(0,0)[lb]{\smash{{\SetFigFont{5}{6.0}{\familydefault}{\mddefault}{\updefault}{\color[rgb]{0,0,0}$c$}%
}}}}
\put(5762,-276){\makebox(0,0)[lb]{\smash{{\SetFigFont{5}{6.0}{\familydefault}{\mddefault}{\updefault}{\color[rgb]{0,0,0}$a_1$}%
}}}}
\put(7509,-207){\makebox(0,0)[lb]{\smash{{\SetFigFont{5}{6.0}{\familydefault}{\mddefault}{\updefault}{\color[rgb]{0,0,0}$a_2$}%
}}}}
\end{picture}%

%% file: figs/main-zoom.pdf_t
\begin{picture}(0,0)%
\includegraphics{figs/main-zoom.pdf}%
\end{picture}%
\setlength{\unitlength}{1579sp}%
\begingroup\makeatletter\ifx\SetFigFont\undefined%
\gdef\SetFigFont#1#2#3#4#5{%
  \reset@font\fontsize{#1}{#2pt}%
  \fontfamily{#3}\fontseries{#4}\fontshape{#5}%
  \selectfont}%
\fi\endgroup%
\begin{picture}(14074,5805)(4267,-8815)
\put(4351,-8311){\makebox(0,0)[lb]{\smash{{\SetFigFont{10}{12.0}{\familydefault}{\mddefault}{\updefault}$q_1$}}}}
\put(5813,-8679){\makebox(0,0)[lb]{\smash{{\SetFigFont{10}{12.0}{\familydefault}{\mddefault}{\updefault}$q_2$}}}}
\put(9883,-7175){\makebox(0,0)[lb]{\smash{{\SetFigFont{10}{12.0}{\familydefault}{\mddefault}{\updefault}$r_1$}}}}
\put(9255,-5767){\makebox(0,0)[lb]{\smash{{\SetFigFont{10}{12.0}{\familydefault}{\mddefault}{\updefault}$r_2$}}}}
\put(13763,-8679){\makebox(0,0)[lb]{\smash{{\SetFigFont{10}{12.0}{\familydefault}{\mddefault}{\updefault}$q_2$}}}}
\put(17833,-7175){\makebox(0,0)[lb]{\smash{{\SetFigFont{10}{12.0}{\familydefault}{\mddefault}{\updefault}$r_1$}}}}
\put(17205,-5767){\makebox(0,0)[lb]{\smash{{\SetFigFont{10}{12.0}{\familydefault}{\mddefault}{\updefault}$r_2$}}}}
\put(13174,-8503){\makebox(0,0)[lb]{\smash{{\SetFigFont{10}{12.0}{\familydefault}{\mddefault}{\updefault}$q$}}}}
\put(17588,-6631){\makebox(0,0)[lb]{\smash{{\SetFigFont{10}{12.0}{\familydefault}{\mddefault}{\updefault}$r$}}}}
\put(12451,-8311){\makebox(0,0)[lb]{\smash{{\SetFigFont{10}{12.0}{\familydefault}{\mddefault}{\updefault}$q_1$}}}}
\put(13201,-4561){\makebox(0,0)[lb]{\smash{{\SetFigFont{10}{12.0}{\familydefault}{\mddefault}{\updefault}$p_2$}}}}
\put(13726,-3961){\makebox(0,0)[lb]{\smash{{\SetFigFont{10}{12.0}{\familydefault}{\mddefault}{\updefault}$p$}}}}
\put(14026,-3436){\makebox(0,0)[lb]{\smash{{\SetFigFont{10}{12.0}{\familydefault}{\mddefault}{\updefault}$p_1$}}}}
\put(5176,-4636){\makebox(0,0)[lb]{\smash{{\SetFigFont{10}{12.0}{\familydefault}{\mddefault}{\updefault}$p_2$}}}}
\put(6726,-3747){\makebox(0,0)[lb]{\smash{{\SetFigFont{10}{12.0}{\familydefault}{\mddefault}{\updefault}$p_1$}}}}
\end{picture}%

%% file: figs/poly.pdf_t
\begin{picture}(0,0)%
\includegraphics{figs/poly.pdf}%
\end{picture}%
\setlength{\unitlength}{868sp}%
\begingroup\makeatletter\ifx\SetFigFont\undefined%
\gdef\SetFigFont#1#2#3#4#5{%
  \reset@font\fontsize{#1}{#2pt}%
  \fontfamily{#3}\fontseries{#4}\fontshape{#5}%
  \selectfont}%
\fi\endgroup%
\begin{picture}(26699,10552)(571,-9755)
\put(27255,-9477){\makebox(0,0)[lb]{\smash{{\SetFigFont{5}{6.0}{\familydefault}{\mddefault}{\updefault}$c$}}}}
\put(586,-9632){\makebox(0,0)[lb]{\smash{{\SetFigFont{5}{6.0}{\familydefault}{\mddefault}{\updefault}$b$}}}}
\put(13248,-9524){\makebox(0,0)[lb]{\smash{{\SetFigFont{5}{6.0}{\familydefault}{\mddefault}{\updefault}$c$}}}}
\put(7210,494){\makebox(0,0)[lb]{\smash{{\SetFigFont{5}{6.0}{\familydefault}{\mddefault}{\updefault}$a$}}}}
\put(5887,-313){\makebox(0,0)[lb]{\smash{{\SetFigFont{5}{6.0}{\familydefault}{\mddefault}{\updefault}$a_1$}}}}
\put(21317,402){\makebox(0,0)[lb]{\smash{{\SetFigFont{5}{6.0}{\familydefault}{\mddefault}{\updefault}$a$}}}}
\put(14593,-9585){\makebox(0,0)[lb]{\smash{{\SetFigFont{5}{6.0}{\familydefault}{\mddefault}{\updefault}$b$}}}}
\put(19465,-3530){\makebox(0,0)[lb]{\smash{{\SetFigFont{5}{6.0}{\familydefault}{\mddefault}{\updefault}$p_i$}}}}
\end{picture}%

%% file: 06-conclusion.tex
\section{Conclusion}


There are many open problems that are interesting challenges for future work. In particular, 
can the upper bound approaches for $\kuniguard{n}{k}$ and $\kuniguardholes{n}{k}$ be improved by making use of the number of shells?
	Can the general approach of Theorem~\ref{thm:upperBoundskugp} be improved?
	What about lower bounds for $k$-UGP for $k \geq 7$?

The quest for better bounds is also closely related to other combinatorial challenges.
Is an instance of the $2$-UGP $5$-colorable? If so, our results give a first trivial upper bound of $\frac{3}{5}n$ for the $2$-UGP, which would be of
independent interest. 
Is the bound of $\frac{1}{2}n$ for the intersection-free $k$-UGP tight?
Further questions consider the setting in which each vertex $v$ has a bounded candidate set of vertices that may be adjacent to $v$.
Other variants arise when the ratio of the lengths of the edges of the considered polygons is upper- and lower-bounded by given constants.
	It may also be interesting to explore possible relations between universal guard problems and universal graphs.